\documentclass[11pt]{article}

\usepackage[utf8]{inputenc}

\def\martin#1{\marginpar{$\leftarrow$\fbox{Ma}}\footnote{$\Rightarrow$~{\sf\textcolor{ForestGreen}{#1 --Martin}}}}  
\def\sayan#1{\marginpar{$\leftarrow$\fbox{Sa}}\footnote{$\Rightarrow$~{\sf\textcolor{red}{#1 --Sayan}}}}

\usepackage[ruled,vlined,linesnumbered,noend]{algorithm2e}
\SetKwInOut{Parameter}{Parameters}

\usepackage{amsthm}
\usepackage{amsmath}
\usepackage{amssymb}
\usepackage{enumerate}
\usepackage{graphicx}
\usepackage[margin=1in]{geometry}
\usepackage{dsfont}
\usepackage{booktabs} 
\usepackage{tabularx}
\usepackage{multirow}
\usepackage{enumitem}
\usepackage{tablefootnote}

\usepackage[dvipsnames]{xcolor}
\definecolor{ForestGreen}{rgb}{0.1333,0.5451,0.1333}
\definecolor{DarkRed}{rgb}{0.65,0,0}

\usepackage{hyperref}
\hypersetup{
    colorlinks=true,
    linkcolor=DarkRed,
    filecolor=magenta,      
    urlcolor=cyan,
    citecolor=ForestGreen,
}

\usepackage{cleveref}

\usepackage{framed}
\usepackage[framemethod=tikz]{mdframed}
\usepackage{tikz-cd}

\newenvironment{wrapper}[1]
{
	\begin{center}
		\begin{minipage}{\linewidth}
			\begin{mdframed}[hidealllines=true, backgroundcolor=gray!20, leftmargin=0cm,innerleftmargin=0.4cm,innerrightmargin=0.4cm,innertopmargin=0.4cm,innerbottommargin=0.4cm,roundcorner=10pt]
				#1}
			{\end{mdframed}
		\end{minipage}
	\end{center}
}

\DeclareMathAlphabet{\mathmybb}{U}{bbold}{m}{n}
\newcommand{\1}{\mathmybb{1}}

\newtheorem{theorem}{Theorem}[section]
\newtheorem{lemma}[theorem]{Lemma}

\newtheorem{corollary}[theorem]{Corollary}

\newtheorem{definition}[theorem]{Definition}

\newtheorem{observation}[theorem]{Observation}
\newtheorem{claim}[theorem]{Claim}

\newcommand{\Afl}{\textsc{ALG}^{\textnormal{fl}}}

\newcommand{\Opt}{\textsc{OPT}}
\newcommand{\optf}{\textsc{OPT}^{\textnormal{fl}}}
\newcommand{\optc}{\textsc{OPT}^{\textnormal{cl}}}
\newcommand{\foptf}{\widehat{\textsc{OPT}}{}^{\textnormal{fl}}}
\newcommand{\foptc}{\widehat{\textsc{OPT}}{}^{\textnormal{cl}}}
\newcommand{\fl}{\textnormal{fl}}
\newcommand{\cl}{\textnormal{cl}}
\newcommand{\Proj}{\textnormal{\texttt{Proj}}}

\newcommand{\RadiiMP}{\textnormal{\textsc{RadiiMP}}}
\newcommand{\FracLMP}{\textnormal{\textsc{FracLMP}}}
\newcommand{\FracKMed}{\textnormal{\textsc{FracKMed}}}

\newcommand{\Sparsifier}{\textnormal{\textsc{Sparsifier}}}

\newcommand{\RandLoc}{\textnormal{\textsc{RandLocalSearch}}}

\title{Fully Dynamic $k$-Clustering \\
with Fast Update Time and Small Recourse
}

\author{
Sayan Bhattacharya 
\\University of Warwick 
\\\texttt{s.bhattacharya@warwick.ac.uk}  
\and 
Mart\'{i}n Costa 
\\University of Warwick
\\\texttt{martin.costa@warwick.ac.uk}  
\and 
Naveen Garg 
\\IIT Delhi 
\\\texttt{naveen@cse.iitd.ac.in} 
\and Silvio Lattanzi 
\\Google Research 
\\\texttt{silviol@google.com}  
\and Nikos Parotsidis 
\\Google Research 
\\\texttt{nikosp@google.com}
}

\date{}

\begin{document}

\maketitle

\pagenumbering{gobble}

\begin{abstract}
In the dynamic metric $k$-median problem, we wish to maintain a set of $k$ centers $S \subseteq V$ in an input metric space $(V, d)$ that gets updated via point insertions/deletions, so as to minimize the objective $\sum_{x \in V} \min_{y \in S} d(x, y)$. The quality of a dynamic algorithm is measured in terms of its approximation ratio, ``recourse'' (the number of changes in $S$ per update) and ``update time'' (the time it takes to handle an update). The ultimate goal in this line of research is to obtain a dynamic $O(1)$ approximation algorithm with  $\tilde{O}(1)$ recourse and $\tilde{O}(k)$ update time.

Dynamic $k$-median is a canonical example of a class of problems known as dynamic $k$-clustering, that has received significant attention in recent years [Fichtenberger et al, SODA'21], [Bateni et al, SODA'23], [Lacki et al, SODA'24]. To the best of our knowledge, however, all these previous papers either attempt to minimize the algorithm's recourse while ignoring its update time, or minimize the algorithm's update time while ignoring its recourse. For dynamic $k$-median in particular, the state-of-the-art results get $\tilde{O}(k^2)$ update time and $O(k)$ recourse~[Cohen-Addad et al, ICML'19], [Henzinger and Kale, ESA'20], [Bhattacharya et al, NeurIPS'23]. But, this recourse bound of $O(k)$ can be trivially obtained by recomputing an optimal solution from scratch after every update,  provided we ignore the update time. In addition, the update time of $\tilde{O}(k^2)$ is polynomially far away from the desired bound of $\tilde{O}(k)$. We come {\em arbitrarily close} to resolving the main open question on this topic, with the following results.

\medskip
\noindent 
(I) We develop a new framework of {\em randomized local search} that is suitable for adaptation in a dynamic setting. For every $\epsilon > 0$, this gives us a dynamic $k$-median algorithm with $O(1/\epsilon)$ approximation ratio, $\tilde{O}(k^{\epsilon})$ recourse and $\tilde{O}(k^{1+\epsilon})$ update time. This framework also generalizes to dynamic $k$-clustering with $\ell^p$-norm objectives. As a corollary, we obtain similar bounds for the dynamic $k$-means problem, and a new trade-off between approximation ratio, recourse and update time for the dynamic $k$-center problem.

\medskip
\noindent
(II) If it suffices to maintain only an estimate of the {\em value} of the optimal $k$-median objective, then we obtain a $O(1)$ approximation algorithm with $\tilde{O}(k)$ update time. We achieve this result via adapting the Lagrangian Relaxation framework of [Jain and Vazirani, JACM'01], and a facility location algorithm of [Mettu and Plaxton, FOCS'00] in the dynamic setting.

\noindent
\end{abstract}

\newpage

\tableofcontents

\newpage


\part{Extended Abstract}
\label{part:0}

\pagenumbering{arabic}

\section{Introduction}
Consider a general metric space $(V, w, d)$ on $|V| = n$ points, with a distance function $d : V \times V \rightarrow \mathbb{R}^+$ satisfying triangle inequality, a weight function $w : V \rightarrow \mathbb{R}^+$, and a positive integer $k \in [n]$.  In the ``metric $k$-median'' problem, we are asked to open a subset $S \subseteq V$ of  $k$ ``centers'', and assign each point $x \in V$ to its nearest center in $S$. 
The objective is to compute  an $S$ so as to minimize its ``clustering cost'' $\cl(S, V) := \sum_{x \in V} w(x) \cdot d(x, S)$; where $d(x, S) := \min_{y \in S} d(x, y)$. Throughout the paper, we make the standard assumption that we have access to a ``distance oracle'' which, in $O(1)$ time, returns the value  $d(x, y)$ for any pair of points $x, y \in V$.

As a  fundamental, textbook clustering problem with wide-ranging applications, metric $k$-median has been extensively studied over several decades~\cite{WilliamsonShmoys2011}. The problem happens to be NP-hard, and has  served as a testbed for powerful techniques such as local search~\cite{AryaGKMMP04} and Lagrangian relaxation~\cite{JainV01} in the field of approximation algorithms.
Due to a celebrated result by Mettu and Plaxton~\cite{MettuP02}, it admits a $O(1)$ approximation  in $\tilde{O}(kn)$ time,\footnote{Throughout, the notation $\tilde{O}(\cdot)$ hides terms that are polylogarithmic in $n$ and the aspect ratio of the metric space.} and this runtime bound is is known to be tight upto polylogarithmic factors~\cite{BadoiuCIS05}. 


\medskip
\noindent {\bf The Dynamic Setting.}
Now, suppose that the input keeps changing by means of ``updates''.  Each update inserts/deletes a point $x$ in $V$ (in case of an insertion, it also specifies the weight $w(x)$), and we have to maintain a clustering $S \subseteq V$ in the current input. The time spent by the algorithm per update is  its ``update time'', whereas the number of changes (i.e., insertions/deletions of points) in $S$ per update is  its ``recourse''. Intuitively, the ``update time''  of an algorithm measures how efficient it is in the presence of dynamic inputs, whereas its ``recourse'' measures the {\em stability} of the maintained clustering over time.  This leads us naturally  to the following question.

\begin{wrapper}
\textbf{Q1:} Is there an algorithm for the dynamic $k$-median problem that {\bf simultaneously} achieves  $O(1)$ approximation ratio, $\tilde{O}(1)$ recourse and $\tilde{O}(k)$ update time?
\end{wrapper}

\noindent
Note that if there is a dynamic $O(1)$ approximation algorithm for this problem with update time $T$, then we also get a {\em static} $O(1)$ approximation algorithm that runs in $O(T n)$ time.\footnote{Simply insert the $n$ points in $V$,  let the dynamic algorithm deal with these insertions, and return its final output.} Thus, we cannot expect to have $T = o(k)$, for otherwise it will imply a static algorithm with runtime $o(kn)$, contradicting the results of~\cite{MettuP02,BadoiuCIS05} as discussed above. In this sense, an affirmative answer to {\bf Q1} will (nearly) settle the complexity of the dynamic $k$-median problem.

\medskip
\noindent {\bf Metric $k$-Clustering.} We use the term ``metric $k$-clustering''  to refer to any problem where we have to compute a set $S \subseteq V$ of $k$ centers, so as to minimize a given objective. Besides $k$-median, two other canonical problems  under this category are: (i) ``metric $k$-center'', where we wish to minimize $
\max_{x \in V} d(x, S)$, and (ii) ``metric $k$-means'', where we wish to minimize $\sum_{x \in V} w(x) \cdot \left( d(x, S)\right)^{2}$. In the dynamic setting,  a question analogous to {\bf Q1} naturally arises for each of these problems.

\subsection{Our Contributions} 

In recent years, a long and influential line of work has been devoted to the study of dynamic $k$-clustering~\cite{ChanGS18,icml/LattanziV17,soda/FichtenbergerLN21,LHRJ23,nips/Cohen-AddadHPSS19,esa/HenzingerK20,nips/BiabaniHM023,ourneurips2023,BateniEFHJMW23}. To the best of our knowledge, however, all the existing papers focus on optimizing either: (a) exclusively the algorithm's update time (while ignoring its recourse) or (b) exclusively the algorithm's recourse (while ignoring its update time). {\bf Q1}, therefore, remains an outstanding open question on this topic.  Indeed, if we take {\em any} state-of-the-art $O(1)$ approximation algorithm for dynamic $k$-median/$k$-means/$k$-center, then currently there is a polynomial in $k$ gap between:  {\em either} its recourse {\em or}  update time, and the corresponding bound that {\bf Q1} asks for. Furthermore, for dynamic $k$-median and $k$-means, currently we do not know of  \emph{any algorithm with better than $O(k)$ recourse}, which is trivial to achieve by simply recomputing a new optimal solution from scratch after every update.


Using  {\em a single algorithm}, we come {\em arbitrarily close} to bridging these polynomial gaps for both $k$-median and $k$-means (see Table~\ref{tab:kmedian}) and derive a new trade-off between approximation ratio, recourse and update time for $k$-center (see Table~\ref{tab:kcenter}). In addition, if it suffices to maintain an estimate of the optimal objective {\em value}, then via a different approach we completely attain the approximation ratio and update time bound stated in {\bf Q1} for $k$-median (see Theorem~\ref{thm:value:main:2}).

\begin{table}[h!]
  \centering
  \begin{tabular}{cccc}
    \toprule
    Approximation Ratio & Update Time & Recourse & Paper \\
    \midrule
     $O(1)$ & $\tilde O(n + k^2)$ & $O(k)$ & \cite{nips/Cohen-AddadHPSS19} \\
     $O(1)$ & $\tilde O(k^2)$ & $O(k)$ & \cite{esa/HenzingerK20} \\
     $O(1)$ & $\tilde O(k^2)$ & $O(k)$ & \cite{ourneurips2023}\tablefootnote{We remark that \cite{ourneurips2023} actually maintain a $O(1)$ approximate ``sparsifier'' of size $\tilde{O}(k)$ in $\tilde{O}(k)$ update time (see Section~\ref{sec:prelim:new}), and they need to run the static algorithm of~\cite{MettuP00} on top of this sparsifier after every update. This leads to an update time of $\tilde{O}(k^2)$ for the dynamic $k$-median and $k$-means problems.}\\
     $O(1/\epsilon)$ & $\tilde O(k^{1 + \epsilon})$ & $\tilde O(k^\epsilon)$ & \textbf{Our Result} \\
    \bottomrule
  \end{tabular}
  \caption{State-of-the-art for {\bf fully dynamic $k$-median} in general metrics. The table for {\bf fully dynamic $k$-means} is identical, except that the approximation ratio of our algorithm is $O(1/\epsilon^2)$.}
  \label{tab:kmedian}
\end{table}

\begin{table}[h!]
  \centering
  \begin{tabular}{cccc}  
    \toprule
    Approximation Ratio & Update Time & Recourse & Paper \\
    \midrule
     $2 + \epsilon$ & $\tilde O(k^2)$ & $O(k)$ & \cite{ChanGS18} \\
     $2 + \epsilon$ & $\tilde O(k)$ & $O(k)$ & \cite{BateniEFHJMW23} \\
     $O(1)$ & $\text{poly}(n, k)$ & $O(1)$ & \cite{LHRJ23} \\
      $O(\log n \log k)$ & $\tilde O(k)$ & $\tilde O(1)$ & \textbf{Our Result} \\
    \bottomrule
  \end{tabular}
  \caption{State-of-the-art algorithms for {\bf fully dynamic $k$-center} in general metric spaces.}
  \label{tab:kcenter}
\end{table}

We obtain our results by adapting two versatile techniques, ``local search''~\cite{AryaGKMMP04} and ``Lagrangian relaxation''~\cite{JainV01}, in the dynamic setting.
Throughout this paper, we only consider ``amortized update time'' (resp.~``amortized recourse'').\footnote{This is the total time (resp.~total recourse) spent by the algorithm across an update sequence, divided by the number of updates.}
{\em All our algorithms can handle an ``adaptive adversary'', which  determines the  future updates based on the concerned algorithm's past outputs.}

The machinery behind our dynamic local search framework even generalizes to any ``$\ell^p$-norm objective''. Here, for a given $p \geq 1$,  we wish to find a set $S \subseteq V$ of $k$ centers so as to minimize $\cl_p(S, V) := \left( \sum_{x \in V} w(x) \cdot \left( d(x, S) \right)^p \right)^{1/p}$. Our main result is summarized in the theorem below.

\begin{theorem}
\label{thm:lpnorm:main}
    For any $\epsilon > 0$, there is a dynamic $k$-clustering algorithm for the $\ell^p$-norm objective, with $O(p/\epsilon)$ approximation ratio, $\tilde O(k^{1 + \epsilon}p/\epsilon^3)$ update time and $\tilde O(k^\epsilon/\epsilon)$ recourse.
\end{theorem}
Since $\cl_p(S, V)$ respectively corresponds to  $k$-median,  $k$-center and  square root of the $k$-means objectives when $p = 1$,  $p = \Theta(\log n)$ and $p = 2$, Theorem~\ref{thm:lpnorm:main} implies the following corollaries.

\begin{corollary}
\label{thm:solution:main:1}
For any  $\epsilon > 0$, there is an  algorithm for dynamic $k$-median with $O(1/\epsilon)$ approximation ratio, $\tilde{O}(k^{1+\epsilon}/\epsilon^3)$   update time and $\tilde{O}(k^{\epsilon}/\epsilon)$  recourse.
\end{corollary}

\begin{corollary}
    \label{cor:lpnorm:1}
    For any  $\epsilon > 0$, there is an  algorithm for dynamic $k$-means with $O(1/\epsilon^2)$ approximation ratio, $\tilde{O}(k^{1+\epsilon}/\epsilon^3)$   update time and $\tilde{O}(k^{\epsilon}/\epsilon)$  recourse.
\end{corollary}

\begin{corollary}
    \label{cor:lpnorm:2}
    There is an  algorithm for dynamic $k$-center with $O(\log n \log k)$ approximation ratio, $\tilde{O}(k)$   update time and $\tilde{O}(1)$  recourse.
\end{corollary}

Corollary~\ref{thm:solution:main:1} brings us tantalizingly close towards resolving {\bf Q1} in the affirmative. But can we obtain a {\em truely} $O(1)$ approximation in $\tilde{O}(k)$ update time? Theorem~\ref{thm:value:main:2}, which is based on the first adaptation of the Lagrangain relaxation framework~\cite{JainV01} in the dynamic setting, provides further hope for us to be optimistic. In particular, it resolves {\bf Q1} for the {\em value version} of the problem.

\begin{theorem}
\label{thm:value:main:2}
There is a dynamic algorithm that maintains a $O(1)$ approximate estimate of the {\em value} of  the optimal $k$-median objective, in $\tilde{O}(k)$  update time.
\end{theorem}

\medskip
\noindent {\bf Related Work.} In addition to the  results summarized in Table~\ref{tab:kmedian} and Table~\ref{tab:kcenter}, there is a  related line of work   on  optimizing  recourse  (while ignoring any consideration regarding update time) in the ``incremental setting''. Specifically, this model  allows {\em only} for point insertions, as opposed to our model which allows for {\em both} insertions and deletions of points. The concerned line of work here has led to $O(1)$ approximation algorithms for $k$-center,  $k$-median and $k$-means that respective incur a ``total recourse'' of $\tilde{O}(k)$, $\tilde{O}(k)$ and $\tilde{O}(k^2)$, while handling a sequence of $n$ insertions~\cite{icml/LattanziV17,soda/FichtenbergerLN21}.

\subsection{Our Techniques}
\label{sec:techniques}


To achieve  Corollary~\ref{thm:solution:main:1}, our starting point is a  ``template algorithm'' that  constructs a hierarchy of nested subsets $V \supseteq S_0 \supseteq S_1 \supseteq \cdots \supseteq S_{\ell+1}$, for $\ell := 1/\epsilon$ (see Section~\ref{sec:building:new}). For each $i \in [0, \ell+1]$, we define a ``slack parameter'' $s_i := \lfloor k^{1-i\epsilon} \rfloor$ and ensure that $|S_i| = k+  s_i$. Each ``layer'' $S_i$  is computed by (approximately) solving the $(k+s_i)$-median problem on $V$, subject to the constraint that we are only allowed to pick our centers from within the set $S_{i-1}$ that defines the previous layer. Note that $|S_{\ell+1}| = k$ since $s_{\ell+1} = 0$, and so the set $S_{\ell+1}$ corresponds to the solution maintained by our algorithm. In the dynamic setting, we maintain this hierarchy using a lazy approach: After every $(s_i+1)$ updates to $V$, we ``rebuild'' the layers $S_i \supseteq S_{i+1} \supseteq \cdots \supseteq S_{\ell+1}$ in the hierarchy (see Section~\ref{sec:dyn:recouse:main}). We now have to overcome two major challenges. (1) We need to prevent the approximation ratio from blowing up exponentially with $1/\epsilon$ (the number of layers). (2) We need a subroutine that is fast enough to handle the ``rebuild'' step for layers $S_i \supseteq \cdots \supseteq S_{\ell+1}$, which occurs after every $(s_i+1)$ updates (see Section~\ref{sec:challenge:new}).

To overcome challenge (1), we show that to compute the layer $S_i$ we can perform a local search with only the points in $S_{i-1}$ being potential centers. With this ``restriction'' on the set of possible local moves, it is not feasible to expect that the solution $S_i$ we get is a $O(1)$ approximation to the global optimum. But surprisingly, we manage to show that the cost of $S_i$ is at most $O(1)$ times the global optimum, {\em plus} the cost of  $S_{i-1}$ (see Section~\ref{sec:localsearch:approx:new}).
This {\em fine-grained} analysis of the approximation guarantee of restricted local search helps us address challenge (1).   

Next, to overcome challenge (2), we obtain a deceptively simple implementation of the local search  procedure, which we refer to as ``randomized local search''  (see Algorithm~\ref{alg:local:new}). It  runs in iterations. In each iteration, it picks an eligible non-center uniformly at random, and performs the best possible  local move involving the sampled non-center. We show that if the number of iterations is only proportional to the number of eligible non-centers (up to polylogarithmic factors), then w.h.p.~the procedure still terminates with a sufficiently fine-grained approximation guarantee (see the discussion in the preceding paragraph) that we need to overcome challenge (1). We now address challenge (2)  by exploiting the  upper bound on the number of iterations (see Section~\ref{sec:random:new}). This concept of implementing local search via only a few {\em uniform random samples} should prove to be useful in other settings. In the full version (see Part~\ref{part:I}), we show that this entire framework generalizes to dynamic $k$-clustering with $\ell^p$-norm objective, leading to Theorem~\ref{thm:lpnorm:main}.

Finally, to achieve Theorem~\ref{thm:value:main:2}, we adapt an approach pioneered by Jain and Vazirani~\cite{JainV01}, who established a link between the $k$-median problem and a related problem called ``uniform facility location'' (UFL). Specifically, they showed that any ``Lagrange multiplier preserving'' (LMP) approximation algorithm for UFL  can be converted in a black-box manner into an algorithm for $k$-median, with negligible overheard in approximation ratio and running time. Our first contribution is to extend this reduction to the dynamic setting. Subsequently, we consider a known static UFL algorithm~\cite{MettuP00}, and show that a subtle modification to the algorithm makes it return a {\em fractional} LMP $O(1)$ approximation for UFL. We next  design a dynamic algorithm for maintaining this fractional solution with fast update time, and plugging it back into our  dynamic reduction, we are able to maintain a $O(1)$ approximate {\em fractional} solution for dynamic $k$-median. This leads us to Theorem~\ref{thm:value:main:2}, as the LP-relaxation for $k$-median has $O(1)$ integrality gap (see Section~\ref{sec:value:new}).

\section{Preliminaries}
\label{sec:prelim:new}

To highlight the essence of our technical contributions, for the rest of this extended abstract we will focus  on the dynamic $k$-median problem. But, before proceeding any further, we recall a recent dynamic algorithm by~\cite{ourneurips2023} that allows us to {\em sparsify} the input metric space $(V, w, d)$.

\begin{definition}[\texttt{Informal}]
\label{def:sparsify:new}
Consider a metric space $(V, w, d)$, a subset $V^{\star} \subseteq V$ and a weight-function $w^{\star} : V^{\star} \rightarrow \mathbb{R}^+$. We say that $(V^{\star}, w^{\star})$ is a $k$-median {\em sparsifier} of $(V, w, d)$ iff the following  holds for every subset $S^{\star} \subseteq V^{\star}$ of size $|S^{\star}| = k$: If $S^{\star}$ is a $O(1)$ approximate $k$-median solution w.r.t.~$(V^{\star}, w^{\star}, d)$, then it is also a $O(1)$ approximate $k$-median solution w.r.t.~$(V, w, d)$.
\end{definition}

An extension of the~\cite{ourneurips2023} algorithm can maintain a $k$-median sparsifier $(V^{\star}, w^{\star})$ of size $|V^{\star}| = \tilde{O}(k)$ in the input metric space $(V, w, d)$,  with $\tilde{O}(1)$ recourse and $\tilde{O}(k)$ update time (see Section~\ref{app:sparsification}). By Definition~\ref{def:sparsify:new}, it suffices to get a $O(1)$ approximate $k$-median clustering within this sparsifier. To summarize, because of this machinery, we can safely pretend that the dynamic input metric space consists of $n = \tilde{O}(k)$ points. Accordingly, Theorem~\ref{thm:solution:main}  implies Corollary~\ref{thm:solution:main:1}. 
 
For  Theorem~\ref{thm:value:main:2}, however,  we need to overcome an additional bottleneck before we can make use of this  dynamic sparsifier~\cite{ourneurips2023}. This is because the optimal $k$-median objective in  $(V^{\star}, w^{\star},d)$ need not necessarily be within a constant factor of the optimal $k$-median objective in the original input $(V, w, d)$. Indeed,  a closer inspection of Definition~\ref{def:sparsify:new} will convince the reader that a sparsifier need not approximately preserve the objective {\em value}, as opposed to preserving approximately optimal {\em solutions}. We show how to overcome this bottleneck in Part~\ref{part:II}. Skipping over this technicality,   the extended abstract outlines our approach for proving Theorem~\ref{thm:value:main}, which asks for an update time bound of $\tilde{O}(n)$ (instead of the $\tilde{O}(k)$ update time bound in Theorem~\ref{thm:value:main:2}).

\begin{theorem}
\label{thm:solution:main}
For any  $\epsilon > 0$, there is an  algorithm for dynamic $k$-median with $O(1/\epsilon)$ approximation ratio, $\tilde{O}(k^{\epsilon} \cdot n /\epsilon^3)$   update time and $\tilde{O}(k^{\epsilon}/\epsilon)$  recourse.
\end{theorem}

\begin{theorem}
\label{thm:value:main}
There is a dynamic algorithm that maintains a $O(1)$ approximate estimate of the {\em value} of the optimal $k$-median objective in $\tilde{O}(n)$ update time.
\end{theorem}

\noindent
{\bf Unweighted $k$-Median.} For simplicity of exposition, in the extended abstract we will further assume that the input metric space $(V, w, d)$ is {\em unweighted}, i.e., $w(x) = 1$ for all $x \in V$. Thus, from now on, in the rest of Part~\ref{part:0} we will omit mentioning the weight-function $w$ altogether. See the full version (i.e., Part~\ref{part:I} and Part~\ref{part:II}) for how we extend our dynamic algorithms to the weighted setting, and how we extend Theorem~\ref{thm:solution:main:1} to $\ell^p$-norm objective, which leads us to Theorem~\ref{thm:lpnorm:main}.

\subsection{Organization} We present an overview of the proof of Theorem~\ref{thm:solution:main}  (resp.~Theorem~\ref{thm:value:main}) in  Section~\ref{sec:solution:new}  (resp.~Section~\ref{sec:value:new}), for  {\em unweighted} $k$-median. The full version of our results appear in  Part~\ref{part:I} and Part~\ref{part:II}.

\section{Our Algorithm for Dynamic $k$-Median (Theorem~\ref{thm:solution:main})}
\label{sec:solution:new}

We start by considering the following setting, which we refer to as the {\bf improper dynamic $k$-median} problem. Let $\mathcal{V}$ be a collection of points defining a metric space that is given to the dynamic algorithm in advance, at preprocessing.\footnote{We assume that $\mathcal{V}$ is known in advance to simplify the presentation in the extended abstract; we can easily remove this assumption later.} Subsequently, each update consists of inserting/deleting a point $x \in \mathcal{V}$. Thus, throughout the sequence of updates we continue to have $V \subseteq \mathcal{V}$, where $V$ denotes the current input, and we let $n = |V|$. The algorithm can open a center at any point $x \in \mathcal{V}$, regardless of whether or not $x$ is part of  $V$. The goal is to maintain a set $S \subseteq \mathcal{V}$ of  $k$ centers, so as to minimize $\cl(S, V)$. We let $\Opt_k^{\mathcal{V}}(V) := \min_{S \subseteq \mathcal{V} : |S| = k} \cl(S, V)$ denote the optimal objective of the ``improper dynamic $k$-median'' problem.\footnote{Note that although we can open centers at points that are not currently present in $V$, the objective is measured by the sum of the distances of only the {\em existing} points in $V$ to their nearest open centers.} In contrast, we let $\Opt_k(V) := \min_{S \subseteq V : |S| = k} \cl(S, V)$ denote the optimal objective of the original problem, where each open center needs to be present in the current input. 

\medskip
\noindent {\bf Roadmap.}
We will first explain how to derive Theorem~\ref{thm:solution:main} for improper dynamic $k$-median, to showcase the main conceptual ideas behind our approach. Towards this end, we will start by introducing a few key building blocks in Section~\ref{sec:building:new}. This will be followed by an algorithm in Section~\ref{sec:dyn:recouse:main} that achieves the approximation and recourse guarantees of Theorem~\ref{thm:solution:main}. Section~\ref{sec:challenge:new} will explain how we intend to exploit the algorithmic framework developed in Section~\ref{sec:dyn:recouse:main}, so as to {\em simultaneously} achieve our goals: (a) good update time, (b) recourse and (c) constant approximation ratio. In Section~\ref{sec:localsearch:approx:new} and Section~\ref{sec:random:new}, we will show how to attain these three goals via a new ``randomized local search'' procedure. Finally, in Section~\ref{sec:assumption:new}, we will outline how to convert our algorithm for improper dynamic $k$-median into an algorithm for the original problem; with negligible overhead in the approximation ratio, recourse and update time. This will lead to the proof of Theorem~\ref{thm:solution:main}.

\subsection{Basic Building Blocks}
\label{sec:building:new}

\noindent {\bf Two Insights.} The first insight we will exploit is captured by Lemma~\ref{lem:proj lemma:new} and Corollary~\ref{cor:proj lemma:new}, and is well-known  in the literature~\cite{arxivGT2008, ChrobakKY06}. For any subset  $X \subseteq \mathcal{V}$, let $\pi_X : \mathcal{V} \rightarrow X$ be a mapping that assigns each point $y \in \mathcal{V}$ to its nearest neighbor in $X$ (breaking ties arbitrarily). Thus, we have $d(y, \pi_X(y)) = d(y, X)$ for all $y \in \mathcal{V}$. We refer to the point $\pi_X(y)$ as the ``projection'' of $y$ onto $X$. Further, for any  subset $S \subseteq \mathcal{V}$, we use the symbol $\pi_X \circ \pi_S : \mathcal{V} \rightarrow X$ to denote the ``composition'' of the two functions $\pi_X$ and $\pi_S$, so that $\pi_X \circ \pi_S (y) = \pi_X(\pi_S(y))$ for all points $y \in \mathcal{V}$. 

\begin{lemma}[Projection Lemma]
\label{lem:proj lemma:new}
$d(y, \pi_X \circ \pi_S(y)) \leq d(y, X) + 2 \cdot d(y, S)$ for all $y \in \mathcal{V}$, $S, X \subseteq \mathcal{V}$.
    \end{lemma}

\begin{proof}
Since $d(\pi_S(y), \pi_X \circ \pi_S(y)) \leq d(\pi_S(y), y')$ for every point $y' \in X$, we derive that: 
\begin{eqnarray*}
d(y, \pi_X \circ \pi_S(y)) & \leq & d(\pi_S(y), \pi_X \circ \pi_S(y)) + d(y,\pi_S(y))   \\
& \leq & d(\pi_S(y), \pi_X(y)) +  d(y,\pi_S(y))  \\
& \leq & d(y, \pi_X(y)) +  d(y, \pi_S(y)) + d(y, \pi_S(y)) =  d(y, X) + 2 \cdot d(y, S).
\end{eqnarray*}
This concludes the proof.
\end{proof}

For all  $S, X \subseteq \mathcal{V}$, let $\Proj(S, X) := \{ \pi_X(y) : y \in S \}$ denote the ``projection'' of $S$ onto $X$. Let $Q$ be a placeholder for a set of points. Corollary~\ref{cor:proj lemma:new} says that as switch from $Q = X$ to $Q = \Proj(S, X)$, the distance from  any $y \in \mathcal{V}$ to its nearest neighbor in $Q$ increases by  $\leq 2 \cdot d(y, S)$. 

\begin{corollary}
\label{cor:proj lemma:new}
We have $d(y, \Proj(S, X)) \leq d(y, X) + 2 \cdot d(y, S)$ for all $y \in \mathcal{V}$, $S, X \subseteq \mathcal{V}$.
\end{corollary}

\begin{proof}
Since $d(y, \Proj(S, X)) \leq d(y, \pi_X \circ \pi_S(y))$, the corollary follows from Lemma~\ref{lem:proj lemma:new}.
\end{proof}

Our second insight is summarized in 
Lemma~\ref{lem:lazy:update:new}. To see what it signifies, suppose that we are  maintaining a set $S \subseteq \mathcal{V}$ of $k$ centers for an input point-set $V \subseteq \mathcal{V}$. We next deal with $\gamma$ updates to $V$. Consider a lazy procedure which deals with this update sequence as follows: Whenever a point $x \in \mathcal{V}$ gets inserted into $V$, it sets $S \leftarrow S \cup \{x\}$. In contrast, if a point $x \in \mathcal{V}$ gets deleted from $V$, then it leaves the status of $S$ unchanged.\footnote{We can afford to do this because we are dealing with the improper dynamic $k$-median problem.}  This ensures that throughout the update sequence, the cost $\cl(S, V)$ does not increase, whereas the size of $S$ grows by at most an additive $\gamma$. Lemma~\ref{lem:lazy:update:new} is an immediate corollary of this lazy procedure, and it will play a key role in our analysis.

\begin{lemma}[Lazy Updates Lemma]
\label{lem:lazy:update:new} Consider any two subsets  $V', V'' \subseteq \mathcal{V}$ with $|V' \oplus V''| \leq \gamma$, where $\oplus$ denotes the symmetric difference. Then we have $\Opt_{k+\gamma}^{\mathcal{V}}(V'') \leq \Opt_k^{\mathcal{V}}(V')$. 
\end{lemma}

\begin{proof}
Let $S' \subseteq \mathcal{V}$ be a set of $k$ centers s.t.~$\cl(S', V') = \Opt_k^{\mathcal{V}}(V')$, and let $S'' = S' \cup (V'' \setminus V')$. Note that $S'' \subseteq \mathcal{V}$, and $|S''| \leq |S'| + \gamma \leq k + \gamma$. Since $d(x, S'') = 0$ for all  $x \in V'' \setminus V'$ and $d(x, S'') \leq d(x, S')$ for all  $x \in V'$, it follows that $\Opt_{k+\gamma}^{\mathcal{V}}(V'') \leq \cl(S'', V'') \leq  \cl(S', V') = \Opt_k^{\mathcal{V}}(V')$. 
\end{proof}

\noindent {\bf A Template Static Algorithm.} Given a (static) input $(V,d)$, we now  outline  a  procedure for computing a ``hierarchy''  of subsets $V  \supseteq S_0 \supseteq S_1 \supseteq \ldots \supseteq S_{\ell + 1}$, where $\ell = 1/\epsilon$. This procedure, described in Algorithm~\ref{alg:static}, will form the basis of our dynamic algorithm in subsequent sections. For each  $i \in \{0, \ldots, \ell+1\}$, we have $|S_i| = k+ s_i$, where $s_i := \lfloor k^{1-i  \epsilon} \rfloor$ is the ``slack'' at ``layer $i$''. (For simplicity,  assume that $|V| > 2k = k+s_0$.) The hierarchy is constructed in a bottom-up manner, starting with layer $0$.  At each layer $i$, the algorithm chooses $S_i$  to be a set of  $k+ s_i$ centers, nested within $S_{i-1}$, which minimizes the cost $\cl(S_i, V)$. Lemma~\ref{lem:static:new} and Corollary~\ref{cor:static:new} summarize the key properties of this hierarchy. We note that the proof of Lemma~\ref{lem:static:new} crucially relies on Corollary~\ref{cor:proj lemma:new}.

\medskip
\noindent {\bf Remark.} At this stage, it is not even clear that we can implement Algorithm~\ref{alg:static} in polynomial time (because of line 5). We will address this issue in Section~\ref{sec:challenge:new}. For now, we ignore any consideration regarding an efficient implementation of our algorithm.

\smallskip \ 

\begin{algorithm}[H]\label{alg:static}
    \SetAlgoLined
    \DontPrintSemicolon
    \texttt{// Let $S_{-1}$ denote the set $V$} \;
    \texttt{// Let $\epsilon > 0$ be a small constant (assume that $1/\epsilon$ is an integer)} \;
    \texttt{// Let $\ell := 1/\epsilon$ and $s_i := \lfloor k^{1-i  \epsilon} \rfloor$ for all integers $i \geq 0$}\;
    \For{$i = 0 \dots (\ell+1)$}{
        Let $S_i \leftarrow \arg\min_{S \subseteq S_{i-1} \, : \, |S| = k+  s_i} \cl(S, V)$\;
    }
    \Return the hierarchy of subsets $S_0 \supseteq S_1 \supseteq \ldots \supseteq S_{\ell + 1}$\;
    \caption{\textsc{StaticHierarchy}$((V, d), k)$}
\end{algorithm}

\begin{lemma}[Static Hierarchy Lemma]
\label{lem:static:new}
In Algorithm~\ref{alg:static} above, for all $i \in \{0, \ldots, \ell+1\}$, we have $$\cl(S_i, V) \leq \cl(S_{i-1}, V) + 2 \cdot \Opt_{k+s_i}(V).$$ 
\end{lemma}

\begin{proof}
Let $S \subseteq V$ be a set of $k+s_i$ centers s.t.~$\cl(S, V) = \Opt_{k+s_i}(V)$. Let $S' := \Proj(S, S_{i-1})$  be the projection of $S$ onto $S_{i-1}$. Note that $|S'| \leq  k+s_i$. By Corollary~\ref{cor:proj lemma:new}, we have $d(x, S')  \leq d(x, S_{i-1}) + 2 \cdot d(x, S)$ for all $x \in V$. Summing over all $x \in V$, we get $\cl(S', V) \leq \cl(S_{i-1}, V) + 2 \cdot \cl(S, V) = \cl(S_{i-1}, V) + 2 \cdot \Opt_{k+s_i}(V)$. Since $\cl(S_i, V) \leq \cl(S', V)$, the lemma follows.  
\end{proof}

\begin{corollary}
\label{cor:static:new}
In Algorithm~\ref{alg:static} above,  $S_{\ell+1}$ is a $O(1/\epsilon)$-approximate $k$-median solution on $(V, d)$.
\end{corollary}

\begin{proof}
First, we observe that $|S_{\ell+1}| = k+ s_{\ell+1}  = k$, since $s_{\ell+1} = \lfloor k^{-\epsilon} \rfloor = 0$. Next, summing the inequality from Lemma~\ref{lem:static:new} over all $i \in \{0, \ldots, \ell+1\}$, we get:  $\cl(S_{\ell+1}, V) \leq \cl(S_{-1}, V) + 2 \cdot \sum_{i=0}^{\ell+1} \Opt_{k+s_i}(V) = \cl(V, V) + 2 \cdot \sum_{i=0}^{\ell+1} \Opt_{k+s_i}(V) = 2 \cdot \sum_{i=0}^{\ell+1} \Opt_{k+s_i}(V) \leq 2 \cdot \sum_{i=0}^{\ell+1} \Opt_{k}(V)$ $=O(1/\epsilon) \cdot \Opt_k(V)$. This concludes the proof.
\end{proof}

\subsection{A Dynamic Algorithm with Small Recourse}
\label{sec:dyn:recouse:main}

We now outline how to use Algorithm~\ref{alg:static} to get a dynamic algorithm with small recourse. In this section, however, we do {\em not} pay any attention to bounding the update time of our algorithm  (we will address this issue in Section~\ref{sec:challenge:new}).  Our main result in this section is stated below.

\begin{theorem}
\label{thm:dyn:recourse:new} There is an algorithm for dynamic improper $k$-median with $O(1/\epsilon)$-approximation ratio and $\tilde{O}(k^{\epsilon}/\epsilon)$ amortized recourse.
\end{theorem}

We start by describing our algorithm. For simplicity,    assume that we always have $|V| > 4k$.

\medskip
\noindent {\bf Preprocessing.} At this stage, we receive the set  $\mathcal{V}$ and the initial input $V \subseteq \mathcal{V}$, and we build the hierarchy $\mathcal{V}  \supseteq S_0  \supseteq \cdots \supseteq S_{\ell+1}$ by calling $\textsc{StaticHierarchy}((V, d), k)$, as described in Algorithm~\ref{alg:static}. In addition, we associate a counter $\tau_i$ with each layer $i \geq 0$ of the hierarchy.   At preprocessing, we initialize $\tau_i \leftarrow 0$ for all $i \in \{0, \ldots, \ell+1\}$. We next explain how the algorithm handles an update. 

\medskip
\noindent {\bf Handling the Insertion/Deletion of a Point $x \in \mathcal{V}$ in $V$.} If the point $x$ is being inserted, then we simply set $S_i \leftarrow S_i \cup \{x\}$ for all $i \in \{0, \ldots, \ell+1\}$. In contrast, if the point $x$ is being deleted, then we leave the sets $S_0, \ldots, S_{\ell+1}$ unchanged (even if $x$ appears in some of these sets). We are allowed to do this because we are dealing with the dynamic improper $k$-median problem. We next increment all the counters, by setting $\tau_i \leftarrow \tau_i+1$ for all $i \in \{0, \ldots, \ell+1\}$. We then find the smallest index  $i \in \{0, \ldots, \ell+1\}$ such that $\tau_i > s_i$. Let $i'$ be the {\em smallest} such index $i$ (it always exists, because $\tau_{\ell+1} >  0 = s_{\ell+1}$).  Next, we call the subroutine {\sc Rebuild}$(i')$, as explained in Algorithm~\ref{alg:rebuild:new}. This subroutine updates the sets $S_{i'}, S_{i'+1}, \ldots, S_{\ell+1}$ in the hierarchy, by following the same principle as in the $\textbf{for}$ loop of Algorithm~\ref{alg:static}. In addition, it resets the counters $\tau_{i'}, \tau_{i'+1}, \ldots, \tau_{\ell+1}$ to zero. 

\smallskip \

\begin{algorithm}[H]\label{alg:rebuild:new}
    \SetAlgoLined
    \DontPrintSemicolon
    \For{$j = i' \dots (\ell+1)$}{
        Let $S_j \leftarrow \arg\min_{S \subseteq S_{j-1} \, : \, |S| = k+s_j} \cl(S, V)$ \,  // \texttt{If $j=0$, then let $S_{j-1} = S_{-1} = V$.}\;
        $\tau_j \leftarrow 0$\; 
    }    \caption{\textsc{Rebuild}$(i')$}
\end{algorithm}

\smallskip \

In the lemma below, we summarize a few key properties of this dynamic algorithm. 

\begin{lemma}
\label{lem:key:new}
After the dynamic algorithm has finished processing any given  update, we have: 
\begin{itemize}
\item {\em (P1)} $\mathcal{V} \supseteq S_0 \supseteq S_1 \supseteq \ldots \supseteq S_{\ell + 1}$.
\item {\em (P2)} $k + s_i \leq |S_i| \leq k+2s_i$ for all $i \in \{0, \ldots, \ell+1\}$.
\end{itemize}
\end{lemma}

\begin{proof}
Clearly, both the properties hold at pre-processing. It is easy to check that the algorithm never violates (P1) while handling an update. 

For (P2), fix any index $i \geq 0$, and say a ``critical event'' occurs whenever we call the subroutine {\sc Rebuild}$(i')$ for some $i' \leq i$. Note that at the end of any critical event, we have $|S_{i}| = k+s_i$ and $\tau_i = 0$. It is easy to verify that there can be at most $s_i$ updates to $V$ in between any two successive critical events (for otherwise we would have $\tau_i > s_i$), and each such update can increase the size of $S_i$ additively by at most one. Hence, we always have $k+s_i \leq |S_i| \leq k+2s_i$.  
\end{proof}

It now remains to bound the approximation ratio and recourse of this dynamic algorithm. In particular, Theorem~\ref{thm:dyn:recourse:new} will follow from Corollary~\ref{cor:approx:key:new} and Corollary~\ref{cor:key:recourse:new}.

\medskip
\noindent {\bf Analyzing the Approximation Ratio.} Consider any  $i, j \in \{0, \ldots, \ell+1\}$.  We use the symbols $S_i^j$ and $V^j$ to respectively denote the contents of the sets $S_i$ and $V$, at the end of the last update during which the dynamic algorithm reset the counter $\tau_j$ to $0$. In contrast, if we don't use the superscript $j$, then the symbols $S_i$ and $V$ respectively denote the status of the concerned sets at the present moment. Armed with these notations, we next derive the following lemma.

\begin{lemma}
\label{lem:approx:key:new}
For any index $i \geq 1$, we always have $\cl(S_i^i, V^i) \leq \cl(S_{i-1}^{i-1}, V^{i-1}) + 2 \cdot \Opt_{k}^{\mathcal{V}}(V)$. 
\end{lemma}

\begin{proof}
Suppose that we are currently at time $t$, and let $t_i$ (resp.~$t_{i-1}$) be the last time-step at which the counter $\tau_i$ (resp.~$\tau_{i-1}$) was reset to $0$. Clearly, we have $t_{i-1} \leq t_i \leq t$. From Lemma~\ref{lem:static:new}, we get:
\begin{equation}
\label{eq:lem:approx:key:new:1}
\cl(S_i^i, V^i) \leq \cl(S_{i-1}^i, V^i) + 2 \cdot \Opt_{k + s_i}^{\mathcal{V}}(V^i).
\end{equation}
At most $s_i$ updates have occurred to $V$ during the time-interval $(t_i, t)$; for otherwise we would have called the subroutine {\sc Rebuild}$(j)$, for some $j \leq i$, in between time $t_i$ and $t$, and this in turn would have reset the counter $\tau_i$ to $0$. Thus, we have $|V^i \oplus V| \leq s_i$, and  Lemma~\ref{lem:lazy:update:new} implies that:
\begin{equation}
\label{eq:lem:approx:key:new:2}
\Opt_{k+s_i}^{\mathcal{V}}(V^i) \leq \Opt_k^{\mathcal{V}}(V).
\end{equation}
Finally, the counter $\tau_{i-1}$ was {\em never} reset to $0$ during the time-interval $(t_{i-1}, t_i)$. Instead, during this time-interval, the set $S_{i-1}$ grew lazily by an additive one after each insertion in $V$, and remained unchanged after each deletion in $V$. Thus, we have $S_{i-1}^{i} = S_{i-1}^{i-1} \cup (V^i \setminus V^{i-1})$, and applying an argument similar to the proof of Lemma~\ref{lem:lazy:update:new}, we get:
\begin{equation}
\label{eq:lem:approx:key:new:3}
\cl(S_{i-1}^i, V^i) \leq \cl(S_{i-1}^{i-1}, V^{i-1}).
\end{equation}
The lemma follows from~(\ref{eq:lem:approx:key:new:1}),~(\ref{eq:lem:approx:key:new:2}) and~(\ref{eq:lem:approx:key:new:3}).
\end{proof}

\begin{corollary}
\label{cor:approx:key:new}
The set $S_{\ell+1}$  is a $O(1/\epsilon)$-approximate solution to the improper $k$-median problem. 
\end{corollary}

\begin{proof}
Since $\ell = 1/\epsilon$, summing the inequality from Lemma~\ref{lem:approx:key:new} over all $i \in \{1, \ldots, \ell+1\}$, we get:
\begin{equation}
\label{eq:cor:approx:key:new:1}
\cl(S_{\ell+1}^{\ell+1}, V^{\ell+1}) \leq \cl(S_0^0, V^0) + 2 \cdot \sum_{i=1}^{\ell+1} \Opt_{k}^{\mathcal{V}}(V) \leq  \cl(S_0^0, V^0) + O(1/\epsilon) \cdot \Opt_k^{\mathcal{V}}(V).
\end{equation}
Since $S_{-1} = V$, line 2 of Algorithm~\ref{alg:rebuild:new} gives us: $\cl(S_0^0, V^0) = \Opt_{k+s_0}^{\mathcal{V}}(V^0)$. Thus, by \Cref{lem:lazy:update:new}, we get that:  $\cl(S_0^0, V^0) = \Opt_{k+s_0}^{\mathcal{V}}(V^0) \leq \Opt_k^{\mathcal{V}}(V)$. Combining this observation with~(\ref{eq:cor:approx:key:new:1}), we infer that: 
\begin{equation}
\label{eq:cor:approx:key:new:2}
\cl(S_{\ell+1}^{\ell+1}, V^{\ell+1}) \leq   O(1/\epsilon) \cdot \Opt_k^{\mathcal{V}}(V).
\end{equation}
Next, the argument used to derive inequality~(\ref{eq:lem:approx:key:new:3}) (in the proof of Lemma~\ref{lem:approx:key:new}) also implies that:
\begin{equation}
\label{eq:cor:approx:key:new:3}
\cl(S_{\ell+1}, V) \leq  \cl(S_{\ell+1}^{\ell+1}, V^{\ell+1}).
\end{equation}
Finally, by Lemma~\ref{lem:key:new}, we have $S_{\ell+1} \subseteq \mathcal{V}$ and $k+s_{\ell+1} \leq |S_{\ell+1}| \leq k+2s_{\ell+1}$. Since $s_{\ell+1} = \lfloor k^{-\epsilon} \rfloor = 0$,  this implies that $|S_{\ell+1}| =  k$. The corollary now follows from~(\ref{eq:cor:approx:key:new:2}) and~(\ref{eq:cor:approx:key:new:3}).
\end{proof}

\noindent {\bf Analyzing the Recourse.} We next bound the amortized recourse of this dynamic algorithm.

\begin{lemma}
\label{lem:key:recourse:new} Consider any call to  {\sc Rebuild}$(i)$, for some $i \in [1,\ell+1]$. Let $j \in [i, \ell+1]$, and let $S_{j}^-$ (resp.~$S_{j}^+$) denote the contents of  $S_{j}$ just before (resp.~after) the call. Then  $|S_{j}^- \oplus S_{j}^+| \leq 4 s_{i-1}$. 
\end{lemma}

\begin{proof}
We clearly have $S_{j}^{-}, S_{j}^+ \subseteq S_{i-1}$, and $|S_{j}^+|, |S_{j}^-| \geq k+s_j \geq k$. Furthermore, Lemma~\ref{lem:key:new} implies that $|S_{i-1}| \leq k + 2s_{i-1}$. In  words, we have a set $S_{i-1}$ of (at most) $k+2s_{i-1}$ points, and we also have two subsets $S_{j}^-, S_{j}^+$ of $S_{i-1}$, each of size $\geq k$. It is easy to verify that under such a scenario, the concerned  subsets can differ in at most $4s_{i-1}$ points.  Hence, we get $|S_{j}^- \oplus S_{j}^+| \leq 4 s_{i-1}$.
\end{proof}

\begin{corollary}
\label{cor:key:recourse:new}
For every $j \in \{0,  \ldots, \ell+1\}$, the set $S_{j}$ incurs $O(k^{\epsilon}/\epsilon)$  amortized recourse. 
\end{corollary}

\begin{proof}
Excluding the lazy insertions (which only lead to an amortized recourse of at most $1$), the set $S_{j}$ changes iff we call {\sc Rebuild}$(i)$ for some $i \in [0,  j]$. Let $R_i$ be the amortized recourse of $S_j$, incurred due to calls to {\sc Rebuild}$(i)$. We will show that $R_i = O(k^{\epsilon})$ for all $i \in [0, j]$. Since $j = O(1/\epsilon)$, this implies an overall amortized recourse of $\sum_{i=0}^{j} R_i = O(k^{\epsilon}/\epsilon)$ for the set $S_j$. 

First, consider any index $i \in [1,j]$. There are at least $s_i$ updates to $V$ in between any two successive calls to {\sc Rebuild}$(i)$, and any such call changes $S_{j}$ by at most an additive $4s_{i-1}$ (see Lemma~\ref{lem:key:recourse:new}).  Hence, we get: $R_i = O(s_{i-1}/s_i) = O(k^{\epsilon})$ for all $i \geq 1$.\footnote{If $i = j =  \ell+1$, then we  pay  $O(s_{\ell}) = O(k^{\epsilon})$ recourse per update. So, we still have $R_i = O(k^{\epsilon})$ although $s_{\ell+1} = 0$.}  Finally, set $i = 0$. There are at least $s_0$ updates to $V$ in between any two successive calls to {\sc Rebuild}$(0)$, and any such call changes the set $S_{j}$, trivially, by at most  $O(k+s_j) = O(k)$. Hence, we get: $R_0 = O(k/s_0) = O(1)$.
\end{proof}

\subsection{Towards Simultaneously Achieving Good Recourse and Update Time}
\label{sec:challenge:new}

Our next challenge is to efficiently {\em implement} (a variant of) the algorithmic framework developed in  Section~\ref{sec:dyn:recouse:main},  so that we can achieve the  guarantees of Theorem~\ref{thm:dyn:recourse:new} in $\tilde{O}_{\epsilon}(k^{\epsilon} \cdot n)$ update time. The essence of our approach to address this challenge is captured in the lemma below.


\begin{lemma}[\texttt{Informal}]
\label{lem:informal:new}
Consider a (static) algorithm $\textsc{Alg}(X, V, k')$ which takes as input: two sets of points $X, V \subseteq \mathcal{V}$ and an integer $k'$ s.t.~$|X| > k'$; and returns a set $S \subseteq X$ of $|S| = k'$ points. Suppose that this algorithm satisfies the following two properties.
\begin{itemize}
\item {\em (P1)} $\cl(S, V) \leq \alpha \cdot \cl(X, V) + \beta \cdot \Opt_{k'}(V)$; for some constants $\alpha, \beta \geq 1$. 
\item {\em (P2)} The algorithm runs in $\tilde{O}(|V| \cdot (|X| - k'))$ time.
\end{itemize}
Then we can obtain an algorithm for dynamic improper $k$-median that has approximation ratio $O(\beta \cdot \sum_{i=0}^{\ell+1} \alpha^i)$, amortized recourse $O(k^{\epsilon}/\epsilon)$, and amortized update time $\tilde{O}_{\epsilon}(k^{\epsilon} \cdot n)$.\footnote{The notation $O_{\epsilon}(\cdot)$ hides  $\text{poly}(\epsilon^{-1})$ terms, which we ignore in the extended abstract while bounding  update time.}
\end{lemma}

\noindent {\bf An Important Caveat.} Note that property (P2) is too good to be true. Indeed, suppose that $X = V$ and $|X| = k'+1 = n'$ (say). Under this scenario, the optimal solution will simply identify the closest pair of points $x, y$ in $X$, and return the set $S := X \setminus \{x\}$.  Thus, property (P2) is  asking for a  (static) constant approximation algorithm for the closest pair of points in a {\em general metric space} of size $n'$, that runs in $\tilde{O}(n')$ time. Unfortunately for us, such an algorithm cannot exist.\footnote{Suppose that there are two special points $x, y \in X$ with $d(x, y) = 0$, and $d(x', y') = \infty$ whenever $\{x', y'\} \neq \{x, y\}$. Given such an input, it is not possible to identify  one of these special points $\{x, y\}$ in $o((n')^2)$ time.}

We will circumvent this impossibility result by designing an ``auxiliary'' {\em dynamic} data structure $\mathcal{D}$  which will run ``in the background'', processing the insertions/deletions in  $V$ in $\tilde{O}_{\epsilon}(k^{\epsilon} \cdot n)$ update time. Whenever we  call $\textsc{Alg}(\cdot, \cdot, \cdot)$ as a subroutine, it will be able to make certain queries to $\mathcal{D}$ and get answers to these queries very quickly (because $\mathcal{D}$ will have very fast query time). The time spent on each call to $\textsc{Alg}(\cdot, \cdot, \cdot)$ will indeed be proportional to the bound in (P2), {\em provided we exclude the update time of $\mathcal{D}$ running in the background} (see Section~\ref{sec:random:new} for further details).

\medskip
\noindent {\bf Restricted $k$-Median.} Henceforth, we will refer to any algorithm satisfying (P1) as an $(\alpha, \beta)$-approximation algorithm for the ``restricted $k$-median'' problem. As $\ell = 1/\epsilon$ and we are shooting for a $O(1/\epsilon)$ approximation ratio (see Theorem~\ref{thm:dyn:recourse:new}), we need to ensure that $\alpha \leq 1+\epsilon$ and $\beta = \Theta(1)$. Otherwise, if $\alpha = \Theta(1)$, then the approximation guarantee will degrade {\em exponentially} with $1/\epsilon$.

\medskip
\noindent {\bf Proof (Sketch) of Lemma~\ref{lem:informal:new}.} We run the  algorithm from Section~\ref{sec:dyn:recouse:main}, with the following twists. First, whenever we have to  recompute some $S_j$ (for $j \geq 1$) in line 2 of Algorithm~\ref{alg:rebuild:new}, we set: 
$$S_j \leftarrow \textsc{Alg}(S_{j-1}, V, k+s_j). \ \ \ \   (\text{This step takes } \tilde{O}(n \cdot (s_{j-1}-s_{j})) = \tilde{O}(n \cdot s_{j-1}) \text{ time, by (P2)})$$
 Second, whenever we have to recompute the set $S_0$ in line 2 of Algorithm~\ref{alg:rebuild:new} (for $j = 0$), we run a $O(1)$-approximate static $(k+s_0)$-median algorithm, on input $(V, d)$, in $\tilde{O}(n k)$ time~\cite{MettuP02}. This gives us  the new (recomputed) set  $S_0 \subseteq V$ of size $|S_0| = k+s_0$, with  $\cl(S_0, V) \leq O(1) \cdot \Opt_{k+s_0}(V)$.

Since $\textsc{Alg}(\cdot,\cdot,\cdot)$ returns an $(\alpha, \beta)$-approximation for restricted $k$-median, we   retrace the argument in the proof of Lemma~\ref{lem:approx:key:new} to get: $\cl(S_i^i, V^i) \leq \alpha \cdot \cl(S_{i-1}^{i-1}, V^{i-1}) + \beta \cdot \Opt_k^{\mathcal{V}}(V)$, for all $i \geq 1$. Further, by the discussion in the preceding paragraph, we  have: $\cl(S_0^{0}, V^{0}) \leq O(1) \cdot \Opt_{k+s_0}(V^{0}) \leq O(1) \cdot \Opt_k^{\mathcal{V}}(V)$. Plugging these inequalities back in the proof of Corollary~\ref{cor:approx:key:new}, we obtain the desired approximation ratio. The recourse analysis stays the same as in Lemma~\ref{lem:key:recourse:new} and Corollary~\ref{cor:key:recourse:new}.

Finally, note that the update time is dominated by the calls to the subroutine $\{\textsc{Rebuild}(i)\}_{i \geq 0}$. The key observation is that for all $i \geq 0$,  there are at least $s_i$ updates to $V$ between any two successive calls to {\sc Rebuild}($i$). For consistency of notations, assume that $s_{-1} = \Theta(k)$. Accordingly, for all $i \geq 0$, a given call to {\sc Rebuild}$(i)$ takes time $\sum_{j=i}^{\ell+1} \tilde{O}(n \cdot s_{j-1}) =  \tilde{O}(n \cdot s_{i-1}/\epsilon) = \tilde{O}_{\epsilon}(n \cdot s_{i-1})$. This leads to an overall amortized update time of $\sum_{i=0}^{\ell+1} \tilde{O}_{\epsilon}(n \cdot s_{i-1}/s_i) = \tilde{O}_{\epsilon}(n \cdot k^{\epsilon})$.

\subsection{Approximation Guarantee of Local Search for Restricted $k$-Median}
\label{sec:localsearch:approx:new}

We will show that a specific variant of local search is able to satisfy the two properties summarized in Lemma~\ref{lem:informal:new} (subject to the ensuing caveat), with $\alpha = 1+\epsilon$ and $\beta = \Theta(1)$. For starters, we focus on studying the approximation guarantee of  local search for the restricted $k$-median problem, without paying any attention to its running time. Subsequently, in Section~\ref{sec:random:new}, we will explain how to achieve property (P2), by slightly modifying the local search procedure described in this section.

\medskip
\noindent {\bf The Set Up.} Recall that $V \subseteq \mathcal{V}$ is the set of input points for the  $k$-median problem. {\bf For the rest of Section~\ref{sec:solution:new}, we will write  $\cl(S)$ instead of $\cl(S, V)$, since the clustering cost of a set $S$ will always be measured w.r.t.~the points in $V$.} (Also, the rest of this section will deal with the static setting.)

Consider any two subsets $\emptyset \subset S \subset X \subseteq \mathcal{V}$, with $|S| = k$. A ``local move'' (or, equivalently, a ``swap'') w.r.t.~$S$ and $X$  is an ordered pair $(x, y) \in S \times X$. It corresponds to inserting  $y$ into $S$ and deleting $x$ from $S$. We define $\delta_S(x,y) := \cl(S + y - x) - \cl(S)$ to be the change in the clustering cost of $S$ due to this swap.  We say that $S$ is  ``locally optimal'',  w.r.t.~$X$ and $V$,  iff no local move can improve (i.e., reduce)  the clustering cost of $S$. Thus, $S$ is locally optimal iff  $\delta_S(x, y) \geq 0$ for all $(x, y) \in S \times X$. The main result in this section is summarized  below.

\begin{theorem}
\label{thm:local:new}
If $S$ is  locally optimal w.r.t.~$X$ and $V$, then  $\cl(S) \leq \cl(X) + 6 \cdot \Opt_k^{\mathcal{V}}(V)$. 
\end{theorem}

\noindent {\bf The Plan.} Henceforth, we focus on proving Theorem~\ref{thm:local:new}. The theorem  follows from Lemma~\ref{lem:key:100:new} and Lemma~\ref{lem:key:101:new}. At the end of the proof of Lemma~\ref{lem:key:101:new}, we conclude this section by deriving Corollary~\ref{cor:local:search:new},  which will be useful later on in Section~\ref{sec:random:new}.

\medskip
\noindent {\bf Notations.} Let $S^{\star} \subseteq \mathcal{V}$ be a set of $|S^{\star}| = k$ centers  that forms the optimal solution to the $k$-median problem on the given input, i.e., $\cl(S^{\star}) = \Opt_k^{\mathcal{V}}(V)$. Recall the notations introduced just before Lemma~\ref{lem:proj lemma:new}, and let $\pi := \pi_X \circ \pi_{S^{\star}}$. 
We slightly abuse notation to define $\cl(\pi) := \sum_{z \in V} d(z, \pi(z))$.

\medskip
\noindent {\bf The Main Technical Challenge.} It is relatively straightforward to show that
\begin{equation}
\label{eq:easy:new}\cl(S) \leq 5  \cdot \cl(X) + 10 \cdot \Opt_k^{\mathcal{V}}(V).
\end{equation}
Indeed, let $Y = \Proj(S^{\star}, X)$ denote the projection of $S^{\star}$ onto $X$. By Corollary~\ref{cor:proj lemma:new}, we have $\cl(Y) \leq \cl(X) + 2 \cdot \cl(S^{\star}) = \cl(X) + 2 \cdot \Opt_k^{\mathcal{V}}(V)$. Furthermore, since $Y \subseteq X$ and has size $|Y| \leq |S^{\star}| = k$,  and since $S \subseteq X$ is a locally optimal set of $k$ centers w.r.t.~$X$ and $V$, the standard analysis of local search for $k$-median~\cite{AryaGKMMP04} yields: $\cl(S) \leq 5 \cdot \cl(Y) \leq 5 \cdot \cl(X) + 10 \cdot \Opt_k^{\mathcal{V}}(V)$. Thus, this simple analysis already implies that $S$ is a $(5, 10)$-approximate solution to the restricted $k$-median problem, on input $(X, V)$. As discussed in Section~\ref{sec:challenge:new}, however, this gives an approximation guarantee for {\em our} problem that degrades exponentially with $(1/\epsilon)$, because here $\alpha = 5$. In contrast, Theorem~\ref{thm:local:new} implies that $S$ is a $(1, 6)$-approximation for restricted $k$-median, with $\alpha = 1$. Such a bound is much more amenable towards getting  an $O(1/\epsilon)$-approximation ratio for dynamic improper $k$-median.  Unsurprisingly, proving Theorem~\ref{thm:local:new} requires a lot more extra care and subtle analysis.

The first key  ingredient in the proof (of Theorem~\ref{thm:local:new}) is to consider the mapping $\pi := \pi_X \circ \pi_{S^{\star}}$, and work with the benchmark $\cl(\pi)$ as opposed to $\cl(Y)$. In particular, Theorem~\ref{thm:local:new}  follows if we add the two inequalities from  Lemma~\ref{lem:key:100:new} and Lemma~\ref{lem:key:101:new}. The second ingredient refers to our strategy for proving Lemma~\ref{lem:key:101:new}. We carry out the analysis of local search, as presented in Section 2.1 and Section 2.2 of~\cite{arxivGT2008}, with one crucial twist:  The analysis in~\cite{arxivGT2008} constructs a specific mapping $\sigma : S^{\star} \rightarrow S$,  upper bounds the change in the cost of the clustering $S$ due to the local move $(\sigma(x^{\star}), x^{\star})$, for all $x^{\star} \in S^{\star}$, and then adds up the resulting inequalities. In the proof of Lemma~\ref{lem:key:101:new}, we  construct the same mapping $\sigma : S^{\star} \rightarrow S$ and follow the same framework, but we  instead consider local moves of the form $(\sigma(x^{\star}), \pi_X(x^{\star}))$,  for all $x^{\star} \in S^{\star}$.

\begin{lemma}
\label{lem:key:100:new}
$\cl(\pi) - \cl(X) \leq 2 \cdot \Opt_k^{\mathcal{V}}(V)$.
\end{lemma}

\begin{proof}
Fix any $z \in V$. Lemma~\ref{lem:proj lemma:new} implies that $d(z, \pi(z)) \leq d(z, X) + 2 \cdot d(z, S^{\star})$, and so $d(z, \pi(z)) - d(z, X) \leq 2 \cdot d(z, S^{\star})$. Summing this inequality over all $z \in V$ concludes the lemma.
\end{proof}

\begin{lemma}
\label{lem:key:101:new}
$\cl(S) - \cl(\pi) \leq 4 \cdot \Opt_k^{\mathcal{V}}(V)$.
\end{lemma}

\begin{proof}
We first partition  the set $S$ into three subsets:\footnote{Here, $\pi_S^{-1}(x)=\{z \in V \mid \pi_S(z) = x\}$.} 
$$S_0 := \{ x \in S : |\pi^{-1}_S(x) \cap S^{\star}| = 0\}, \ S_1 := \{ x \in S : |\pi^{-1}_S(x) \cap S^{\star}| = 1\}, \text{ and } S_{\geq 2} :=  S \setminus (S_0 \cup S_1).$$
Let $S_1^{\star} := \{ x^{\star} \in S^{\star} : \pi_S(x^{\star}) \in S_1\}$. Observe that $S^{\star}_1 \subseteq S^{\star}$, $|S| = |S^{\star}| = k$ and $|S_1| = |S^{\star}_1|$. Thus, we get $|S^{\star} \setminus S^{\star}_1| = |S_0| + |S_{\geq 2}|$. Since $\pi_S$ assigns at least two points from $S^{\star} \setminus S^{\star}_1$ to each point in $S_{\geq 2}$, we also have $|S_{\geq 2}| \leq (1/2) \cdot |S^{\star} \setminus S^{\star}_1|$, and hence $|S_0| \geq (1/2)  \cdot |S^{\star} \setminus S^{\star}_1|$. This observation implies that we can construct a mapping $\sigma : S^{\star} \rightarrow S$ which satisfies the following properties.
\begin{eqnarray}
\label{eq:prop:1:new}
\sigma(x^{\star}) & = & \pi_S(x^{\star}) \in S_1 \  \text{ for all } x^{\star} \in S^{\star}_1. \\
\label{eq:prop:2:new}
\sigma(x^{\star}) & \in & S_0 \qquad  \qquad \ \, \text{ for all } x^{\star} \in S^{\star} \setminus S^{\star}_1. \\
\label{eq:prop:3:new}
|\sigma^{-1}(x)| & \leq & 2 \qquad \qquad \ \ \ \text{ for all } x \in S.
\end{eqnarray}
W.l.o.g., suppose that $S^{\star} = \{x^{\star}_1, \ldots, x^{\star}_k\}$; and $x_i = \sigma(x^{\star}_i) \in S$  and $x'_i = \pi_X(x^{\star}_i)$ for all $i \in [k]$. Consider the following collection of local moves $\mathcal{M} := \{(x_1, x'_1), \ldots, (x_k, x'_k) \}$. 

\begin{corollary}
    \label{cor:swap:new}
    For any two distinct indices $i, j \in [k]$, we have $\pi_S(x^{\star}_j) \neq x_i$. 
\end{corollary}

\begin{proof}
Since $\sigma(x^{\star}_i) = x_i$, from~(\ref{eq:prop:1:new}) and~(\ref{eq:prop:2:new}) we infer that $x_i \in S_0 \cup S_1$. If $x_i \in S_1$, then $x^{\star}_i$ is the unique point  $x^{\star} \in S^{\star}$ with $\pi_S(x^{\star}) = x_i$, and hence $\pi_S(x^{\star}_j) \neq x_i$. Otherwise, if $x_i \in S_0$, then there is no point $x^{\star} \in S^{\star}$ with $\pi_S(x^{\star}) = x_i$, and hence $\pi_S(x^{\star}_j) \neq x_i$.
\end{proof}

Since $S$ is locally optimal w.r.t.~$X$ and $V$, we have:
\begin{equation}
\label{eq:swap:1:new}
\delta_S(x_i, x'_i) \geq 0 \text{ for all } i \in [k], \text{ and hence } 0 \leq \sum_{i=1}^k \delta_S(x_i, x'_i).
\end{equation}
Fix any $i \in [k]$. We will  next derive an upper bound on $\delta_S(x_i, x'_i)$. Towards this end, let $S^+ := S+x'_i - x_i$ be the set of centers {\em after} we make this local move. In the clustering given by $S$ (resp.~$S^{\star}$), every point $z \in V$ is assigned to a center $\pi_S(z) \in S$ (resp.~$\pi_{S^{\star}}(z) \in S^{\star}$).  Suppose that as we switch from $S$ to $S^{+}$, we change the assignment of points in $V$ to the centers as follows. 
\begin{wrapper}
\begin{enumerate}
\item Each point $z \in \pi_{S^{\star}}^{-1}(x^{\star}_i) \cap V$ gets reassigned to the center $x'_i$. 

This increases the distance between $z$ and its assigned center by $d(z, x'_i) - d(z, \pi_S(z))$. 
\item Each point $z \in (\pi_{S}^{-1}(x_i) \cap V) \setminus (\pi_{S^{\star}}^{-1}(x^{\star}_i) \cap V)$ gets reassigned to the center $\pi_S \circ \pi_{S^{\star}}(z)$.

(In this case, Corollary~\ref{cor:swap:new} guarantees that $\pi_S \circ \pi_{S^{\star}}(z) \in S \setminus \{x_i\}$.) 

This increases the distance between $z$ and its assigned center by: 
$$d(z, \pi_S \circ \pi_{S^{\star}}(z)) - d(z, \pi_S(z)) \leq 2 \cdot d(z, \pi_{S^{\star}}(z)),$$
where the last inequality follows from the proof of Lemma~\ref{lem:proj lemma:new}.
\item Every other point $z \in V \setminus (\pi_{S}^{-1}(x_i) \cup \pi_{S^{\star}}^{-1}(x^{\star}_i))$ remains assigned to its old center $\pi_S(z)$.

This increases the distance between $z$ and its assigned center by $0$.
\end{enumerate}
\end{wrapper}

\noindent
Thus, due to the above reassignments, the sum of the distances from the points in $V$ to their assigned centers increases by at most: 
$$\sum_{z \in \pi_{S^{\star}}^{-1}(x^{\star}_i) \cap V} \left\{ d(z, x'_i) - d(z, \pi_S(z)) \right\} + \sum_{z \in \pi_{S}^{-1}(x_i) \cap V} 2 \cdot d(z, \pi_{S^{\star}}(z)) = T_i \text{ (say)}.$$ 
It is easy to see that $\delta_S(x_i, x'_i) \leq T_i$. Summing this inequality over all $i \in [k]$, we now get:
\begin{eqnarray}
\nonumber 
\sum_{i=1}^k \delta_S(x_i, x'_i) & \leq & \sum_{i=1}^k \sum_{z \in \pi_{S^{\star}}^{-1}(x^{\star}_i) \cap V} \left\{ d(z, x'_i) - d(z, \pi_S(z)) \right\} + \sum_{i=1}^k \sum_{z \in \pi_{S}^{-1}(x_i) \cap V} 2 \cdot d(z, \pi_{S^{\star}}(z)) \\
& = &  \cl(\pi) - \cl(S) + \sum_{i=1}^k \sum_{z \in \pi_{S}^{-1}(x_i) \cap V} 2 \cdot d(z, \pi_{S^{\star}}(z))  \label{eq:derive:1:new} \\
& \leq & \cl(\pi) - \cl(S) + 2 \cdot 2 \cdot \cl(S^{\star}) \label{eq:derive:2:new} \\
& = & \cl(\pi) - \cl(S) + 4 \cdot \Opt_k^{\mathcal{V}}(V). \label{eq:derive:3:new}
\end{eqnarray}
Equality~(\ref{eq:derive:1:new}) holds because the sets $\{ \pi^{-1}_{S^{\star}}(x^{\star}_i)\cap V\}_{i \in [k]}$ partition  $V$, and because $\pi(z) = \pi_X \circ \pi_{S^{\star}}(z) = \pi_X(x^{\star}_i) = x'_i$ for all $z \in \pi_{S^{\star}}^{-1}(x^{\star}_i) \cap V$. Inequality~(\ref{eq:derive:2:new}) follows from~(\ref{eq:prop:3:new}), because the latter implies that for every $x \in S$ there are {\em at most} two indices $i \in [k]$ such that $x_i = x$. Combining~(\ref{eq:derive:3:new}) with~(\ref{eq:swap:1:new}), we get $\cl(\pi) - \cl(S) + 4 \cdot \Opt_k^{\mathcal{V}}(V) \geq 0$, and hence $\cl(\pi) - \cl(S) \leq  4 \cdot \Opt_k^{\mathcal{V}}(V)$. 
\end{proof}


\begin{corollary}
\label{cor:local:search:new}
Consider any two sets $\emptyset \subset S \subset X \subseteq \mathcal{V}$, with $|S| = k$. If $(1-\epsilon) \cdot \cl(S) \geq \cl(X) + 6  \cdot \Opt_k^{\mathcal{V}}(V)$, then there is a collection  $\mathcal{M} = \{(x_1, x'_1), \ldots, (x_k, x'_k) \}$ of local moves w.r.t.~$S$ and $X$, such that: $\sum_{i=1}^k \delta_S(x_i, x'_i) \leq - \epsilon \cdot \cl(S)$. 
\end{corollary}

\begin{proof}
We define the collection $\mathcal{M}$  as in the proof of Lemma~\ref{lem:key:101:new}. From Lemma~\ref{lem:key:100:new} and the proof of Lemma~\ref{lem:key:101:new} (in particular, from inequality~(\ref{eq:derive:3:new})), we get:
$$\sum_{i=1}^k \delta_S(x_i, x'_i) \leq \cl(\pi) - \cl(S) + 4 \cdot \Opt_k^{\mathcal{V}}(V) 
 \leq \cl(X) + 6 \cdot \Opt_k^{\mathcal{V}}(V) - \cl(S) \leq -\epsilon \cdot \cl(S).$$
 This concludes the proof.
\end{proof}

\subsection{Achieving Fast Update Time via Randomized Local Search}
\label{sec:random:new}

In this section, we derive a fast implementation of the algorithm from Section~\ref{sec:localsearch:approx:new}. The procedure is summarized in Algorithm~\ref{alg:local:new}, and we refer to it as ``randomized local search''. The set up is  the same as in Section~\ref{sec:localsearch:approx:new} (recall the notations introduced just before the statement of Theorem~\ref{thm:local:new}). We will show how to efficiently compute a set of points $S \subseteq X$ of size $|S| = k$, such that $\cl(S) \leq (1-\epsilon)^{-1} \cdot \cl(X) + 6(1-\epsilon)^{-1} \cdot \Opt_k^{\mathcal{V}}(V)$. 

The randomized local search starts with an arbitrary subset $S \subseteq X$ of size $|S| = k$, and runs for $\tilde{\Theta}_{\epsilon}(s)$ iterations.\footnote{The notation $\tilde{\Theta}_{\epsilon}(\cdot)$ hides $\text{poly}(1/\epsilon)$ and $\text{polylog}(n, \Delta)$ factors, where $\Delta$ is the aspect ratio of the metric space.} In each iteration, we simply pick a point $y \in X \setminus S$ uniformly at random, and perform the best possible local move (w.r.t.~$X$ and $V$) involving $y$, with one important caveat: We ensure that the objective $\cl(S)$ never increases during an iteration. This means that we pick an $x \in S \cup \{y\}$ which minimizes $\cl(S+y-x) - \cl(S)$, and set $S \leftarrow S +y - x$. Since we could potentially set $x = y$, this ensures that line 5 of Algorithm~\ref{alg:local:new} returns an $x^{\star}$ such that $\cl(S+y-x^{\star}) - \cl(S) \leq 0$. 

\medskip
\noindent {\bf The Plan.} The approximation ratio of randomized local search  is summarized in Lemma~\ref{lem:random:new}. We analyze its runtime in Lemma~\ref{lem:local:runtime:new}. We conclude this section by putting everything back together, and by pointing out how this leads us to  Theorem~\ref{thm:solution:main} for the dynamic improper $k$-median problem.

\smallskip \

\begin{algorithm}[H]\label{alg:local:new}
    \SetAlgoLined
    \DontPrintSemicolon
    Consider any arbitrary subset $S \subset X$ of points, of size $|S| = k$.\;
    Let $s := |X| - |S|$. \;
    \For{$\tilde{\Theta}_{\epsilon}(s)$ \textup{\textbf{iterations}}}{
        Sample a point $y \sim X \setminus S$ independently and u.a.r.\;
        $x^{\star} \leftarrow \arg\min_{x \in S+y} \{\cl(S+y-x) - \cl(S) \}$\;
        $S \leftarrow S + y - x^{\star}$\;
    }
    \Return{$S$}
    \caption{\RandLoc$(X, V, k)$, where $X, V \subseteq \mathcal{V}$ and $|X| > k$}
\end{algorithm}


 \begin{lemma}
\label{lem:random:new}
W.h.p.,  $\cl(S) \leq (1-\epsilon)^{-1} \cdot \cl(X) + 6(1-\epsilon)^{-1} \cdot \Opt_k^{\mathcal{V}}(V)$ when Algorithm~\ref{alg:local:new} terminates.
 \end{lemma}

 \begin{proof}(Sketch)
Given any point $y \in \mathcal{V}$, define $\delta_S(y) := \min_{x \in S+y} \{\cl(S+y-x) - \cl(S)\}$. Next, suppose that at the start of a  given iteration of the {\bf for} loop in Algorithm~\ref{alg:local:new}, we have:   
\begin{equation}
\label{eq:random:1}
(1-\epsilon) \cdot \cl(S) \geq  \cl(X) + 6  \cdot \Opt_k^{\mathcal{V}}(V).
\end{equation}
Then, by Corollary~\ref{cor:local:search:new}, we  have: $\sum_{i=1}^k \delta_S(x_i, x'_i) \leq -\epsilon \cdot \cl(S)$. Further, note that $\sum_{y \in X \setminus S} \delta_S(y) = \sum_{y \in X} \delta_S(y) \leq \sum_{i=1}^k \delta_S(x'_i) \leq \sum_{i=1}^k \delta_S(x_i, x'_i)$; where the first step holds because $\delta_S(y) = 0$ for all $y \in S$, and the second step holds because $\delta_S(y) \leq 0$ for all $y \in X$. Putting all these observations together, we conclude that $\sum_{y \in X \setminus S} \delta_S(y) \leq -\epsilon \cdot \cl(S)$, and hence:
\begin{equation}
\label{eq:random:2}
\mathbb{E}_{y \sim X \setminus S}[\delta_S(y)] \leq -(\epsilon/s) \cdot \cl(S).
\end{equation}
The LHS of~(\ref{eq:random:2}) is the expected (additive) change in the clustering cost of $S$ due to the concerned iteration of the {\bf for} loop. Thus,  as long as~(\ref{eq:random:1}) holds, the next iteration of the {\bf for} loop in Algorithm~\ref{alg:local:new} decreases the value of $\cl(S)$ by a multiplicative factor of $(1-\epsilon/s)$, in expectation. From here, it is relatively straightforward to argue that after $\tilde{\Theta}_{\epsilon}(s)$ iterations of the {\bf for} loop, the value of $\cl(S)$ drops by such an extent that~(\ref{eq:random:1}) no longer holds, w.h.p. This concludes the proof.
 \end{proof}

 We next turn towards analyzing the runtime of randomized local search. As per the caveat  after the statement of Lemma~\ref{lem:informal:new}, we first describe a relevant auxiliary data structure.

 \medskip
 \noindent {\bf Auxiliary Data Structure.} For every $z \in V$, consider an ordered list $L_z$ of the points $x \in X$, sorted in non-decreasing order of their distances from $z$. For $i \in [n]$, let $L_z(i) \in X$ denote the point in the $i^{th}$ position of the list $L_z$ (i.e., the $i^{th}$ closest neighbor of $z$ in $X$).   These lists $\mathcal{L} = \{ L_z\}_{z \in V}$ are stored as balanced search trees, and we are given $\mathcal{L}$ for {\em free}, just before calling Algorithm~\ref{alg:local:new}. Under this assumption, the lemma below summarizes the runtime of the algorithm.

 \begin{lemma}
\label{lem:local:runtime:new}
Excluding the time taken to prepare the auxiliary data structure $\mathcal{L}$, Algorithm~\ref{alg:local:new} can be implemented in $\tilde{O}_{\epsilon}(n s)$ time. The algorithm modifies the lists in $\mathcal{L}$ during its execution. Before it terminates, however, the algorithm restores the data structure $\mathcal{L}$ to its initial state. 
 \end{lemma}

 \begin{proof}(Sketch)
After implementing line 1 (in Algorithm~\ref{alg:local:new}), we delete every point $x \in X \setminus S$ from every list in $\mathcal{L}$. Since $|X \setminus S| = s$ and there are $n = |V|$ lists in $\mathcal{L}$, this takes $\tilde{O}(ns)$ time. Accordingly, now the point $L_z(i)$  is  the $i^{th}$ closest neighbor of $z \in V$ {\em in $S$} (as opposed to in $X$).  

We next outline how to implement lines 5, 6 (within the {\bf for} loop in Algorithm~\ref{alg:local:new}). We start by inserting the point $y$ into appropriate positions of all the lists in $\mathcal{L}$. This takes $\tilde{O}(n)$ time. Next, we calculate the value  $\cl(S+y) = \sum_{z \in V} d(z, L_z(1))$ in $\tilde{O}(n)$ time, by scanning through all the points $z \in V$ and querying the list $L_z$ to obtain the point $L_z(1)$. Our next task is to calculate the value  $\cl(S+y-x) - \cl(S+y) = \lambda_x$ (say), for all $x \in S+y$. In the paragraph below, we show how to calculate these values $\{\lambda_x\}_{x \in S+y}$, altogether in $\tilde{O}(|S+y| + n) = \tilde{O}(k+n) = \tilde{O}(n)$ time.

For all $x \in S+y$, define $\Gamma(x) := \{ z \in V : L_z(1) = x\}$. We  compute these sets $\{ \Gamma(x)\}_{x \in S+y}$, altogether in $\tilde{O}(n)$ time, by scanning through all $z \in V$ and adding $z$ to  $\Gamma(L_z(1))$. Now, note that for all $x \in S+y$, we have $\lambda_x := \sum_{z \in \Gamma(x)} \left( d(z, L_z(2)) - d(z, x)\right)$. This is because if we delete the center $x$ from $S$, then the point $z$ will get reassigned to its $2^{nd}$ closest neighbor in $S$, given by $L_z(2)$. Since the sets $\{\Gamma(x)\}_{x \in S+y}$ form a partition of $V$, it is easy to verify that we can compute all the $\{\lambda_x\}_{x \in S+y}$ values in $\tilde{O}(|S+y| + n) = \tilde{O}(n)$ time.

At this stage, we identify  an $x \in S+y$ with minimum $\lambda_{x}$ value, set $x^{\star} \leftarrow x$, and delete $x^{\star}$ from all the lists in $\mathcal{L}$. This works because  a minimizer for $\lambda_x = \cl(S+y-x) - \cl(S+y)$ is also a minimizer for $\cl(S+y-x) - \cl(S)$. All these operations can be implemented in $\tilde{O}(n)$ time. 

To summarize,  lines 5, 6 (in Algorithm~\ref{alg:local:new}) can be implemented in $\tilde{O}(n)$ time, and hence the entire {\bf for} loop takes $\tilde{O}_{\epsilon}(ns)$ time. At the end of the {\bf for} loop, we insert all the points $x \in X \setminus S$ back in their appropriate positions in each of the $n$ lists $\{L_z\}_{z\in V}$, so as to return the data structure $\mathcal{L}$ to its initial state. This last step  also takes $\tilde{O}(|X \setminus S| \cdot n) = \tilde{O}(ns)$ time.
\end{proof}

\medskip
\noindent {\bf Our Overall Algorithm for Dynamic Improper $k$-Median.}
Lemma~\ref{lem:random:new} essentially implies that Algorithm~\ref{alg:local:new} satisfies property (P1) of Lemma~\ref{lem:informal:new} with $\alpha = (1-\epsilon)^{-1}$ and $\beta = \Theta(1)$. This translates into an approximation ratio of $O(\beta \cdot \sum_{i=0}^{\ell+1} \alpha^i)= O(\alpha^{1/\epsilon}) = O(1)$ for the improper $k$-median problem. Lemma~\ref{lem:local:runtime:new}, in contrast, guarantees the desired property (P2) of Lemma~\ref{lem:informal:new}. 
 
 As far as dynamically maintaining the  data structure $\mathcal{L}$ is concerned (see the caveat after the statement of Lemma~\ref{lem:informal:new}), we recall the proof (sketch) of Lemma~\ref{lem:informal:new} in Section~\ref{sec:challenge:new}, which implies that we need to invoke Algorithm~\ref{alg:local:new} only with $X = S_{j-1}$ and $k = k+s_j$, for  $j \in [\ell+1]$. Thus, we maintain $\ell+1 = O(1/\epsilon)$  different ``versions'' of the data structure $\mathcal{L}$, where version  $j \in [\ell+1]$ corresponds to  $X = S_{j-1}$ and we refer to the concerned data structure (i.e., the $j^{th}$ version)  as $\mathcal{L}^{(j)}$. 
 
 Fix any index $j \in [\ell+1]$. From Lemma~\ref{lem:key:recourse:new}, Corollary~\ref{cor:key:recourse:new} and the proof (sketch) of Lemma~\ref{lem:informal:new}, we infer that the set $S_{j-1}$ incurs an amortized recourse of $O(k^{\epsilon}/\epsilon)$. Furthermore, it is easy to verify that after each insertion/deletion in $S_{j-1}$, we can update the data structure  $\mathcal{L}^{(j)}$ in $\tilde{O}(n)$ time. Indeed, once a point $x$ gets inserted/deleted in $S_{j-1}$, all we need to do is update the positions of $x$ in the list $L_z^{(j)}$ for all $z \in V$, and it takes $\tilde{O}(1)$ time per list to perform this task. Using a similar procedure, after each insertion/deletion in  $V$, we can also update the data structure $\mathcal{L}^{(j)}$ in $\tilde{O}(|S_{j-1}|) = \tilde{O}(k+s_{j-1}) = \tilde{O}(n)$ time. Thus, we can maintain the data structure $\mathcal{L}^{(j)}$ in the dynamic setting, by paying an amortized update time of $\tilde{O}(n) \cdot O(k^{\epsilon}/\epsilon) = \tilde{O}_{\epsilon}(k^{\epsilon} \cdot n)$. 

 It follows that we can maintain all the data structures $\{\mathcal{L}^{(j)}\}_{j \in [\ell+1]}$ in the dynamic setting, for an overall amortized update time of $\sum_{j=1}^{\ell+1} \tilde{O}_{\epsilon}(k^{\epsilon} \cdot n) = \tilde{O}_{\epsilon}(k^{\epsilon} \cdot n)$. Referring  back to the discussion  after the statement of Lemma~\ref{lem:informal:new}, this collection $\{\mathcal{L}^{(j)}\}_{j \in [\ell+1]}$ plays the role of the auxiliary dynamic data structure $\mathcal{D}$. To summarize, Algorithm~\ref{alg:local:new} satisfies the desired properties of Lemma~\ref{lem:informal:new}. This gives us all the guarantees of Theorem~\ref{thm:solution:main}, albeit for the dynamic {\em improper} $k$-median problem.

 \subsection{Getting Rid of the Assumption Regarding Improper Dynamic $k$-Median}
 \label{sec:assumption:new}

\newcommand{\Al}{\texttt{Alg}}

So far, we have obtained a dynamic algorithm (say) $\Al$ for the improper $k$-median problem (see the discussion at the start of Section~\ref{sec:solution:new}), which achieves the guarantees of Theorem~\ref{thm:solution:main}. We now outline how to convert $\Al$ into an appropriate dynamic algorithm $\Al^{\star}$ for the original problem.  

To begin with, we note that while presenting our algorithm for improper dynamic $k$-median, we never used the fact that the algorithm needs to know the set $\mathcal{V}$ in advance. 
Moving on, let $S \subseteq \mathcal{V}$ be the clustering maintained by $\Al$, and let $S^{\star} := \{ \pi_V(x) : x \in S\}$ denote the projection of $S$ onto $V$. Since $|S| \leq k$, it follows that $|S^{\star}| \leq k$, and so $S^{\star}$ is a  solution to the original $k$-median problem. 

\begin{claim}[\cite{focs/GuhaMMO00}, Theorem 2.1]
\label{cl:projection:new}
We have $\cl(S^{\star}) \leq 2 \cdot \cl(S)$.
\end{claim}

\begin{proof}
Consider any point $x \in V$. The claim follows if we sum inequality~(\ref{eq:projection:new:1}) over all $x \in V$. 
\begin{eqnarray}
d(x, S^{\star}) & \leq & d(x, \pi_{V} \circ \pi_S(x))
 \leq  d(x, \pi_{S}(x)) + d(\pi_S(x), \pi_V \circ \pi_{S}(x)) 
 \leq   d(x, \pi_S(x)) + d(\pi_S(x), x) \nonumber \\
& = & 2 \cdot d(x, \pi_S(x)). \label{eq:projection:new:1}
\end{eqnarray}
In the above derivation, the first inequality holds because $\pi_V \circ \pi_S(x) \in S^{\star}$, whereas the third inequality holds because $\pi_S(x) = y \text{ (say)} \in S$ and $d(y, \pi_V(y)) \leq d(y, x)$ (since $y \in S$ and  $x \in V$).
\end{proof}

$\Al$  maintains the set $S$ with  $R = \tilde{O}(k^{\epsilon}/\epsilon)$ recourse (see Theorem~\ref{thm:solution:main}). It turns out that using $\Al$, we can  also maintain the set $S^{\star}$ with an  $R \cdot O(n) = \tilde{O}(k^{\epsilon} \cdot n/\epsilon)$ overhead in  update time: We keep an auxiliary data structure similar to the one defined just before Lemma~\ref{lem:local:runtime:new}. Specifically, for every center $z \in S$, we keep a sorted list (as a balanced search tree) $L_z$ which contains all the points in $V$, in non-decreasing order of their distances from $z$. For $i \in [n]$, let $L_z(i) \in V$ denote the point in the $i^{th}$ position of this list. It is easy to verify that after every point insertion/deletion in $S$ or $V$, all these lists can be updated in $\tilde{O}(n)$ total time. Further, we  have $S^{\star} := \{ L_z(1) : z \in S\}$.

Because of Claim~\ref{cl:projection:new}, the above reduction already gives us the desired approximation ratio and update time bound of Theorem~\ref{thm:solution:main}. The only bottleneck is that the recourse of the maintained clustering $S^{\star}$ can be as large as  $k$. In Section~\ref{sec:proper} (in the full version of the paper), we show how to get rid of this remaining bottleneck, which leads us to Theorem~\ref{thm:solution:main} for dynamic $k$-median.

\section{Our  Algorithm for Dynamic $k$-Median Value (Theorem~\ref{thm:value:main})}\label{sec:value:new}

Due to space constraints in the extended abstract, we only give a very brief overview of our dynamic algorithm for $k$-median value. The full algorithm and its analysis can be found in \Cref{part:II}.

Our algorithm from \Cref{sec:solution:new} maintains a $O(1/\epsilon)$-approximate solution to the $k$-median problem under insertions and deletions with $\tilde O(k^{1+\epsilon})$ update time. If we want an optimal update time of $\tilde O(k)$, the approximation ratio of this algorithm becomes $O(\log k)$.
However, if we are only interested in the \emph{value} of the optimal $k$-median solution, then we will show  that we can maintain a $O(1)$-approximation to this quantity in $\tilde O(k)$ update time.
The key to our procedure is the relationship between a \emph{Lagrange multiplier preserving} (LMP) algorithms for the \emph{uniform facility location} (UFL) problem and approximation algorithms for $k$-median.

\medskip
\noindent
\textbf{UFL and LMP Algorithms.}
In the (uncapacitated) facility location problem, we are given a metric space $(V,d)$ and a facility opening cost $f_i$ for each point $i \in V$. The objective is to find a subset of points $S\subseteq V$ on which to open facilities such that the sum of the facility opening costs $F(S)=\sum_{i\in V} f_i$ and the total service costs $C(S)=\sum _{j\in V} d(j, S)$ is minimized; here, $d(j, S)$ is the distance from $j$ to the nearest point in $S$, i.e.~the service cost of the client $j$. Facility location can be formulated as an integer program as follows. 
\begin{eqnarray*}
\text{Minimize }   \sum_{i \in V} f_i \cdot y_i + \sum_{i, j \in V} d(i, j) \cdot x_{j \to i}  \\
\text{s.t.~} \sum_{i \in V} x_{j \to i}  \geq  1 & & \text{ for all } j \in V \\
x_{j \to i} \leq y_i & & \text{ for all } i, j \in V \\
y_i, x_{j \to i} \in\{0,1\} & & \text{ for all  } i, j \in V
\end{eqnarray*}
The variable $y_i$ is 1 if a facility is opened at $i$ and is 0 otherwise. Similarly, $x_{j\to i}$ is 1 if $j$ is served by the facility opened at $i$, and is 0 otherwise. In an LP-relaxation, the integrality constraints on these variables are replaced by non-negativity constraints. An $\alpha$-approximation algorithm for  location  is called ``Lagrange multiplier preserving'' (LMP) iff it produces a solution $S$ that satisfies $F(S)+\alpha\cdot C(S)\le \alpha\cdot \Opt$, where $\Opt$ denotes the cost of the optimal solution. In the \emph{uniform} facility location (UFL) problem, we assume that all facilities have the same opening cost, $\lambda$.

\medskip
\noindent
\textbf{From LMP Algorithms to $k$-Median.}
If all facilities have a uniform cost $\lambda$, then decreasing (resp.~increasing) $\lambda$ would increase (resp.~decrease) the number of facilities we open in any solution that approximates the optimal solution. If there is a value of $\lambda$, say $\lambda^\star$ at which the number of facilities opened by an $\alpha$-approximate LMP algorithm is exactly $k$, then this set of facilities is also an $\alpha$-approximation to the optimal $k$-median solution. However, such a value of $\lambda$ might not exist. Jain-Vazirani~\cite{JainV01} used sufficiently close values of $\lambda$, say $\lambda^-,\lambda+$, such that the number of facilities opened at $\lambda^-$ (resp.~$\lambda^+$) are more (resp.~less) than $k$, and constructed a convex combination of the solutions to obtain a solution that opens exactly $k$ facilities. 

This approach to the $k$-median problem requires an LMP algorithm for UFL, and Jain and Vazirani~\cite{JainV01} gave such an algorithm that is a 3-approximation. However, their 2-phase primal-dual procedure picks facilities in a manner unsuitable for a dynamic setting. In contrast,  Mettu and Plaxton~\cite{MettuP02} gave a different take on the Jain-Vazirani algorithm that naturally combines the 2-phases, but it continues to pick facilities ``sequentially", making it challenging to get fast update times in a dynamic setting. While the original algorithm of Mettu-Plaxton was not LMP, Archer et al.~\cite{esa/ArcherRS03} showed that a  variation of it yields a 3-approximate LMP algorithm. 

In this paper, we also modify the Mettu-Plaxton algorithm, but to obtain a {\em fractional solution} that is 4-approximate and LMP. In a sense, this fractional solution can be computed ``in parallel'', and this helps us get around the bottleneck of ``sequentially'' picking facilities, as was the case in the original version of the algorithm. In particular, we prove Theorem~\ref{th:MP:main:new} stated below.


\medskip
\noindent
\textbf{(Our Modification of) the Mettu-Plaxton Algorithm.}
For all $i \in V$ and $r\geq 0$, let $B(i, r) := \{ j \in V \; | \; d(i, j) \leq 4r\}$ denote the (closed) ball of radius $r$ around the point $i$. Define
\begin{equation}
\label{MAIN:eq r}
r_i := \min \left\{ r \geq 0 \; \middle| \; \sum_{j \in B(i, r)} \left(r - \frac{d(i, j)}{4} \right) \geq \lambda\right\}.
\end{equation}
We use these $r_i$ values to define a fractional solution to UFL.
\begin{enumerate}
\item For $i \in V$, let $y_i :=  r_i/\lambda$, and
\item for $i,j \in V$, let $x_{j \to i} := \max\left(0, r_j - d(i, j)/4  \right)/ \lambda$.
\end{enumerate}
\begin{theorem}\label{th:MP:main:new}
The fractional solution $\{y_i\}_{i\in V}$,$\{x_{j \to i}\}_{i,j\in V}$ is an LMP $4$-approximation for UFL. 
\end{theorem}
Using some appropriate auxiliary data structures, we show that this fractional solution can be maintained very efficiently, in only $\tilde{O}(n)$ time, under deletions and insertions of points in $V$. Specifically, we build a data structure $\FracLMP$ which, when given a facility opening cost $\lambda$, returns the fractional opening values $\{y_i\}_{i\in V}$ and the total connection cost $\sum_{i,j\in V} d(i,j) \cdot x_{j\to i}$ in $\tilde{O}(n)$ time. Thus, the worst-case update time for $\FracLMP$ is $\tilde{O}(n)$.

To construct a fractional $k$-median solution, we need to determine a value of the facility opening cost $\lambda$ at which the 4-approximate LMP fractional solution opens ``close" to $k$ facilities. This can be achieved by querying $\FracLMP$ for the $\{y_i\}$ values for different inputs $\lambda$. By doing a binary search we determine two values of $\lambda$ which are within $(1+\epsilon)$factors of each other and then take a convex combination of the fractional solutions corresponding to these values, obtaining a fractional solution to the $k$-median problem.
This gives us the following theorem.
\begin{theorem}
There is a deterministic dynamic algorithm with  update time $\tilde{O}(n)$, which given any $\epsilon > 0$, $k \in \mathbb N$, returns a $4(1 + \epsilon)$-approximate fractional solution for $k$-median in $\tilde{O}(n)$ time.
\end{theorem}

\noindent
Our algorithm is also capable of maintaining the cost of this fractional solution. Since the LP-relaxation for $k$-median has a constant integrality gap \cite{esa/ArcherRS03}, this gives a $O(1)$-approximation to the cost of the optimal integral solution, leading us to Theorem~\ref{thm:value:main}.



\newpage

\part{Our Algorithm for Dynamic $k$-Clustering}\label{part:I}

\section{Preliminaries}\label{sec:prelim}

We begin by formally describing the notation that we use throughout \Cref{part:I}. See \Cref{sec:P2 org} below for a summary of how \Cref{part:I} is organized and the main result of this part of the paper.

\subsection{$k$-Clustering}\label{sec:k-clus def}
In this paper, we consider the \emph{metric $k$-clustering problem}. As the input to this problem, we are given some (weighted) metric space $(V, w, d)$ and an integer $k \in \mathbb N$, where $d : V \times V \longrightarrow \mathbb R_{\geq 0}$ is a metric satisfying the triangle inequality and $w : V \longrightarrow \mathbb R_{> 0}$ assigns a non-negative weight $w(x)$ to each point $x \in V$ in the space. Let $\Delta := d_{\max}w_{\max}/(d_{\min}w_{\min})$ denote the aspect ratio of the metric space.\footnote{Here, $d_{\min}$ (resp.~$d_{\max}$) is the smallest non-zero (resp.~largest) distance between any two points in $V$, and $w_{\min}$ (resp.~$w_{\max}$) is the minimum (resp.~maximum) weight of a point in $V$. We sometimes write $\Delta_d := d_{\max}/d_{\min}$ and $\Delta_w := w_{\max}/w_{\min}$ to denote the aspect ratio of the distances and the weights respectively.}
Our objective is to find a subset $S \subseteq V$ of at most $k$ points which minimises the \emph{clustering cost} $\cl(S)$ of the solution $S$, where $\cl(\cdot)$ is the objective function which defines the $k$-clustering problem. Specifically, we consider the \emph{$(k,p)$-clustering objective}
\begin{equation}\label{eq:cl_p}
    \cl_p(S, V, w) := \left( \sum_{x \in V} w(x) d(x, S)^p \right)^{1/p}
\end{equation}
where $p \in \mathbb N$ and $d(x, S) := \min_{y \in S} d(x, y)$. Furthermore, for $p = \infty$, we define $\cl_\infty (S, V, w) := \max_{x \in V} d(x, S)$. 
We refer to the points in $S$ as \emph{centers} and denote the cost of the optimal solution to the $(k,p)$-clustering problem in the space $(V, w, d)$ by $\Opt^p_k(V,w)$.
To ease notations, we often omit arguments that are clear from the context. In particular, we often abbreviate $\Opt^p_k(V, w)$ and $\cl_p(S, V, w)$ by $\Opt_k$ and $\cl(S)$ respectively.

\subsection{Relation to $k$-Median, $k$-Means, and $k$-Center}

The $(k,p)$-clustering problem can be seen as a generalization of the other well known $k$-clustering problems, such as $k$-median, $k$-means, and $k$-center. The $(k,p)$-clustering objective is often defined as $\cl_p(S)^p = \sum_{x \in V} w(x) d(x,S)^p$
(such as in \cite{ourneurips2023}). For technical reasons, we instead define it as $\cl_p(S) = ( \sum_{x \in V} w(x) d(x,S)^p)^{1/p}$
(such as in \cite{arxivGT2008}). Thus, with our definition, the $(k,p)$-clustering objective is precisely the $\ell^p$-norm of the distances from each point in $V$ to $S$ when the metric space is unweighted. We refer to the problem of $k$-clustering with respect to objective $\cl_p(\cdot)^p$ as the \emph{unnormalized $(k,p)$-clustering problem}. The following lemma relates the approximation ratio of solutions to these two clustering problems.

\begin{lemma}\label{lem:unnorm to norm}
    Let $p \in \mathbb N$. A set $S \subseteq V$ is an $\alpha$-approximation to the $(k,p)$-clustering problem if and only if it is an $\alpha^p$-approximation to the unnormalized $(k,p)$-clustering problem.
\end{lemma}

\begin{proof}
    Since the optimal solutions with respect to clustering objectives $\cl_p(\cdot)$ and $\cl_p(\cdot)^p$ coincide, let $S^\star$ be an optimal solution to the $(k,p)$-clustering and the unnormalized $(k,p)$-clustering problems. It follows that $\cl_p(S) \leq \alpha \cdot \cl_p(S^\star)$ if and only if $\cl_p(S)^p \leq \alpha^p \cdot \cl_p(S^\star)^p$ for any $\alpha \geq 1$.
\end{proof}

\noindent
We note that the $(k,1)$-clustering and $(k, \infty)$-clustering problems correspond to the $k$-median and $k$-center problems respectively, and that the unnormalized $(k,2)$-clustering problem corresponds to the $k$-means problem.
Furthermore, using standard inequalities, it can be shown that \emph{as long as the metric space has unit weights} (i.e.~$w(x) = 1$ for all $x \in V$) $\cl_\infty(S) \leq \cl_p(S) \leq O(1) \cdot \cl_\infty(S)$ for any $p \geq \log n$.\footnote{Note that the $(k,\infty)$-clustering/$k$-center objective does not depend on the weights of points. Thus, when dealing with these objectives, we can assume the metric space has unit weights.} It follows that any $\alpha$-approximate solution to the $(k,\log n)$-clustering problem is a $O(\alpha)$-approximate solution to the $k$-center problem.

\subsection{Dynamic $k$-Clustering}

In the dynamic setting, we are given a metric space $(V, w, d)$ and a stream of \emph{updates} $\sigma_1, \dots, \sigma_T$ corresponding to (weighted) point insertions and deletions from the metric space $(V, w, d)$.
We assume that, at each point in time, $(V,d)$ is a metric subspace of some fixed (possibly infinite) underlying metric space $(\mathcal V, d)$. 
Thus, the metric space is defined by the subset of points $V \subseteq \mathcal V$. 
Furthermore, we assume constant time query access to the distances between pairs of points in $\mathcal V$.

Our objective is to maintain a set $S \subseteq V$ of size at most $k$ that minimizes $\cl_p(S,V,w)$, i.e.~a good solution to the $(k,p)$-clustering problem in the current metric space $(V,w,d)$. 
We will also be concerned with the \emph{amortized recourse} of our solution and the \emph{amortized update time} of our algorithm. We define the \emph{total} recourse to be the number of changes that we make to the solution $S$ throughout the sequence of updates. More precisely, we define it as
$$\sum_{t = 1}^{T} \left|S^{t-1} \oplus S^{t}\right|$$
where $S^{t}$ denotes the set $S$ at time $t$, i.e.~after the $t^{th}$ update. The amortized recourse is the total recourse divided by the length of the update sequence, $T$. The total update time of a dynamic clustering algorithm is the total time the algorithm spends handling updates, and the amortized update time is the total update time divided by $T$.

\subsection{A Projection Lemma}
A key component in many of our proofs throughout this paper is the following well-known \emph{projection lemma} \cite{arxivGT2008, ChrobakKY06}. Given sets $X, S \subseteq \mathcal V$, this simple but very useful lemma bounds the cost of the solution $S' \subseteq X$ obtained by \emph{projecting} $S$ onto the set $X$. More precisely, let $\Proj(S,X)$ the set of points that, for each $y \in S$, contains the point $x \in X$ that is closest to $y$, i.e.~$\Proj(S,X) = \{y \in S \, | \, \arg \min_{x \in X} d(x, y) \}$.\footnote{The set $\Proj(S,X)$ is not necessarily unique. Thus, when computing $\arg \min_{x \in X} d(x, y)$ for some $y \in S$, we assume that we break ties arbitrarily and consistently so that $\Proj(S,X)$ is well defined.} Then we have the following lemma.

\begin{lemma}[Projection Lemma]\label{lem:proj lemma}
    For all $x \in \mathcal V$, $d(x, \Proj(S,X)) \leq d(x, X) + 2d(x, S)$.
\end{lemma}

\begin{proof}
    Let $S'$ denote the set $\Proj(S,X)$. Let $x \in \mathcal V$, $y$ and $y^\star$ be the points that are closest to $x$ in $X$ and $S$ respectively, and $y'$ be the point in $S'$ that is closest to $y^\star$.
    Then we have that
    $$ d(x,S') \leq d(x,y') \leq d(x,y^\star) + d(y^\star, y') \leq d(x,y^\star) + d(y^\star, y) \leq d(x, y) + 2d(x,y^\star).$$
\end{proof}

\noindent
The following lemma allows us to use \Cref{lem:proj lemma} to bound the cost of the projected solution.
\begin{lemma}\label{lem:minkow}
    For all $p \in \mathbb N$, $A, B \subseteq \mathcal V$ and $\alpha, \beta \geq 0$, we have that
    $$ \sum_{x \in V} w(x) (\alpha d(x,A) + \beta d(x, B))^p \leq \left( \alpha \cl_p(A) + \beta \cl_p(B) \right)^p. $$
\end{lemma}

\begin{proof}
    This follows immediately from applying Minkowski's inequality to the left-hand side.
\end{proof}



\noindent
As an immediate corollary of \Cref{lem:minkow}, we obtain the following bound on the cost of the projected solution for all $p \in \mathbb N \cup \{ \infty\}$.

\begin{corollary}\label{lem:proj cor}
    $\cl_p(\Proj(S,X)) \leq \cl_p(X) + 2 \cl_p(S)$.
\end{corollary}


\subsection{Proper vs. Improper Solutions}\label{sec: proper improper def}
Even though we define a solution $S$ to the $(k,p)$-clustering problem in the metric space $(V,w,d)$ to be a subset of $V$, it will be useful to consider solutions that are not necessarily subsets of $V$.
Thus, we define an \emph{improper} (resp.~\emph{proper}) solution to the $(k,p)$-clustering problem in the metric space $(V,w,d)$ as a set $S \subseteq \mathcal V$ (resp.~$S \subseteq V$) of size at most $k$.

Our main dynamic algorithm, which we describe in \Cref{sec:consistency}, maintains an improper solution $S$ to the $(k,p)$-clustering problem. In \Cref{sec:proper}, we show how to dynamically maintain $\Proj(S, V)$ (the projection of $S$ onto the current metric space $V$), which is a subset of $V$ and hence a proper solution. By \Cref{lem:proj cor}, we have that $\cl_p(\Proj(S, V)) \leq 2 \cl_p(S)$.
Thus, we only incur a small constant loss in approximation as we convert our improper solution to a proper solution.

We note that, whenever we use the notation $\Opt_k(\cdot)$ in \Cref{part:I}, we always refer to the cost of the optimal \emph{proper} solution. 

\subsection{Restricted $k$-Clustering}\label{sec:restricted def}
Let $X \subseteq \mathcal V$ be some subset in the underlying metric space $(\mathcal V,d)$.
We will often need to find good solutions to the $(k,p)$-clustering problem in the space $(V,w,d)$ while also enforcing that the solutions we obtain are a subset of $X$. We refer to such solutions as $X$-restricted.
More precisely, we define an $(\alpha, \beta, X)$-restricted solution $S$ to the $(k,p)$-clustering problem to be a set $S \subseteq X$ of size $k$ such that $\cl_p(S, V, w) \leq \alpha \cdot \cl_p(X, V, w) + \beta \cdot \Opt^p_k(V,w)$. Note that $X$ might not necessarily be a subset of $V$, so $X$-restricted solutions are not necessarily proper (see \Cref{sec: proper improper def}).

\begin{lemma}\label{lem:restricted}
    For all $X \subseteq \mathcal V$, there is a $(1, 2, X)$-restricted solution to the $(k,p)$-clustering problem.
\end{lemma}

\begin{proof}
    Follows immediately from applying \Cref{lem:proj cor} with the set $X$ and an optimal solution $S^\star$ to the $(k,p)$-clustering problem.
\end{proof}

\subsection{Organization of \Cref{part:I}}\label{sec:P2 org}

In \Cref{part:I}, we describe our dynamic algorithm for $k$-clustering and prove the following theorem.

\begin{theorem}
    There exists a
    fully dynamic algorithm that can maintain a $O(p/\epsilon)$-approximate solution to the $(k,p)$-clustering problem with $\tilde O(k^\epsilon/\epsilon)$ amortized recourse and $\tilde O(k^{1 + \epsilon}p/\epsilon^3)$ amortized update time, for any $\epsilon > 1/\log k$.
\end{theorem}

\noindent In \Cref{sec:consistency}, we describe our algorithm for the consistent clustering problem. In \Cref{sec:RLS}, we describe our randomized local search algorithm. In \Cref{sec:imp consistency}, we show how to use the randomized local search in order to give an efficient implementation of our consistent clustering algorithm. In \Cref{sec:proper}, we show how to efficiently convert the improper solution maintained by our algorithm into a proper solution. Finally, in \Cref{app:sparsification}, we describe the dynamic algorithm that we use in order to sparsify the input to our algorithm.

\noindent

\section{The Consistency Hierarchy}\label{sec:consistency}

In this section, we describe and analyze our algorithm for the consistent clustering problem. In particular, we prove the following theorem.

\begin{theorem}
There exists a deterministic fully dynamic algorithm that can maintain a $O(1/\epsilon)$-approximate solution to the $(k,p)$-clustering problem with $\tilde O(k^\epsilon/\epsilon)$ amortized recourse, for any $\epsilon > 1/ \log k$ and any $p \in \mathbb N \cup \{ \infty\}$.
\end{theorem}
\noindent
This algorithm works by dynamically maintaining a hierarchy of nested solutions to the $(k,p)$-clustering problem, which it reconstructs periodically in a way that leads to good amortized recourse. For now, we describe our algorithm without dealing with data structures or implementation details and only analyse the recourse and approximation ratio of the algorithm.


\subsection{Our Dynamic Algorithm}





Let $\epsilon > 0$, $k \in \mathbb N$ and $p \in \mathbb N \cup \{\infty\}$ be parameters such that $\ell := 1 /\epsilon \in \mathbb N$ and $\epsilon \geq 1/\log k$. The parameters $\alpha, \beta \geq 1$ depend on how our algorithm constructs restricted solutions (see \Cref{sec:approx}).
Let $s_i := \lfloor k^{(\ell - i)\epsilon} \rfloor$ be the \emph{slack} at \emph{layer} $i$.
We initialize our algorithm by calling \textsc{ReconstructFromLayer}$(0)$. 
After an update occurs, we first update the metric space $(V, w, d)$ by inserting or deleting the appropriate point.
We then proceed to call Algorithm~\ref{alg:update} which updates the sets $S_i$ and counters $\tau_i$ that are maintained by our algorithm. The pseudocode below gives the full formal description of the algorithm.

\begin{algorithm}[H]\label{alg:update}
    \SetAlgoLined
    \DontPrintSemicolon
    \texttt{// Perform lazy updates:} \;
    \If{\textnormal{a point $x$ as been inserted}}{
        		$S_i \leftarrow S_i + x$ for all $i$\;
        }
    \texttt{// Check if reconstruction is necessary:} \;
    \For{$i = 0 \dots \ell + 1$}{
    	$\tau_i \leftarrow \tau_i + 1$\;
        \If{$\tau_i > s_i$}{
        		\Return \textsc{ReconstructFromLayer}$(i)$\;
        }
    }
   \Return $S_{\ell + 1}$\;
    \caption{\textsc{UpdateHierarchy}}
\end{algorithm}

\begin{algorithm}[H]\label{alg:recon}
    \SetAlgoLined
    \DontPrintSemicolon
    \texttt{// Let $S_{-1}$ denote the set $V$} \;
    \For{$i = j \dots \ell + 1$}{
        Let $S_i$ be an $(\alpha, \beta, S_{i-1})$-restricted solution to the $(k + s_i, p)$-clustering problem in the metric space $(V,w,d)$\label{line:compute}\;
        $\tau_i \leftarrow 0$\;
    }
    \Return $S_{\ell + 1}$\;
    \caption{\textsc{ReconstructFromLayer}$(j)$}
\end{algorithm}

\medskip
\noindent\textbf{Lazy Updates.}
After an update, we lazily update each set $S_i$. 
When a point $x$ is inserted, we add $x$ to each $S_i$ for all $0 \leq i \leq \ell + 1$. When a point $x$ is deleted, we do not modify any of the sets $S_i$. Thus, the sets $S_i$ are not necessarily subsets of the current metric space $V$, since points in $S_i$ might have been deleted from $V$.

\medskip
\noindent\textbf{Computing $S_i$.}
We assume that the set $S_i$ computed in \Cref{line:compute} of Algorithm~\ref{alg:recon} satisfies the condition that $|S_{i}| = k + s_i$ as long as $|S_{i-1}| \geq k + s_i$. Otherwise, we set $S_i$ to be the whole set $S_{i-1}$. In other words, when computing an $X$-restricted solution to the $(k,p)$-clustering problem, we always return a solution of size exactly $k$---unless $X$ has size less than $k$, in which case we simply return all the points in $X$ as the solution.

\subsection{Basic Properties of the Hierarchy}\label{sec:structural}

We now describe some of the basic properties of the hierarchy.
We begin by observing that the sets maintained at each of the layers are nested.

\begin{observation}\label{lem:nested}
    We have that $V \supseteq S_0 \supseteq S_1 \supseteq \dots \supseteq S_{\ell + 1}$.
\end{observation}

\begin{proof}
    Let $0 \leq i \leq \ell + 1$. When the set $S_i$ is reconstructed, we clearly have that $S_i \subseteq S_{i -1}$. Lazily updating the sets $S_{i-1}$ and $S_{i}$ after an update maintains this invariant since lazy deletions have no effect on these sets and lazy insertions will add the same point to both.
\end{proof}

\noindent
The following lemma describes the structure of the sets $S_i$ maintained by the algorithm.

\begin{lemma}\label{lem:size Si}
    For all $0 \leq i \leq \ell + 1$, we have that $|S_i| \leq k + 2s_i$. Furthermore, either $|S_i| \geq k + s_i$ or $S_i = \dots = S_0 \supseteq V$.
\end{lemma}

\begin{proof}
To see that $|S_i| \leq k + 2 s_i$ for all $0 \leq i \leq \ell + 1$, note that $|S_i| \leq k + s_i$ every time $S_i$ is reconstructed and that its size can only increase by at most $1$ with each lazy update.
Since there are at most $s_i$ lazy updates before $S_i$ is reconstructed again, the bound follows.

We prove the rest by induction on $i$. For $i = 0$, consider the last time that $S_0$ was reconstructed. If $|V| < k + s_0 = 2k$ at this point, then the algorithm would have set $S_0 = V$. It follows from the lazy updates performed by our algorithm that $S_0 \supseteq V$ after each update. Otherwise, the algorithm sets $S_0$ to be a solution to the $2k$-median problem in the metric space $(V, w, d)$, where $|S_0| = 2k$. Since lazy updates cannot decrease the size of $S_0$, we have that $|S_0| \geq 2k$.

Now, let $1 \leq i \leq \ell + 1$. For the inductive step, suppose that this holds for all $0 \leq j < i$. Consider the last time that $S_i$ was reconstructed. If $|S_{i-1}| < k + s_i$ at this point, then the algorithm sets $S_i$ to be $S_{i-1}$. Since $|S_{i-1}| < k + s_i \leq k + s_{i-1}$, we have that $S_{i-1} = \dots = S_0 \supseteq V$ at this point.
Thus, by the lazy updates performed by our algorithm, we have that $S_i = \dots = S_0 \supseteq V$ after each update.
Otherwise, $|S_i| \geq k + s_i$ by an analogous argument to the base case.
\end{proof}

\noindent
Based on this lemma, we derive the following corollary, which describes the set $S_{\ell + 1}$, which our algorithm maintains as its output.

\begin{corollary}\label{cor:S_l size}
    We either have that $|S_{\ell + 1}| = k$, or $S_{\ell + 1} \supseteq V$ and $|V| < k$.
\end{corollary}
\begin{proof}
Since $k^{(\ell - (\ell + 1))\epsilon} = k^{-\epsilon} < 1$, we have that $s_{\ell + 1} = \lfloor k^{(\ell - (\ell + 1))\epsilon} \rfloor = 0$.
It follows from \Cref{lem:size Si} that $|S_{\ell + 1}| \leq k$. Furthermore, if $|S_{\ell + 1}| < k$, then we have that $S_{\ell + 1} \supseteq V$.
\end{proof}

\subsection{Recourse Analysis}\label{sec:recourse}

We now proceed to analyse the recourse of the sets maintained at each layer of the hierarchy.
Suppose that we run our algorithm on a sequence of updates $\sigma_1,\dots,\sigma_T$.
For all $0 \leq i \leq \ell + 1$, let $S_i^{t}$ denote the state of the set $S_i$ after the end of the $t^{th}$ update.
We now prove following lemma, which immediately implies the recourse bound of our algorithm.
\begin{lemma}\label{lem:rec bound}
    For all $0 \leq i \leq \ell + 1$, we have that
    $$\sum_{t=1}^T \left|S_{i}^{t-1} \oplus S_{i}^{t}\right| = O\!\left(\frac{k^\epsilon}{\epsilon} \right).$$
\end{lemma}
\noindent
Fix some $0 \leq i \leq \ell + 1$.
Let $r^{t} \in \{0,\dots,\ell + 1\}$ 
be the layer that we reconstruct from
while handling the $t^{th}$ update.\footnote{i.e. the value such that we call \textsc{ReconstructFromLayer}$(r^{t})$.} The following lemma bounds the recourse of the $t^{th}$ update.

\begin{lemma}\label{lem:rec bound t} For all $t \in [T]$, we have that
    $\left|S_{i}^{t-1} \oplus S_{i}^{t}\right| \leq 4 s_{r^{t} - 1}$.
\end{lemma}

\begin{proof}
    We first note that, since $s_{-1} \geq k$, this bound is trivial when $r^{t} = 0$. Furthermore, if $i < r^{t}$, then we have that $\left|S_{i}^{t-1} \oplus S_{i}^{t}\right| \leq 1$ by the definition of the lazy updates. Thus, we assume that $i \geq r^{t} > 0$.

    Applying \Cref{lem:nested}
    and noting that $S_{r^{t} - 1}$ is not reconstructed during the $t^{th}$ update, it follows that $S^{t-1}_{i}$ and $S^{t}_{i}$ are subsets of $S^{t}_{r^{t} - 1}$. 
    Now, suppose that both $S^{t-1}_{i}$ and $S^{t}_{i}$ have size at least $k + s_i$. Then, by \Cref{lem:size Si}, $\left|S^{t}_{r^{t} - 1} \right| \leq k + 2s_{r^{t} - 1}$, and hence $\left|S_{i}^{t-1} \oplus S_{i}^{t}\right| \leq 4s_{r^{t} - 1}$. Otherwise, one of $S^{t-1}_{i}$ and $S^{t}_{i}$ has size less than $k + s_i$. In this case, by \Cref{lem:size Si}, we have that $S^{t-1}_{i} = S^{t-1}_{r^{t} - 1}$ and $S^{t}_{i} = S^{t}_{r^{t} - 1}$, and so
    $$ \left|S_{i}^{t-1} \oplus S_{i}^{t}\right| = \left|S_{r^{t} - 1}^{t-1} \oplus S_{r^{t} - 1}^{t}\right| \leq 1. $$
    Since $s_{r^{t} - 1} \leq s_\ell = 1$, the bound follows.
\end{proof}

\noindent Applying \Cref{lem:rec bound t}, we can derive the following upper bound on the total recourse of the set $S_i$ maintained by our algorithm across the sequence of updates
$$ \sum_{t = 1}^T \left|S_{i}^{t-1} \oplus S_{i}^{t}\right| 
\leq 4 \cdot \sum_{t = 1}^T \left\lfloor k^{(\ell - r^{t} + 1)\epsilon} \right\rfloor
= O(1) \cdot \sum_{t = 1}^T \sum_{j = 0}^{\ell + 1} \1 \! \left[r^{t} = j\right] \cdot \left\lfloor k^{(\ell - j + 1)\epsilon} \right\rfloor$$
$$ \leq O(1) \cdot \sum_{j = 0}^{\ell + 1} k^{(\ell - j + 1)\epsilon} \cdot \sum_{t = 1}^T  \1 \! \left[r^{t} = j\right]. $$
For $0 \leq j \leq \ell$, we note that every time the set $S_{j}$ is reconstructed we initialize a counter that will then call $\textsc{ReconstructFromLayer}(j)$ after more than $\lfloor k^{(\ell - i) \epsilon} \rfloor$ updates have occurred---unless some set $S_{j'}$ for $j' <  j$ is reconstructed before this can happen and causes this counter to be reset to $0$. Hence, we get that 
$$\sum_{t = 1}^T  \1 \! \left[r^{t} = j\right] \leq \frac{T}{\left\lfloor k^{(\ell - j) \epsilon} \right\rfloor + 1}.$$
It follows that the amortized recourse of the set $S_i$ is upper bounded by
$$ O(1) \cdot \sum_{j = 0}^{\ell + 1} \frac{k^{(\ell - j + 1)\epsilon}}{\left\lfloor k^{(\ell - j) \epsilon} \right\rfloor + 1} \leq O(1) \cdot \sum_{j = 0}^{\ell + 1} k^\epsilon = O\!\left(\frac{k^\epsilon}{\epsilon} \right). $$






\subsection{Approximation Analysis}\label{sec:approx}

We now analyze the approximation ratio of our algorithm for the $(k,p)$-clustering problem.

Consider the state of our algorithm after handling some arbitrary sequence of updates. For all $0 \leq i,j \leq \ell + 1$, let $S_{i}^{j}$ and $V^{j}$ respectively denote the states of the sets $S_i$ and $V$ at the end of the last update during which the set $S_{j}$ was reconstructed. The following lemma captures the key invariant maintained by our dynamic algorithm.

\begin{lemma}\label{lem:key inv}
    For all $0 \leq i \leq \ell + 1$, we have that 
    \begin{equation}
        \cl \! \left(S_i^{i}, V^{i}\right) \leq \alpha \cdot \cl \! \left(S_{i-1}^{i-1}, V^{i-1}\right) + 2\beta \cdot \Opt_{k} \!\left(V \right).
    \end{equation}
\end{lemma}

\begin{proof}
    It follows immediately from our algorithm that
    \begin{equation}\label{eq:key inv}
        \cl \!\left(S_i^{i}, V^{i} \right) \leq \alpha \cdot \cl \! \left(S_{i-1}^{i}, V^{i} \right)+ \beta \cdot \Opt_{k + s_i} \! \left(V^{i} \right).
    \end{equation}
    The following claim shows that the cost of the solution $S_{i-1}$ can only decrease as we perform lazy updates, and is used to upper bound the first term in Equation~(\ref{eq:key inv}).
    \begin{claim}\label{claim:inv 1}
        $\cl \! \left(S_{i-1}^{i}, V^{i} \right) \leq \cl \! \left(S_{i-1}^{i-1}, V^{i-1}\right)$.
    \end{claim}
    \begin{proof}
        We first note that $\cl(S_{i-1}, V)$ cannot increase as we perform lazy updates. This follows from the fact that we do not remove any points from $S_{i-1}$ 
        and that every time a point is added to $V$ we also add it to $S_{i-1}$. Thus, the connection costs of points in $V$ cannot increase after a lazy update and points that are inserted into $V$ have a connection cost of $0$. By noting that $S_{i}$ is reconstructed whenever $S_{i-1}$ is reconstructed, the claim follows
    \end{proof}
    \noindent
    We use the following claim to upper bound the second term in Equation~(\ref{eq:key inv}).
    \begin{claim}\label{claim:inv 2}
        $\Opt_{k + s_i}\!\left (V^{i} \right ) \leq 2 \cdot \Opt_k\! \left (V \right )$.
    \end{claim}
    \begin{proof}
    Let $V^{i} \cup V$ denote the metric space obtained by taking the union of the (weighted) sets $V^{i}$ and $V$.
    We first note that $\Opt_{k + s_i} \!\left(V\right) \leq 2 \cdot \Opt_{k + s_i}\!\left(V^{i} \cup V\right)$. This follows from the fact that, for any $X \subseteq V^{i} \cup V$, there exists some $X' \subseteq V$ such that $|X'| \leq |X|$ and $d(x, X') \leq 2d(x, X)$ for all $x \in V$.
    Thus, we have that $\cl(X', V) \leq 2 \cl(X, V) \leq 2 \cl\!\left(X, V^{i} \cup V \right)$ and so $\Opt_{k + s_i}(V) \leq 2 \cdot \Opt_{k + s_i}\!\left(V^{i} \cup V\right)$.

    We now show that $\Opt_{k + s_i}\!\left(V^{i} \cup V\right) \leq \Opt_{k}(V)$. Note that, for any $X \subseteq V^{i}$, we have
    $$ \cl\!\left(X \cup \left(V \setminus V^{i}\right), V^{i} \cup V\right) = \cl\!\left(X \cup \left(V \setminus V^{i}\right), V^{i} \right) \leq \cl\!\left(X, V^{i}\right). $$
    Since $|V \setminus V^{i}| \leq s_i$, it follows that
    $$\Opt_{k + s_i}\! \left(V^{i} \cup V\right) = \min_{X \subseteq V^{i} \cup V, \, |X| \leq k + s_i} \cl\!\left(X, V^{i} \cup V \right) \leq \min_{X \subseteq V^{i}, \, |X| \leq k} \cl\!\left(X \cup \left(V \setminus V^{i}\right), V^{i} \cup V\right)$$
    $$ \leq \min_{X \subseteq V^{i}, \, |X| \leq k} \cl\!\left(X, V^{i}\right) = \Opt_{k}\!\left(V^{i}\right). $$
The claim follows from combining both inequalities.
    \end{proof}
    \noindent
    The lemma follows directly from combining the two preceding claims with Equation~(\ref{eq:key inv}).
\end{proof}

\noindent
We now use the key invariant described in \Cref{lem:key inv} in order to analyze the approximation ratio of our algorithm.

\begin{lemma}\label{lem:approx}
The solution $S_{\ell + 1}$ maintained by our algorithm satisfies
    $$\cl(S_{\ell + 1}, V) \leq \left( 2\beta \cdot \sum_{i = 0}^{\ell + 1} \alpha^i \right) \cdot \Opt_k(V).$$
\end{lemma}

\begin{proof}
    For each $0 \leq i \leq \ell + 1$, let $C_i$ denote the value of $\cl\!\left(S_i^{i}, V^{i}\right)$. By repeatedly applying \Cref{lem:key inv}, we get that, for all $0 \leq i \leq \ell + 1$
    $$ C_{\ell + 1} \leq \alpha^{i} \cdot C_{\ell + 1 - i} + \left( 2\beta \cdot \sum_{j = 0}^{i-1} \alpha^j \right) \cdot \Opt_k(V). $$
    Thus, setting $i = \ell + 1$ and noting that $C_{0} \leq 2\beta \cdot \Opt_k(V)$, we get that 
    $$C_{\ell + 1} \leq \left( 2\beta \cdot \sum_{i = 0}^{\ell + 1} \alpha^i \right) \cdot \Opt_k(V).$$
    Finally, since the solution $S_{\ell + 1}$ is reconstructed after every update, we have that $\cl(S_{\ell + 1}, V) = C_{\ell + 1}$ and the claim follows.
\end{proof}

\noindent
By plugging in the appropriate values of $\alpha$ and $\beta$ obtained from \Cref{lem:restricted}, we get the following theorem which summarises the behaviour of our consistent clustering algorithm.

\begin{theorem}\label{thm:main}
There exists a deterministic fully dynamic algorithm that can maintain a $O(1/\epsilon)$-approximate solution to the $(k,p)$-clustering problem with $\tilde O(k^\epsilon/\epsilon)$ amortized recourse, for any $\epsilon > 1/ \log k$ and any $p \in \mathbb N \cup \{ \infty\}$.
\end{theorem}

\section{Randomized Local Search}\label{sec:RLS}

In this section, we describe a new randomized variant of the local search algorithm for the $(k,p)$-clustering problem \cite{AryaGKMMP04, arxivGT2008}, which we then use to to give an efficient implementation of our dynamic algorithm from \Cref{sec:consistency}. Algorithm~\ref{alg:local fast full} contains the pseudocode for this algorithm, and its running time and approximation guarantees are summarized in Lemmas~\ref{lem: supp local 2} and \ref{lem: supp local 1} respectively.

\subsection{Algorithm Description}

Let $(V,w,d)$ be a metric space of size $n$ with aspect ratio $\Delta$, $X \subseteq V$ be a subset of size $m$, $k \leq m$ be an integer, $0 < \epsilon \leq 1/2$, $s = m - k$, and $c \geq 1$. Throughout this section, we work with the $(k,p)$-clustering objective for some fixed $p \in \mathbb N$ and abbreviate $\cl_p(\cdot)$ by $\cl(\cdot)$ and $\Opt_k^p$ by $\Opt_k$ accordingly.

The classic $k$-median algorithm of \cite{AryaGKMMP04} works by starting with an arbitrary set $S$ of size $k$ and repeatedly performing iterations of local search. At each iteration, the algorithm checks all possible swaps and attempts to find one that improves the cost of the solution by a multiplicative $1 - \epsilon/s$ factor every iteration.
As long as the current set $S$ is not already a good approximation, it can be shown that such a swap must necessarily exist.
In order to improve the running time of the search, our \emph{randomized local search} algorithm works by randomly choosing a point $x \in X \setminus S$ and performing the `best' swap involving the point $x$ instead of searching through all possible swaps. Note that, since we want to return an $X$-restricted solution, we start with a set $S \subseteq X$ and do not consider swaps involving points outside of $X$. By arguing that the cost of the solution improves by a multiplicative $1 - \epsilon/s$ factor \emph{in expectation}, we can significantly speed up the running time of the algorithm. By then incorporating parts of the analysis of \cite{arxivGT2008}, which extends the analysis of \cite{AryaGKMMP04} to the $(k,p)$-clustering problem, we can show that our algorithm efficiently produces $X$-restricted solutions to the $(k,p)$-clustering problem.
The pseudocode in Algorithm~\ref{alg:local fast full} gives a formal description of our randomized local search algorithm.

\begin{algorithm}[H]\label{alg:local fast full}
    \SetAlgoLined
    \DontPrintSemicolon
    Let $S \subseteq X$ be any set of size $k$\;
    \For{$2c \cdot ps \log(n) \log(n\Delta/\epsilon)/\epsilon^2$ \textup{\textbf{iterations}}}{
        Sample $x \sim X \setminus S$ independently and u.a.r.\;
        $y \leftarrow \textsf{BestSwap}(S, x)$\;
        $S \leftarrow S + x - y$\;
    }
    \Return{$S$}
    \caption{\RandLoc$(S)$}
\end{algorithm}

\medskip
\noindent\textbf{The \textnormal{\textsf{BestSwap}} Oracle.}
We assume that our algorithm has access to an oracle $\textsf{BestSwap}(S, x)$. As input, this oracle takes the current set $S$ and some point $x \in V$, and returns the point $y \in S + x$ that minimizes the value of $\cl(S + x - y)^p$.
We later show how to implement this oracle to handle queries in $\tilde O(n)$ time.

\medskip
\noindent\textbf{Properties of Algorithm~\ref{alg:local fast full}.}
The following two lemmas describe the approximation and running time guarantees of the randomized local search algorithm. We prove these lemmas in Sections~\ref{sec:anal fast} and \ref{sec:supp local 1} respectively.

\begin{lemma}\label{lem: supp local 2}
    Algorithm~\ref{alg:local fast full} returns a solution $S$ such that
    $$\cl(S) \leq (1 + 7\epsilon) \cdot (\cl(X) + 6p \cdot \Opt_k(V))$$
    with probability at least $1 - 1/n^{c}$.
\end{lemma}

\begin{lemma}\label{lem: supp local 1}
    Given the appropriate auxiliary data structures (see Section~\ref{sec:supp local 1}), Algorithm~\ref{alg:local fast full} can be implemented to run in $\tilde O(nps/\epsilon^2)$ time.
\end{lemma}

\subsection{Analysis of \textnormal{\RandLoc} (Proof of \Cref{lem: supp local 2})}\label{sec:anal fast}

Given the set $S \subseteq X$, $y \in S$, and $x \notin S$ we define 
$$\delta_S(y,x) := \cl(S+x-y)^p - \cl(S)^p.$$
That is, $\delta_S(y,x)$ is the change in $\cl(S)^p$ caused by swapping the point $x$ with the center $y$.
We let $\delta_S(x) = \min_{y \in S+x} \delta_S(y,x)$.
We begin with the following lemma, which follows from a modified version of the analysis of the local search algorithm in \cite{arxivGT2008}.

\begin{lemma}\label{lem:local search core}
    Let $S$ be a solution to the $(k,p)$-clustering problem such that $\cl(S) > ( \cl(X) + 6p \cdot \Opt_k) / (1 - \epsilon)$. Then we have that
    $$ \sum_{x \in X \setminus S} \delta_S(x) \leq - \epsilon  \cdot\cl(S)^p.$$
\end{lemma}

\begin{proof}
Let $S$ be a subset of $X$ of size $k$ and let $S^\star$ denote an optimal solution to the $(k,p)$-clustering problem in the metric space $V$. For any set $A \subseteq V$ and $x \in V$, let $\pi_A(x)$ denote the point $y \in A$ that is closest to $x$, i.e.~the projection of the point $x$ onto $A$.\footnote{If there are multiple closest points to $x$ in $A$, we break ties arbitrarily but consistently.} Finally, the $\pi = \pi_X \circ \pi_{S^\star}$.

We have the following useful claim, which follows \Cref{lem:proj lemma}.

\begin{claim}\label{lem:1}
    For all $x \in V$, $A \subseteq V$, we have that $d(x, (\pi_A \circ \pi_{S^\star})(x)) \leq d(x, A) + 2d(x, S^\star)$.
\end{claim}

\begin{proof}
    Let $y^\star = \pi_{S^\star}(x)$ and let $y = \pi_A(x)$ be the closest point to $x$. Then we have that
    $$d(x, \pi(x)) \leq d(\pi_A(y^\star), y^\star) + d(y^\star, x) \leq d(y, y^\star) + d(y^\star, x) \leq d(y, x) + 2d(y^\star, x).$$
\end{proof}

\medskip
\noindent\textbf{Constructing a Set of Swaps.}
For all $y \in S$, define $S^\star_y := \pi_S^{-1}(y) \cap S^\star$.
Note that the $S^\star_y$ are disjoint for distinct $y$.
Let
$$S_0 := \{y \in S \; | \; |S^\star_y| = 0 \}, \qquad S_1 := \{y \in S \; | \; |S^\star_y| = 1 \}, \qquad S_{\geq 2} := \{y \in S \; | \; |S^\star_y| > 1 \}.$$
Let $S_1^{\star} := \{ y^{\star} \in S^{\star} \mid \pi_S(y^{\star}) \in S_1\}$. Observe that $S^{\star}_1 \subseteq S^{\star}$, $|S| = |S^{\star}| = k$ and $|S_1| = |S^{\star}_1|$. Thus, we get $|S^{\star} \setminus S^{\star}_1| = |S_0| + |S_{\geq 2}|$. Since $\pi_S$ assigns at least two points from $S^{\star} \setminus S^{\star}_1$ to each point in $S_{\geq 2}$, we also have $|S_{\geq 2}| \leq (1/2) \cdot |S^{\star} \setminus S^{\star}_1|$, and hence $|S_0| \geq (1/2)  \cdot |S^{\star} \setminus S^{\star}_1|$. This observation implies that we can construct a mapping $\sigma : S^{\star} \longrightarrow S$ which satisfies the following properties.
\begin{eqnarray}
\label{MAIN:eq:prop:1:new}
\sigma(y^{\star}) & = & \pi_S(y^{\star}) \in S_1 \  \text{ for all } y^{\star} \in S^{\star}_1. \\
\label{MAIN:eq:prop:2:new}
\sigma(y^{\star}) & \in & S_0 \qquad  \qquad \ \, \text{ for all } y^{\star} \in S^{\star} \setminus S^{\star}_1. \\
\label{MAIN:eq:prop:3:new}
|\sigma^{-1}(y)| & \leq & 2 \qquad \qquad \ \ \ \text{ for all } y \in S.
\end{eqnarray}
W.l.o.g., suppose that $S^{\star} = \{y^{\star}_1, \ldots, y^{\star}_k\}$; and $y_i = \sigma(y^{\star}_i) \in S$  and $y'_i = \pi_X(y^{\star}_i)$ for all $i \in [k]$. Consider the following collection of local \emph{moves} $\mathcal{M} := \{\langle y_1, y'_1 \rangle, \ldots, \langle y_k, y'_k \rangle \}$. 

\begin{corollary}
    \label{MAIN:cor:swap:new}
    For any two distinct indices $i, j \in [k]$, we have $\pi_S(y^{\star}_j) \neq y_i$. 
\end{corollary}

\begin{proof}
Since $\sigma(y^{\star}_i) = y_i$, from~(\ref{MAIN:eq:prop:1:new}) and~(\ref{MAIN:eq:prop:2:new}) we infer that $y_i \in S_0 \cup S_1$. If $y_i \in S_1$, then $y^{\star}_i$ is the unique point  $y^{\star} \in S^{\star}$ with $\pi_S(y^{\star}) = y_i$, and hence $\pi_S(y^{\star}_j) \neq y_i$. Otherwise, if $y_i \in S_0$, then there is no point $y^{\star} \in S^{\star}$ with $\pi_S(y^{\star}) = y_i$, and hence $\pi_S(y^{\star}_j) \neq y_i$.
\end{proof}

\medskip
\noindent\textbf{Bounding the Cost of Swaps.}
We now bound the sum of the costs of the swaps considered by the algorithm.
We do this by considering the moves in $\mathcal M$.
Consider the swap $\langle y_i, y_i' \rangle \in \mathcal M$.
We upper bound the change in the cost of the solution $S$ caused by swapping $y_i$ with $y'_i$ (i.e.~removing $y_i$ and adding $y'_i$ to $S$) as follows:

\begin{enumerate}
    \item \textit{For each point $x \in \pi^{-1}_{S^\star}(y_i^\star)$, reassign $x$ to $y'_i$}
    \item \textit{For each point $x \in \pi_S^{-1}(y_i) \setminus \pi_{S^\star}^{-1}(y^\star_i)$, reassign $x$ to $(\pi_S \circ \pi_{S^\star})(x)$}
\end{enumerate}
\Cref{MAIN:cor:swap:new} guarantees that, for each point $x \in \pi_S^{-1}(y_i) \setminus \pi_{S^\star}^{-1}(y^\star_i)$, $(\pi_S \circ \pi_{S^\star})(x) \in S \setminus \{y_i\}$. Thus, this reassignment is valid.
This yields the following upper bound on $\delta_S(y_i, y'_i)$ for each $\langle y_i, y'_i \rangle \in \mathcal M$,
\begin{multline}\label{eq:ineq 1}
    \delta_S(y_i, y'_i) \leq  \sum_{x \in \pi^{-1}_{S^\star}(y_i^\star)} w(x)(d(x, (\pi_X \circ \pi_{S^\star})(x))^p - d(x,S)^p) \\ + \sum_{x \in \pi_S^{-1}(y_i) \setminus \pi_{S^\star}^{-1}(y^\star_i)} w(x) (d(x, (\pi_S \circ \pi_{S^\star})(x))^p - d(x,S)^p).
\end{multline}
Summing over all the moves in $\mathcal M$, we get that
\begin{multline}\label{eq:ineq 2}
    \sum_{i = 1}^k \delta_S(y_i,y'_i) \leq  \sum_{i = 1}^k \sum_{x \in \pi^{-1}_{S^\star}(y_i^\star)} w(x)(d(x, (\pi_X \circ \pi_{S^\star})(x))^p - d(x,S)^p) \\ + \sum_{i = 1}^k \sum_{x \in \pi_S^{-1}(y_i) \setminus \pi_{S^\star}^{-1}(y^\star_i)} w(x) (d(x, (\pi_S \circ \pi_{S^\star})(x))^p - d(x,S)^p).
\end{multline}
Since each $x \in V$ appears in exactly one $\{\pi^{-1}_{S^\star}(y_i^\star)\}_{i \in [k]}$, the first sum on the RHS of Equation (\ref{eq:ineq 2}) is precisely
$$\sum_{x \in V} w(x)(d(x, (\pi_X \circ \pi_{S^\star})(x))^p - d(x,S)^p).$$
Similarly, by noting that $\pi_S^{-1}(y_i) \setminus \pi_{S^\star}^{-1}(y^\star_i) \subseteq \pi_S^{-1}(y_i)$, that
each $y \in S$ appears in at most $2$ moves in $\mathcal M$, and that each $x \in V$ appears in exactly one $\{\pi_S^{-1}(y)\}_{y \in S}$, the second sum on the RHS of Equation (\ref{eq:ineq 2}) is at most
$$2 \cdot \sum_{x \in V} w(x)(d(x, (\pi_S \circ \pi_{S^\star})(x))^p - d(x,S)^p).$$
By applying \Cref{lem:1}, we get that $d(x, (\pi_X \circ \pi_{S^\star})(x)) \leq d(x,X) + 2d(x,S^\star)$ and similarly that
$d(x, (\pi_S \circ \pi_{S^\star})(x)) \leq d(x,S) + 2d(x,S^\star)$. By then applying \Cref{lem:minkow}, it follows that
\begin{equation}\label{eq:minkow 1}
    \sum_{x \in V} w(x)d(x, (\pi_X \circ \pi_{S^\star})(x))^p  \leq (\cl(X) + 2 \cl(S^\star))^p,
\end{equation}
and similarly,
\begin{equation}\label{eq:minkow 2}
    \sum_{x \in V} w(x)d(x, (\pi_S \circ \pi_{S^\star})(x))^p  \leq (\cl(S) + 2 \cl(S^\star))^p.
\end{equation}
Putting everything together, we get that
\begin{equation}\label{eq:final}
\sum_{i=1}^k \delta_S(y_i,y'_i) \leq (\cl(X) + 2 \cl(S^\star))^p + 2(\cl(S) + 2 \cl(S^\star))^p - 3\cl(S)^p.
\end{equation}
We now have the following claim, which we prove below.
\begin{claim}\label{cl:headache}
    $(\cl(X) + 2 \cl(S^\star))^p + 2(\cl(S) + 2 \cl(S^\star))^p - 3\cl(S)^p \leq -\epsilon \cdot \cl(S)^p$.
\end{claim}
\noindent
Applying \Cref{cl:headache} to Equation~(\ref{eq:final}), we get that
$$ \sum_{i=1}^k \delta_S(y_i,y'_i) \leq - \epsilon \cdot \cl(S)^p. $$
Clearly, $\delta_S(x) \leq 0$ for any $x \in V$ and $\delta_S(x) = 0$ for any $x \in S$. It follows that
$$ \sum_{x \in X \setminus S} \delta_S(x) \leq \sum_{y' \in \pi_X(S^\star)} \delta_S(y') \leq \sum_{i=1}^k \delta_S(y_i,y'_i) \leq - \epsilon \cdot \cl(S)^p$$
which concludes the proof.

\begin{proof}[Proof of \Cref{cl:headache}.]
Recall that from the statement of Lemma~\ref{lem: supp local 2} we assume that
$\cl(S) > (\cl(X) + 6p \cdot \cl(S^\star)/(1-\epsilon)$. We want to show that 
$$(\cl(X) + 2 \cl(S^\star))^p + 2(\cl(S) + 2 \cl(S^\star))^p - 3\cl(S)^p \le -\epsilon\cdot\cl(S)^p.$$
Consider the expression,
$$ \left(\frac{\cl(X)+2\cl(S^\star)}{\cl(S)}\right)^p + 2 \left(1+\frac{2\cl(S^\star)}{\cl(S)}\right)^p $$
$$ \le \left(\frac{(1-\epsilon)(\cl(X)+2\cl(S^\star)}{\cl(X) + 6p \cdot \cl(S^\star)}\right)^p + 2\left(1+\frac{2(1-\epsilon)\cl(S^\star)}{\cl(X) + 6p \cdot \cl(S^\star)}\right)^p $$
$$ = \left(\frac{(1-\epsilon)(1+2x)}{1+6px}\right)^p + 2\left(1+\frac{2(1-\epsilon)x}{1+6px}\right)^p, $$
where $x=\cl(S^\star)/\cl(X)$. It suffices to show that this expression is at most $(3-\epsilon)$, which in turn can be established by showing  that
\begin{equation}\label{eq:poly}
    (1-\epsilon)^p(1+2x)^p+2(1+6px+2(1-\epsilon)x)^p \le (3-\epsilon)(1+6px)^p
\end{equation}
holds for all $x\ge 0$ and $p\ge 1$. 
Equation~(\ref{eq:poly}) trivially holds for $p=1$. We now show that it holds when $p\ge 2$. We do this by arguing that, for every integer $i$ such that $0\le i\le p$, the coefficient of $x^i$ on the LHS of Equation~(\ref{eq:poly}) is at most the coefficient of $x^i$ on the RHS:
$$ (1-\epsilon)^p2^i+2(6p+2(1-\epsilon))^i \leq (3-\epsilon)(6p)^i. $$
Note that we have factored out the term $\binom{p}{i}$, which appears in each coefficient.
It is easy to check that, for $i=0$ and $i=1$, the inequality above holds.
We can rewrite this inequality as 
$$(1-\epsilon)^p \cdot \left(\frac{1}{3p}\right)^i + 2 \left(1 + \frac{1-\epsilon}{3p} \right)^i\le 3-\epsilon$$ 
and observe that for $i,p\ge 2$, $(1/3p)^i\le 1/36$. Since $i\le p$, $(1+(1-\epsilon)/3p)^i\le (1+(1-\epsilon)/3p)^p\le e^{(1-\epsilon)/3}$.
We now need to identify the range of $\epsilon$ for which $(1-\epsilon)^p/36+2e^{(1-\epsilon)/3} < 3-\epsilon$. Since $p\ge 1$ and $0\le \epsilon\le 1$, we have $(1-\epsilon)^p\le 1-\epsilon$. By considering the Taylor series expansion of $e^x$, one can argue that $e^x\le 1+5x/4$ for $x\le 1/3$. Thus,
\begin{align*}
\frac{(1-\epsilon)^p}{36}+2e^{(1-\epsilon)/3} \le \frac{1-\epsilon}{36}+ 2 \left(1+ \frac{5(1 - \epsilon)}{12} \right) = 2+ \frac{31(1-\epsilon)}{36} \le 2+(1-\epsilon)=3-\epsilon,
\end{align*}
for all $\epsilon< 1$.

\end{proof}
\end{proof}

\noindent Now, suppose we run Algorithm~\ref{alg:local fast full} and
let $S_i$ denote the set $S$ at the end of the $i^{th}$ iteration of the algorithm.
We now use \Cref{lem:local search core} to derive a recurrence on the expected costs of the sets $S_i$. For notational convenience, let $C_i$ and $C_i^p$ denote $\cl(S_i)$ and $\cl(S_i)^p$ respectively, and let $\Psi := \cl(X) + 6p \cdot \Opt_k$. Note that the $C_i$ are random variables.
\begin{lemma}\label{lem:rec}
    For all $0 \leq j < i \leq T$, we have that 
    $$\mathbb E[C_i^p] \leq \left(1 - \frac{\epsilon}{s} \right)^{i - j} \cdot \mathbb E[C_{j}^p] + \left( \frac{\Psi}{1 - \epsilon} \right)^p.$$
\end{lemma}
\begin{proof}
    Let $\mathcal E_\ell$ denote the event that $C_\ell \leq \Psi/ (1 - \epsilon)$.
    Now, suppose that we fix all of the random bits used by the algorithm during the first $i-1$ iterations so that the event $\mathcal E_{i-1}$ does not occur. Then $C_{i-1}$ is not a random variable w.r.t.~the remaining random bits and $C_{i-1} > \Psi / (1 - \epsilon)$. 
    Conditioned on these fixed random bits, it follows from \Cref{lem:local search core} that
    $$ \mathbb E [ C_i^p ] = \mathop{\mathbb{E}}_{x \sim X \setminus S_{i-1}} \left[ \delta_{S_{i-1}}(x) \right] + C_{i-1}^p = \frac{1}{s} \cdot \sum_{x \in X \setminus S_{i-1}} \delta_{S_{i-1}}(x) + C_{i-1}^p $$
    $$ \leq - \frac{\epsilon}{s} \cdot  C_{i-1}^p - C_{i-1}^p = \left(1 - \frac{\epsilon}{s} \right) \cdot C_{i-1}^p. $$
    Taking expectation on both sides, we get that $\mathbb E [C_i^p] \leq (1 - \epsilon / s) \cdot \mathbb E [C_{i-1}^p]$. Since this holds as long as the random bits in the first $i-1$ rounds are fixed so that the event $\mathcal E_{i-1}$ does not occur, it follows that
    $$\mathbb E [C_i^p | \overline{\mathcal E_{i-1}} ] \leq \left(1 - \frac{\epsilon}{s} \right) \cdot \mathbb E [C_{i-1}^p | \overline{\mathcal E_{i-1}} ],$$
    where the randomness is now being taken over all the random bits used by the algorithm. We can now observe that
    \begin{align*}
        \mathbb E [ C_i^p ] &= \mathbb E[C_i^p | \mathcal E_{i-1} ] \cdot \Pr[\mathcal E_{i-1}] + \mathbb E[C_i^p |  \overline{\mathcal E_{i-1}} ] \cdot \Pr[ \overline{\mathcal E_{i-1}} ]\\
        &\leq \mathbb E[C_{i-1}^p | \mathcal E_{i-1} ] \cdot \Pr[\mathcal E_{i-1}] + \left(1 - \frac{\epsilon}{s} \right) \cdot \mathbb E [C_{i-1}^p |  \overline{\mathcal E_{i-1}} ] \cdot \Pr[ \overline{\mathcal E_{i-1}} ]\\
        &= \left(1 - \frac{\epsilon}{s} \right) \cdot \mathbb E [C_{i-1}^p] + \frac{\epsilon}{s} \cdot \mathbb E[C_{i-1}^p | \mathcal E_{i-1} ] \cdot \Pr[\mathcal E_{i-1}]\\
        &\leq\left(1 - \frac{\epsilon}{s} \right) \cdot \mathbb E [C_{i-1}^p] + \frac{\epsilon}{s} \cdot \left(\frac{\Psi}{1-\epsilon}\right)^p,
    \end{align*}
    where in the first inequality we further used that the cost of $C_{i}^p$ is decreasing with $i$, and in the last inequality we used the conditioning on $\mathcal E_{i-1}$. Hence, we just showed
    \begin{equation}\label{eq:rec 1}
        \mathbb E[C_i^p] \leq \left(1 - \frac{\epsilon}{s} \right) \cdot \mathbb E[C_{i-1}^p] + \frac{\epsilon}{s} \cdot \left( \frac{\Psi}{1 - \epsilon}\right)^p.
    \end{equation}
    Repeatedly applying Equation~(\ref{eq:rec 1}), we obtain
    $$ \mathbb E[C_i^p] \leq \left(1 - \frac{\epsilon}{s} \right)^{i-j} \cdot \mathbb E[C_{j}^p] + \frac{\epsilon}{s} \cdot \sum_{\ell = 0}^{i - j -1} \left(1 - \frac{\epsilon}{s} \right)^{\ell} \cdot \left( \frac{\Psi}{1 - \epsilon} \right)^p \leq \left(1 - \frac{\epsilon}{s} \right)^{i - j} \cdot \mathbb E[C_{j}^p] + \left(\frac{\Psi}{1 - \epsilon}\right)^p $$
    and the lemma follows.
\end{proof}

\noindent
Now, let $\eta := ps \log(n\Delta/\epsilon)/\epsilon$ and $\tau := 2c \log n / \epsilon$, and note that the algorithm runs for exactly $\tau\eta$ iterations. We split the iterations of local search performed by our algorithm into $\tau$ \emph{phases} of length $\eta$. By \Cref{lem:rec}, for each $0 < i \leq \tau$ we have that 
$$ \mathbb E[C_{i \eta}^p] \leq \left(1 - \frac{\epsilon}{s} \right)^\eta \cdot \mathbb E[C_{(i-1)\eta}^p] + \left( \frac{\Psi}{1 - \epsilon} \right)^p.$$
Since the cost of any solution it at most $n \cdot d_{\max} w_{\max}$ and at least $d_{\min} w_{\min}$ (unless $\Opt_k = 0$), it follows that
$$ \mathbb E[C^p_{i \eta}] \leq \exp \left( - \frac{\epsilon}{s} \cdot \eta \right) \cdot ( n \cdot d_{\max} w_{\max})^p + \left( \frac{\Psi}{1 - \epsilon}\right)^p = \frac{\epsilon^p}{(n\Delta)^p} \cdot (n \cdot d_{\max} w_{\max})^p+ \left(\frac{\Psi}{1 - \epsilon}\right)^p$$
$$ \leq \epsilon \cdot \Psi^p + \left(\frac{\Psi}{1 - \epsilon}\right)^p \leq (1 + \epsilon) \cdot \left(\frac{\Psi}{1 - \epsilon}\right)^p.  $$
By Markov's inequality, we get that
$$ \Pr \left[ C_{i \eta}^p \geq (1 + \epsilon)^2 \cdot\left( \frac{\Psi}{1 - \epsilon} \right)^p \right] \leq \frac{1}{1 + \epsilon}. $$
Now, let $\mathcal B_i$ denote the event that $C_{i \eta}^p$ is at least this quantity. 
Since $C_{0}^p \geq C_1^p \geq \dots \geq C_{\tau \eta}^p$, the event $\mathcal B_i$ does not occur if the event $\mathcal B_{i-1}$ does not occur. Thus
$$ \Pr [ \mathcal B_i ] = \Pr [\mathcal B_i | \mathcal B_{i-1} ] \cdot \Pr[\mathcal B_{i-1}] + \Pr[\mathcal B_i | \overline{\mathcal B}_{i-1} ] \cdot \Pr[\overline{\mathcal B}_{i-1} ] \leq \frac{1}{1 + \epsilon} \cdot \Pr [ \mathcal B_{i-1} ].$$
It follows that 
$$\Pr [\mathcal B_{\tau}] \leq \left(\frac{1}{1 + \epsilon} \right)^\tau \leq \exp \left( -\frac{\epsilon}{1 + \epsilon} \cdot \tau \right) \leq \frac{1}{n^c}.$$
Hence, the algorithm outputs a solution with cost at most
$$(1 + \epsilon)^{2/p} \cdot \frac{\Psi}{1 - \epsilon} \leq (1 + 7\epsilon) \cdot (\cl(X) + 6p \cdot \Opt_k) $$
with probability at least $1 - 1/n^c$.

\subsection{Implementation of $\RandLoc$ (Proof of \Cref{lem: supp local 1})}\label{sec:supp local 1}

We begin by describing the auxiliary data structures that are necessary for the algorithm to admit an efficient implementation.

\medskip
\noindent\textbf{Auxiliary Data Structures.}
We assume that, as part of our input, we are given a collection of lists $\mathcal L = \{L_x\}_{x \in V}$ where the list $L_x$ contains all of the points in $X$ sorted in increasing distance from $x \in V$. We assume that these lists are implemented using balanced binary search trees, allowing for $O(\log n)$ time insertions, deletions, and membership queries.
The algorithm can modify the lists in $\mathcal L$ while performing the computation but must return them to their original state before the algorithm terminates.

\medskip
\noindent\textbf{The Algorithm.}
The algorithm begins by selecting an arbitrary subset $S \subseteq V$ of $k$ points and removing all of the points in $X \setminus S$ from each of the lists in $\mathcal L$. This can be done in $O(ns\log n)$ time. The algorithm then performs $T = 2c \cdot ps \log(n) \log(n\Delta/\epsilon)/\epsilon^2$ iterations of randomized local search. During each iteration, the algorithm samples some $x \sim X \setminus S$ independently and uniformly at random and makes a call to $\textsf{BestSwap}(S,x)$ to obtain the point $y \in S + x$ which minimizes $\cl(S +x-y)^p$. The algorithm then swaps $y$ and $x$ by setting $S \leftarrow S + x - y$, adds $x$ and removes $y$ from all of the lists in $\mathcal L$, and proceeds to the next iteration. Excluding the time taken to handle the call to the \textsf{BestSwap} oracle, all these operations can easily be implemented in $O(n \log n)$ time. Finally, after running for $T$ iterations, the algorithm adds all of the points in $S$ to each of the lists in $\mathcal L$, again taking $O(ns\log n)$ time.

\medskip
\noindent \textbf{Implementing The \textnormal{\textsf{BestSwap}} Oracle.}
The pseudocode in Algorithm~\ref{alg:local iteration} gives a formal description of the \textsf{BestSwap} oracle. We assume that the points in the list $L_x$ are always sorted in increasing distance from $x$.
We let $\Gamma(y)$ denote the set $\{x \in V \; | \; L_x(1) = y\}$ for the current state of the lists in $\mathcal L$.
Note that if each of the lists $L_x$ contains the same set of points $S' \subseteq V$, then $x \mapsto L_x(1)$ is the map that describes the optimal assignment of the points in $V$ to the solution $S'$. Furthermore, for any $y \in S'$, $\Gamma(y)$ denotes the set of points in $V$ that are assigned to $y$ by this optimal assignment.

\begin{algorithm}[H]\label{alg:local iteration}
    \SetAlgoLined
    \DontPrintSemicolon
    $L_{z} \leftarrow L_{z} + x$ for all $z \in V$\label{line add L}\;
        $\texttt{cost}^p[S + x] \leftarrow \sum_{z \in V} w(z) d(z, L_{z}(1))^p$\;
        \For{$y \in S + x$\label{for looop}}{
            $\texttt{cost}^p[S + x - y] \leftarrow \texttt{cost}^p[S + x] + \sum_{z \in \Gamma(y)} w(z) (d(z, L_{z}(2))^p - d(z, y)^p)$\;
        }
        $L_{z} \leftarrow L_{z} - x$ for all $z \in V$\;\label{line rem L}
    \Return{$\arg \min_{y \in S + x} \textup{\texttt{cost}}^p[S + x - y]$}
    \caption{\textsf{BestSwap}$(S, x)$}
\end{algorithm}

\begin{lemma}\label{lem:best swap}
    The \textnormal{\textsf{BestSwap}} oracle runs in $O(n \log n)$ time and correctly computes the point $y \in S + x$ which minimizes $\cl(S + x- y)^p$.
\end{lemma}

\begin{proof}
We first note that we can add or remove $x$ from all the lists in $\mathcal L$ in $O(n \log n)$ time.
Let $\pi : V \longrightarrow S+x$ denote the map that assigns each point $z \in V$ to the closest point in $S+x$.
Then we have that 
$$\cl(S+x)^p = \sum_{z \in V }w(z) d(z, \pi(z))^p = \sum_{z \in V }w(z) d(z, L_{z}(1))^p$$
since $L_{z}$ contains exactly the points in $S+x$.
Hence, we can compute $\texttt{cost}^p[S + j]$ by computing the sum on the right-hand side on $O(n)$ time.
In order to show that we can efficiently implement the floor loop in \Cref{for looop}, we note that the change $\cl(S + x - y)^p - \cl(S+x)^p$ caused by closing a center $y \in S + x$ is precisely
$$\sum_{z \in \Gamma(y)} w(z) (d(z, L_{z}(2))^p - d(z, y)^p). $$
This is because only the connection costs of the points whose closest center is $y$ will change, and in particular, decrease from $d(z, y)^p$ to $d(z, L_{z}(2))^p$.
Assuming that we have access to the set $\Gamma(y)$, we can compute this quantity in $O(|\Gamma(y)|)$ time.
Since each $z \in V$ appears in exactly one $\Gamma(y)$ across all $y \in S+x$, it takes $O(n)$ time to compute $\texttt{cost}^p[S + x - y]$ for all $y \in S + x$.

Finally, we now show how to construct the sets $\Gamma(y) = \{x \in V \; | \; L_x(1) = y\}$ in $O(n)$ time. When the algorithm enters the floor loop in \Cref{for looop},
we can construct an empty list $\Gamma(y)$ for each $y \in S + x$ in $O(n)$ time. We can then iterate over all $x \in V$ and place each $x$ into $\Gamma(L_x(1))$ in $O(n)$ total time. Hence, we can scan over all of the points in the set $\Gamma(y)$ in $O(|\Gamma(y)|)$ time.
\end{proof}

\section{Implementing the Consistency Hierarchy}\label{sec:imp consistency}

We now show how to efficiently implement the consistency hierarchy from \Cref{sec:consistency}. We do this by implementing a data structure that maintains the nested sets in the hierarchy and supports efficient reconstruction operations. The critical observation is that, by using the randomized local search algorithm from \Cref{sec:RLS}, we can perform the reconstructions fast enough to ensure that the amortized reconstruction time in each layer of the hierarchy is only $\tilde O(k^{1 + \epsilon}p /\epsilon)$.
In particular, we prove the following theorem.

\begin{theorem}
    There exists a
    fully dynamic algorithm that can maintain a $O(p/\epsilon)$-approximate solution to the $(k,p)$-clustering problem with $\tilde O(k^\epsilon/\epsilon)$ amortized recourse and $\tilde O(k^{1 + \epsilon}p/\epsilon^3)$ amortized update time, for any $\epsilon > 1/\log k$.
\end{theorem}

\subsection{Dynamic Sparsification}\label{sec:sparse para}

In order to ensure that the reconstruction time at the bottom level of the hierarchy is small,
we use the dynamic algorithm of \cite{ourneurips2023} to dynamically sparsify the underlying metric space before feeding it to our algorithm.
We do this by feeding the metric space $(V,w,d)$ of size $n$ to the sparsifier,
which then efficiently maintains a weighed metric subspace $(V', w', d)$ of size $\tilde O (k)$, which we in turn feed to our dynamic algorithm. It follows from the properties of the algorithm and the sparsified space that (i) any $\alpha$-approximate solution to the $(k,p)$-clustering problem in the space $(V', w', d)$ is also a $O(\alpha)$-approximate solution in the unsparsified space $(V,w,d)$, and (ii) the total number of updates in $V'$ is at most a $\tilde O(1)$ factor larger than the total number of updates in $V$. We describe the sparsification process in more detail in \Cref{app:sparsification}.

As a consequence of performing this sparsification, we can
assume we are running our dynamic algorithm on a weighted metric space $(V,w,d)$ such that the size of $V$ is always at most $\tilde O(k)$.
We can make this assumption while only incurring a $\tilde O(1)$ multiplicative loss in the amortized recourse and update times of our algorithm, as well as $O(1)$ multiplicative loss in the approximation ratio with probability at least $1 - \tilde O(1/n^c)$, for any constant $c \geq 1$.\footnote{We also incur an additive $\tilde O(k)$ factor in our amortized update time. However, this additive factor is dominated by the running time of our algorithm.}

\subsection{Data Structures}

In order to efficiently implement our dynamic algorithm, we design data structures that can explicitly maintain the sets $S_i$ at each layer of the hierarchy. 
Let $S_{-1}$ denote the set $V$.
At each point in time, our algorithm maintains the following data structures for each $-1 \leq i \leq \ell + 1$.
\begin{itemize}
    \item The set $S_i$.
    \item A list $L_{i,x}$ for each $x \in V$, such that $L_{i,x}$ contains all of the points in $S_{i}$ sorted in increasing order by their distance from $x$.
\end{itemize}
We assume that the set $S_i$ and lists $L_{i,x}$ are all stored using balanced binary trees, allowing for $O(\log n)$ time queries, insertions and deletions. We refer to the lists $L_{i,x}$ as the auxiliary data structures.


\medskip
\noindent\textbf{Maintaining the Data Structures.}
We now show how to efficiently maintain these data structures and use them to efficiently implement reconstructions using the randomized local search algorithm from \Cref{sec:RLS}.


\medskip
\noindent\textbf{Maintaining the Sets $S_i$.}
Let $0 \leq i \leq \ell + 1$. We later show how to reconstruct the set $S_i$ using the randomized local search algorithm. Maintaining the set between reconstructions is straightforward since we can lazily update the set $S_i$ (i.e.~add any point inserted into $V$ to $S_i$) in $\tilde O(1)$ worst-case time.

\medskip
\noindent\textbf{Maintaining the Lists $L_{i,x}$.}
Let $0 \leq i \leq \ell + 1$.\footnote{The case where $i = -1$ follows by a similar argument.}
To maintain the lists $\{L_{i,x}\}_{x \in V}$, we update this collection of lists every time a point is either inserted or deleted from $S_i$ or $V$. Whenever a point $x$ is inserted to $V$, we must create a new list $L_{i,x}$ and add all of the points in $S_i$, which can be done in $\tilde O(k)$ time. Whenever a point $x$ is deleted from $V$, we delete the list $L_{i,x}$, which can also be done in $\tilde O(k)$ time. Whenever a point $y$ is inserted into (resp.~deleted from) $S_i$, we must add (resp.~remove) $y$ to each list in $\{L_{i,x}\}_{x \in V}$, which can be done in $\tilde O(k)$ time. Thus, we can maintain the lists $\{L_{i,x}\}_{x \in V}$ with $\tilde O(k)$ overhead as points are inserted and deleted from $S_i$ and $V$. Since $V$ has an amortized recourse of $\tilde O(1)$ (by the sparsification described in \Cref{sec:sparse para}) and $S_i$ has an amortized recourse of $\tilde O(k^{\epsilon}/\epsilon)$ (see \Cref{sec:recourse}), it follows that the amortized time taken to maintain these lists is $\tilde O(k^{1 + \epsilon}/\epsilon)$. Hence, we get the following lemma.

\begin{lemma}\label{lem:total aux time}
    Given some sequence of $T$ point insertions and deletions, the total time our algorithm spends maintaining the data structures is $T \cdot \tilde O(k^{1+\epsilon}/\epsilon^2)$.
\end{lemma}

\medskip
\noindent\textbf{Handling Reconstructions.}
Suppose that, after some update, we need to reconstruct layer $i$ of our hierarchy.
We begin by
making a call to $\RandLoc$ to obtain a $(1 + O(\epsilon), O(p), S_{i-1})$-restricted solution $S_{i}$ to the $(k + s_i, p)$-clustering problem in the metric space $(V, w, d)$ with probability at least $1 - 1/n^c$.\footnote{Recall that $S_{-1}$ denotes the set $V$.} Since we maintain the set $S_{i-1}$ and the lists $\{L_{i-1, x}\}_{x \in V}$, it follows from Lemmas~\ref{lem: supp local 2} and \Cref{lem: supp local 1} that this can be done in $\tilde O(k^{1 + (\ell - i + 1)\epsilon}p/\epsilon^2)$ time. We then proceed to reconstruct all layers $j > i$ of the hierarchy in increasing order. Since $\epsilon \geq 1/ \log k$, it's easy to verify that
$$\tilde O(k^{1 + (\ell - i + 1)\epsilon}p/\epsilon^2) \leq \sum_{j = i}^{\ell + 1} \tilde O(k^{1 + (\ell - j + 1)\epsilon}p/\epsilon^2).$$
In other words, the time taken to reconstruct layer $i$ dominates the time taken to reconstruct layer $i$ and all the subsequent layers.
By a similar argument to \Cref{sec:recourse}, we get the following lemma.

\begin{lemma}\label{lem:total recon time}
    Given some sequence of $T$ point insertions and deletions, the total time our algorithm spends reconstructing layers is $T \cdot \tilde O(k^{1+\epsilon}p/\epsilon^3)$.
\end{lemma}

\begin{proof}
    Let $r^t \in \{0,\dots, \ell + 1\}$ be the layer that we reconstruct from while handling the $t^{th}$ update. Then, by the arguments in the preceding paragraph, it follows that the total time our algorithm spends reconstructing layers is at most
    $$ \sum_{t = 1}^T \tilde O \! \left( k^{1 + (\ell - r^{t} + 1)\epsilon} p /\epsilon^2 \right) \leq \tilde O(kp/\epsilon^2) \cdot \sum_{t = 1}^T  k^{(\ell - r^{t} + 1)\epsilon}. $$
    By the arguments in \Cref{sec:recourse}, we have that $\sum_{t = 1}^T  k^{(\ell - r^{t} + 1)\epsilon} \leq T \cdot O(k^\epsilon / \epsilon)$. Hence, our algorithm spends at most $T \cdot \tilde O(k^{1 +\epsilon} p/\epsilon^3)$ time reconstructing layers.
\end{proof}

\medskip
\noindent\textbf{Amortized Update Time.}
It is easy to maintain the counters $\tau_i$ in Algorithm~\ref{alg:update} to keep track of when we need to reconstruct each of the layers in $\tilde O(1/\epsilon)$ worst-case time per update. Hence,
it follows from Lemmas~\ref{lem:total recon time} and \ref{lem:total aux time} that the total time that our algorithm spends handling a sequence of $T$ point insertions and deletions is $T \cdot \tilde O(k^{1 + \epsilon}p/\epsilon^3)$. Thus, the amortized update time of our algorithm is $\tilde O(k^{1 + \epsilon}p/\epsilon^3)$.

\subsection{Approximation Ratio}
We now bound the approximation ratio of the solution maintained by our dynamic algorithm. Consider the state of our algorithm after handling some arbitrary sequence of updates.
We first note that the solution $S_i$ was a $(1 + O(\epsilon), O(p), S_{i-1})$-restricted solution when it was last reconstructed with probability at least $1 - O(1/n^c)$. Thus, by taking a union bound, this is true for all layers $i \in \{0,\dots, \ell + 1\}$ with probability at least $1 - O(1/(\epsilon n^c))$. Conditioned on this high probability event, we can apply \Cref{lem:approx} to get that the approximation ratio of the solution maintained by the algorithm \emph{with respect to the sparsified metric space the algorithm is being run on} is at most
$$ O(p) \cdot \sum_{i=0}^{\ell + 1} (1 + O(\epsilon))^i \leq O \! \left( \frac{p}{\epsilon} \right). $$
By the properties of the sparsifier we are using (see \Cref{sec:sparse para}) it follows that this solution is a $O(p/\epsilon)$-approximation to the $(k,p)$-clustering problem on the unsparsified metric space with high probability.

\section{Maintaining a Proper Solution}\label{sec:proper}

Recall that the output of our dynamic algorithm, $S_{\ell + 1}$, is an improper solution, i.e.~$S_{\ell + 1}$ is not necessarily a subset of $V$ (see \Cref{sec: proper improper def}).
In this section, we show how to dynamically maintain a projection of $S_{\ell + 1}$ onto the metric space $V$, $\Proj(S_{\ell + 1}, V)$, in order to obtain a proper solution, while incurring only small overhead in recourse, update time, and approximation. In particular, we prove the following lemma.

\begin{lemma}
    Given any fully dynamic $(k,p)$-clustering algorithm that maintains an $\alpha$-approximate improper solution, we can maintain a $2\alpha$-approximate proper solution while incurring a $O(1)$ factor in the recourse and an additive $\tilde O(|V|)$ factor in the update time.
\end{lemma}

\subsection{Projecting $A$ Onto $B$}
Let $A$ and $B$ by subsets of the underlying metric space $(\mathcal V, d)$. We say that $\Lambda \subseteq B$ is a \emph{projection of $A$ onto $B$} if, for each $y \in A$, there exists some $x \in \Lambda$ such that $d(y,x) = d(y,B)$. In other words, if $\Lambda$ contains one of the points in $B$ that is closest to $y$. 
Recall that, by \Cref{lem:proj cor}, we have that $\cl_p(\Lambda, B) \leq 2 \cl_p(A, B)$ for any such projection $\Lambda$.




\subsection{Maintaining a Projection of $A$ Onto $B$}\label{subsec:dynamic proj}

Now suppose that the sets $A$ and $B$ are both evolving via a sequence of point insertions and deletions.
We now describe a dynamic algorithm that maintains a projection $\Lambda$ of $A$ onto $B$. We then show that $O(1)$ many points are added or removed from $\Lambda$ after each update to the sets, before describing the data structures that can be used to implement this algorithm efficiently.

\medskip
\noindent\textbf{The Dynamic Projection Algorithm.}
As points are inserted and deleted from $A$, we maintain an ordering $y_1,\dots , y_{m}$ of the points in $A$. When a point is inserted, we append it to the end of the ordering, and when a point is deleted, we remove it from the ordering and do not change the relative positions of the other points in $A$. The projection $\Lambda$ that we maintain is defined by the output of Algorithm~\ref{alg:project}.

\begin{algorithm}[H]
    \SetAlgoLined
    \DontPrintSemicolon
    $\Lambda_0 \leftarrow \varnothing$\;
    \For{$i = 1 \dots m$}{
        $\Lambda_i \leftarrow \Lambda_{i-1} + \{\pi_{B \setminus \Lambda_{i-1}}(y_i)\}$\;
    }
    \Return{$\Lambda_m$}
    \caption{\textsc{Project}$(A,B)$}
    \label{alg:project}
\end{algorithm}

\noindent
In the above algorithm, $\pi_{B \setminus \Lambda_{i-1}}(y_i)$ denotes a point in $B \setminus \Lambda_{i-1}$ that is closest to $y_i$.
By breaking ties consistently (and arbitrarity) when computing $\pi_{B \setminus \Lambda_{i-1}}(y_i)$, each set $\Lambda_i$ is well defined. It's easy to see that, for each $i \in [m]$, $\Lambda_i$ is a projection of $\{y_1,\dots, y_i\}$ onto $B$, and in particular, that $\Lambda_m$ is a projection from $A$ onto $B$. 

\medskip
\noindent\textbf{Recourse Analysis.}
We now analyse the recourse of the projection $\Lambda_m$ as the sets $A$ and $B$ are updated. Suppose that we have sets $A = \{y_1, \dots, y_m\}$ and $B$, and let $\Lambda$ be the projection returned by our dynamic algorithm on this input. Suppose we update the sets $A$ and $B$ by either inserting or deleting a point from one of these sets, and let $\Lambda'$ denote the updated projection. 

\begin{lemma}\label{lem: proj rec}
    We have that $|\Lambda \oplus \Lambda'| \leq 2$.
\end{lemma}

\begin{proof}
    We assume that $|A| \leq |B|$ before and after the update. Otherwise, the lemma is trivial.
    
    In the case that a point $y$ is inserted to $A$, the new ordering is $y_1, \dots , y_m, y$. Thus, since $B$ does not change, we have that $\Lambda' \supseteq \Lambda$ and $|\Lambda' | = |\Lambda| + 1$. Hence, $|\Lambda \oplus \Lambda'| = 1$. 
    
    In the case that a point $y_j$ is deleted from $A$, the new ordering is $y_1, \dots ,y_{j-1}, y_{j+1}, \dots, y_m$. Let $\Lambda_i$ denote the projection of $\{y_1,\dots ,y_i\}$ onto $B$ maintained by the algorithm before the update, and $\Lambda_i'$ the projection of $\{y_1,\dots ,y_i\} \setminus \{y_j\}$ maintained by the algorithm after the update. Then we have that $\Lambda_i = \Lambda'_i$ for all $i < j$. Let $x_i$ denote the element in $\Lambda_i \setminus \Lambda_{i-1}$, and let $x'_i$ denote the element in $\Lambda'_i \setminus \Lambda'_{i-1}$. Note that $\Lambda'_j \setminus \Lambda'_{j-1}$ is empty, so $x'_j$ is undefined. We can now see that $\Lambda_{i}' \subseteq \Lambda_i$ for all $i$. We can show this by induction. Clearly, this is true for all $i \leq j$. For $i > j$, if $\Lambda_{i-1}' \subseteq \Lambda_{i-1}$ and $x'_i \neq x_i$, then it must be the case that $x'_i \in \Lambda_i$. Thus, $\Lambda_{i}' \subseteq \Lambda_{i}$. It follows that $\Lambda' \subseteq \Lambda$. Since $|\Lambda'| = |\Lambda| - 1$, we get that $|\Lambda \oplus \Lambda'| = 1$.

    Finally, we consider the case that an element $x$ is inserted into $B$. By a symmetric argument, the case of a deletion from $B$ also follows. 
    Let $\Lambda_i$ and $\Lambda'_i$ denote the projections of $A$ onto $B$ maintained by the algorithm before and after the update respectively.
    Let $x_i$ and $x'_i$ denote the elements in $\Lambda_i \setminus \Lambda_{i-1}$ and $\Lambda'_i \setminus \Lambda'_{i-1}$ respectively. We now show that $\Lambda_i' \subseteq \Lambda_i \cup \{x\}$ for all $i$ by induction. First note that $x_1' \in \{x_1, x\}$, thus $\Lambda_1' \subseteq \Lambda_1 \cup \{x\}$. For $i > 0$, if $\Lambda_{i-1}' \subseteq \Lambda_{i-1} \cup \{x\}$ and $x'_i \neq x_i$, then it must be the case that $x_i \in \Lambda_{i-1} \cup \{x\}$. Thus, $\Lambda_i' \subseteq \Lambda_i \cup \{x\}$. If follows that $\Lambda' \subseteq \Lambda \cup \{x\}$. Since $|\Lambda'| = |\Lambda|$, we get that $|\Lambda \oplus \Lambda'| \leq 2$.
\end{proof}

\medskip
\noindent\textbf{Data Structures.}
We now give data structures that can be used to efficiently implement this dynamic algorithm.
For each $y_i \in A$, we maintain a list $L_{y_i}$ of all the points in $B \setminus \Lambda_{i-1}$, sorted in increasing distance from $y_i$. We can implement these lists using balanced binary trees, allowing for $O(\log n)$ time insertions and deletions, where $n$ denotes the size of $B$. Note that we have that $\Lambda_i \setminus \Lambda_{i-1} = \{L_{y_i}(1)\}$. 
When a point $y$ is inserted into (resp.~deleted from) $A$, we create (resp.~delete) the list $L_{y}$, which can be done in $\tilde O(n)$ time. Furthermore, after one of the sets $A$ or $B$ is updated, we need to appropriately update the lists. By \Cref{lem: proj rec}, we know that we only need to make $O(1)$ many changes to each of the lists.\footnote{Note that $|(B \setminus \Lambda_{i-1}) \oplus (B \setminus \Lambda'_{i-1})| = |\Lambda_{i-1} \oplus \Lambda'_{i-1}| \leq 2$ points need to be added and removed from $L_{y_i}$.} We can update the lists by scanning through the $y_i$ in order while maintaining a set of the $O(1)$ many points that might be added or removed from $L_{y_i}$ as we update the lists. Thus, we can update each of the lists in $\tilde O(1)$ time, taking $\tilde O(n)$ time in total.

\subsection{Converting from Improper to Proper Solutions}

Suppose that we have a dynamic algorithm which maintains an $\alpha$-approximate improper solution $S$ to the $(k,p)$-clustering problem in a dynamic metric space $(V,w,d)$. By using the dynamic algorithm from \Cref{subsec:dynamic proj}, we can maintain a projection $S'$ of $S$ onto $V$ with additive $\tilde O(|V|)$ overhead in the update time. By \Cref{lem:proj cor}, we have that $S'$ is a $2\alpha$-approximate proper solution to $(k,p)$-clustering problem in $(V,w,d)$. Furthermore, by \Cref{lem: proj rec}, we have that the recourse of the set $S'$ is at most a constant factor larger than the recourse of the space $V$ and the solution $S$ combined.


\section{Dynamic Sparsification}\label{app:sparsification}

In this section, we show how to dynamically sparsify the input metric space using the dynamic algorithm of \cite{ourneurips2023}. More precisely, given some dynamic metric space $(V,w,d)$ and parameters $k,p \in \mathbb N$, we show how to dynamically maintain a dynamic metric subspace $(V',w',d)$ in $\tilde O(k)$ amortized update time such that the following hold.
\begin{itemize}
    \item $V' \subseteq V$ and $|V'| = \tilde O(k)$ at all times.
    \item A sequence of $T$ updates in $(V,w,d)$ leads to at most $\tilde O(T)$ updates in $(V',w',d)$.
    \item Any $\alpha$-approximate solution to the $(k,p)$-clustering problem in the metric space $(V',w',d)$ is also an $O(\alpha)$-approximate solution to the $(k,p)$-clustering problem in the metric space $(V,w,d)$ with probability at least $1 - \tilde O(1/n^c)$.
\end{itemize}
Now, suppose that we are given a sequence of updates $\sigma_1,\dots,\sigma_T$ in a dynamic metric space $(V,w,d)$. Instead of feeding the metric space $(V,w,d)$ directly to a dynamic clustering algorithm, we can perform this dynamic sparsification to obtain a sequence of updates $\sigma'_1,\dots,\sigma'_{T'}$ for a metric space $(V',w',d)$, where $T' = \tilde O(T)$, and feed the metric space $(V',w',d)$ to our dynamic algorithm instead. If the dynamic algorithm maintains an $\alpha$-approximate solution $S$ to the $(k,p)$-clustering problem in the $(V',w',d)$, then $S$ is also an $O(\alpha)$-approximate solution to the $(k,p)$-clustering problem in $(V,w,d)$ with high probability.

It follows that we can assume that the size of the underlying metric space is always at most $\tilde O(k)$
while only incurring a $\tilde O(1)$ multiplicative loss in the amortized recourse and update times of our algorithms, as well as $O(1)$ multiplicative loss in the approximation ratio with high probability.\footnote{We also incur an additive $\tilde O(k)$ factor in our amortized update time, since we must run the algorithm $\Sparsifier$. However, this additive factor is dominated by the running time of our algorithm.} We summarize the sparsification process with the following theorem.
\begin{theorem}[\texttt{Informal}]
\label{thm:black:box:sparsifier}
    Given an $\alpha$-approximate dynamic algorithm for the $(k,p)$-clustering problem with $T(n,k)$ update time and $R(n,k)$ recourse, we can use sparsification to obtain a $O(\alpha)$-approximate algorithm with $T(\tilde O(k),k)$ update time and $R(\tilde O(k),k)$ recourse.
\end{theorem}

\subsection{The Dynamic Algorithm of \cite{ourneurips2023}}

Given some unweighted dynamic metric space $(V,d)$, the fully dynamic algorithm of \cite{ourneurips2023}, which we refer to as $\Sparsifier$, maintains a weighted metric subspace $(V', w', d)$ of $(V,d)$ of size $\tilde O(k)$. More specifically, $\Sparsifier$ maintains the following.



\medskip
\noindent\textbf{The Dynamic Algorithm $\Sparsifier$.}
Given a dynamic metric space $(V,d)$ which evolves via a sequence of point insertions and deletions and a parameter $k \in \mathbb N$, the algorithm $\Sparsifier$ explicitly maintains the following objects.

\begin{itemize}
    \item A set $V' \subseteq V$ of at most $\tilde O(k)$ points.
    \item An assignment $\sigma : V \longrightarrow V'$ such that, for all $p \in \mathbb N$, we have $\cl_p(\sigma, V) \leq O(1) \cdot\Opt_k^p(V)$ with probability at least $1 - O(1/n^c)$.\footnote{$c \geq 1$ can be any constant.}
\end{itemize}
The following lemma summarises the behaviour of this algorithm. 

\begin{lemma}\label{lem:spasifier prop}
    The algorithm $\Sparsifier$ has an amortized update time of $\tilde O(k)$ and the amortized recourse of the set $V'$ is $\tilde O(1)$.
\end{lemma}

\begin{proof}
    The amortized update time guarantees of the algorithm follow directly from Lemma 3.6 in \cite{ourneurips2023}. The amortized recourse of the set $V'$ is the total number of insertions and deletions in the set $V'$ across the sequence of updates divided by the length of the update sequence. Using the potential function from Section 3.3 of \cite{ourneurips2023}, we can see that this is bounded by $\tilde O(1)$.
\end{proof}

\noindent
The following structural lemma, which is a generalization of Theorem C.1 in \cite{ourneurips2023}, is the key observation behind why the output of this dynamic algorithm can be used as a sparsifier.

\begin{lemma}\label{lem:bicri=sparsifier2}
    Let $(V,w,d)$ be a metric space, $V' \subseteq V$, and $\sigma : V \longrightarrow V'$ be a map such that $\cl_p(\sigma, V, w) \leq \beta \cdot \Opt^p_k(V, w)$. Furthermore, let $w'(y) := \sum_{x \in \sigma^{-1}(y)} w(x)$ for all $y \in V'$ and $S \subseteq V'$ be a set such that $\cl_p(S, V', w') \leq \alpha \cdot \Opt^p_k(V', w')$. Then we have that
    $$ \cl_p(S, V, w) \leq (\beta + 2\alpha(\beta + 1)) \cdot \Opt^p_k(V, w). $$
\end{lemma}

\begin{proof}
    By applying Minkowski's inequality as in \Cref{lem:minkow}, we have that
    \begin{align*}
        \cl_p(S, V, w) &\leq \left( \sum_{x \in V} w(x) (d(x, \sigma(x)) + d(\sigma(x), S))^p \right)^{1/p}\\
        &\leq \left( \sum_{x \in V} w(x) d(x, \sigma(x))^p \right)^{1/p} + \left( \sum_{x \in V} w(x) d(\sigma(x), S)^p \right)^{1/p}\\
        &= \left( \sum_{x \in V} w(x) d(x, \sigma(x))^p \right)^{1/p} + \left( \sum_{y \in V'}w'(y)d(y, S)^p \right)^{1/p}\\
        &= \cl_p(\sigma, V, w) + \cl_p(S, V', w')\\
        &\leq \cl_p(S, V', w') + \beta \cdot \Opt_k^p(V,w).
    \end{align*}    
    Now, let $S^\star$ be an optimal solution to the $(k,p)$-clustering problem in the space $(V,w,d)$. Then, by again applying Minkowski's inequality, we have that
    $$ \cl_p(S, V', w') \leq 2\alpha \cdot \cl_p(S^\star, V', w') = 2\alpha \cdot \left( \sum_{y \in V'} w'(y) d(y, S^\star)^p \right)^{1/p} = 2\alpha \cdot  \left( \sum_{x \in V} w(x) d(\sigma(x), S^\star)^p \right)^{1/p} $$
    $$ \leq 2\alpha \cdot \left( \sum_{x \in V} w(x) (d(x, \sigma(x)) + d(x, S^\star) )^p \right)^{1/p} \leq 2 \alpha \cdot \cl_p(\sigma, V, w) + 2\alpha \cdot \cl_p(S^\star, V, w) $$
    $$ \leq 2 \alpha \beta \cdot \Opt_k^p(V, w)+ 2\alpha \cdot \Opt_k^p(V, w) = 2 \alpha(\beta + 1) \cdot \Opt_k^p(V, w). $$
    Combining these inequalities, we get that
    $$ \cl_p(S, V, w) \leq \cl_p(\sigma, V, w) + \cl_p(S, V', w') \leq (\beta + 2 \alpha(\beta + 1) ) \cdot \Opt_k^p(V, w). $$
\end{proof}

\noindent
It immediately follows from \Cref{lem:bicri=sparsifier2} that any $\alpha$-approximate solution to the $(k,p)$-clustering problem in the weighted metric subspace $(V', w', d)$ defined with respect to the map $\sigma$ maintained by $\Sparsifier$ (i.e.~$w'(y) = |\sigma^{-1}(y)|$ for all $y \in V'$) is also a $O(\alpha)$-approximate solution to the $(k,p)$-clustering problem on the metric space $(V,d)$ with probability at least $1 - O(1/n^c)$. Hence, this data structure can be used to sparsify unweighted dynamic metric spaces. We now show how to extend this to weighted metric spaces.

\subsection{Extension to Weighted Metric Spaces}\label{sec:weighted sparsifier}

Let $(V,w,d)$ be a dynamic metric space evolving via a sequence of point insertions and deletions. For each integer $\log(w_{\min}) \leq i \leq \log(w_{\max})$, let $V_i:= \{x \in V \; | \; 2^i \leq w(x) < 2^{i+1}\}$. Suppose that we create $O(\log (\Delta_w))$ instances of $\Sparsifier$, one for each such $i$, and feed them the (unweighted) metric spaces $(V_i, d)$. In return, these data structures maintain a collection of maps $\sigma_i : V_i \longrightarrow V_i$ such that $\cl_p(\sigma_i, V_i) \leq O(1) \cdot \Opt_k^p(V_i)$ for all $i$ with probability at least $1 - \tilde O(1/n^c)$.\footnote{This probability guarantee follows directly from a union bound.} Let $\sigma : V \longrightarrow V$ be the map defined as $\sigma(x) := \sigma_i(x)$ for all $x \in V_i$. Since the $V_i$ partition $V$, $\sigma$ is well defined. Assume that $\cl_p(\sigma_i, V_i) \leq O(1) \cdot \Opt_k^p(V_i)$ for all $i$. We now prove the following lemma.

\begin{lemma}
    $\cl_p(\sigma, V, w) \leq O(1) \cdot \Opt^p_k(V, w)$.
\end{lemma}

\begin{proof}
    We begin by showing that $2^{O(p)} \cdot \Opt^p_k(V,w)^p \geq \sum_i 2^i \cdot \Opt^p_k(V_i)^p$. Let $S^\star$ be an optimal solution to the $(k,p)$-clustering problem in $(V,w,d)$ and let $S^\star_i$ be an optimal solution to the $(k,p)$-clustering problem in $(V_i,d)$. Then we have that
    $$ \cl_p(S^\star, V, w)^p = \sum_i \cl_p(S^\star, V_i, w)^p \geq \sum_i 2^i \cdot \cl_p(S^\star, V_i)^p \geq \frac{1}{2^{O(p)}} \cdot \sum_i 2^i \cdot \cl_p(S^\star_i, V_i)^p $$
    and hence we have $2^{O(p)} \cdot \Opt_k^p(V,w)^p \geq \sum_i 2^i \cdot \Opt_k(V_i)^p$.
    It follows that
    $$ \cl_p(\sigma, V, w)^p = \sum_i \cl_p(\sigma, V_i, w)^p \leq 2 \cdot \sum_i 2^{i} \cdot\cl_p(\sigma_i, V_i)^p $$
    $$\leq 2^{O(p)} \cdot \sum_i 2^i \cdot \Opt_k^p(V_i)^p \leq 2^{O(p)} \cdot \Opt^p_k(V,w)^p. $$
\end{proof}

\noindent
By \Cref{lem:bicri=sparsifier2}, it again immediately follows that any $\alpha$-approximate solution to the $(k,p)$-clustering problem in the weighted metric subspace $(V', w', d)$ defined with respect to the map $\sigma$ (i.e.~$w'(y) = \sum_{x \in \sigma^{-1}(y)} w(x)$ for all $y \in V'$) is also a $O(\alpha)$-approximate solution to the $k$-median problem on the metric space $(V,w,d)$ with probability at least $1 - \tilde O(1/n^c)$. 

It follows from the properties of $\Sparsifier$ that the set $V' = \sigma (V)$ has an amortized recourse of $\tilde O(1)$, and can be maintained in $\tilde O(k)$ amortized update time. Furthermore, the weights of the points in $V'$ can also be maintained explicitly using the data structures in \cite{ourneurips2023}.

It follows that we can sparsify dynamic weighted metric spaces by maintaining $\tilde O(1)$ many copies of $\Sparsifier$, running them on the appropriate unweighted subspaces, and combing the outputs with the appropriate weights.

\newpage

\part{Our Algorithm for Dynamic $k$-Median Value }\label{part:II}

\section{Preliminaries}\label{sec:prelim 2}

In this section, we formally describe the notation that we use throughout \Cref{part:II}. We note that the notation used in \Cref{part:II} is different from the notation used in \Cref{part:I}. See \Cref{sec:P3 org} below for a summary of how \Cref{part:II} is organized and the main result of this part of the paper.

\subsection{Metric Spaces}

In this paper, we work with clustering problems defined on \emph{weighted metric spaces}. A weighted metric space is a triple $(V,w,d)$, where $V$ is a set of points, $d : V \times V \longrightarrow \mathbb{R}_{\geq 0}$ is a distance function satisfying the triangle inequality, and $w : V \longrightarrow \mathbb R_{> 0}$ is a function assigning a weight $w(i)$ to each point $i \in V$. 
We let $(V,d)$ denote the unweighted metric space where all the weights are $1$.
Given a point $i \in V$ and a subset $X \subseteq V$, we denote the distance from $p$ to the closest point in $X$ by $d(i, X) := \min_{j \in X} d(i,j)$.
We let $d_{\text{min}}$ and $d_{\text{max}}$ respectively denote the minimum (non-zero) and maximum distance between any two points in the metric space. We assume that, for any two distinct points $i$ and $j$ in the metric space, we have $d(i,j) \neq 0$. Similarly, we let $w_{\text{min}}$ and $w_{\text{max}}$ respectively denote the minimum and maximum weights assigned to points in the metric space. We define $\Delta_d := d_{\text{max}} / d_{\text{min}}$ and $\Delta_w := w_{\text{max}} / w_{\text{min}}$. The \emph{aspect ratio} of the space is $\Delta := \Delta_d \Delta_w$.

\medskip
\noindent\textbf{Distance Oracle.}
Whenever we are working with a metric space $(V,d, w)$, we assume that we have constant time query access to the distances between points. In other words, we assume access to an oracle that given points $i,j \in V$, returns the value of $d(i,j)$ in $O(1)$ time.

\subsection{Uniform Facility Location and $k$-Median}\label{sec:problem defs}

As input to the \emph{uniform facility location} (UFL) problem, we are given a metric space $(V,w,d)$ and a parameter $\lambda > 0$ which we call the \emph{facility opening cost}.
Our goal is to find a set $S \subseteq V$ of points (sometimes referred to as \emph{facilites})
that minimizes the following objective function, where each point $i \in V$ (sometimes referred to as a \emph{client}) is \emph{assigned} to the closest point in $S$,
$$\fl_\lambda(S, V, w) := \lambda  |S| + \sum_{i \in V} w(i) d(i, S).$$
We denote the objective value of the optimum solution by $\optf_\lambda(V,w) := \min_{S \subseteq V} \fl_\lambda(S,V,w)$.
As input to the \emph{$k$-median} problem, we are again given a metric space $(V,w,d)$ and an integer $k \geq 1$. The setting is similar to UFL, except that there is no facility opening cost and instead a constraint that $|S| \leq k$. Our goal is to find a set $S \subseteq V$ of points (sometimes referred to as \emph{centers}) of size at most $k$, that minimizes the following objective function, where each $i \in V$ is assigned to the closest point in $S$,
$$ \cl(S,V,w) := \sum_{i \in V} w(i) d(i, S). $$
We denote the objective value of the optimum solution by $\optc_k(V,w) := \min_{S \subseteq V, |S| \leq k} \cl(S, V, w)$. When the underlying metric space is obvious from the context, we omit the arguments $V$ and $w$. For example, we write $\optc_k$ instead of $\optc_k(V,w)$, $\cl(S)$ instead of $\cl(S,V,w)$, and so on.

\medskip
\noindent\textbf{Remark on the Definition of UFL.}
The facility location problem is often stated so that the assignment from the clients to facilities is required as part of the solution. Since our objective is ultimately to obtain solutions to $k$-median, where points are always implicitly assigned to the closest center, we make the same assumption in our definition of facility location.

\subsection{Fractional UFL and $k$-Median}\label{sec:LPs}
Throughout \Cref{part:II}, we work with \emph{fractional} solutions to the \emph{integral} problems described in Section~\ref{sec:problem defs}. In order to define the fractional variants of these problems, we give LP-relaxations.
We now define the standard LP-relaxation for UFL, as well as its dual LP, which plays a crucial role in the analysis of our algorithms. We then give the standard LP-relaxation for $k$-median.

\medskip
\noindent \textbf{Primal LP for UFL:}
\begin{eqnarray}
\label{eq:primal:lp:r}
\text{Minimize }   \sum_{i \in V} \lambda \cdot y_i + \sum_{i, j \in V} w(j)d(i, j) \cdot x_{j \to i} \\
\text{s.t.~} \sum_{i \in V} x_{j \to i}  \geq  1 & & \text{ for all } j \in V \label{eq:primal:lp:1:r} \\
x_{j \to i} \leq y_i & & \text{ for all } i, j \in V \label{eq:primal:lp:2:r} \\
y_i, x_{j \to i} \geq 0 & & \text{ for all  } i, j \in V \label{eq:primal:lp:3:r}
\end{eqnarray}
The variable $y_i$ indicates to what extent the point $i \in V$ is included in the solution (i.e. to what extent it is opened as a facility), and $x_{j \to i}$ indicates to what extent the point $j \in V$ has been assigned to $i \in V$. We refer to a collection of variables $\{y_i, x_{j \to i}\}$ satisfying the constraints of this LP as a \emph{fractional solution to UFL}. We say that such a solution opens $\sum_{i \in V} y_i$ facilities. We denote the objective value of the optimal fractional solution to UFL by $\foptf_\lambda$.


\medskip
\noindent  \textbf{Dual LP for UFL:}
\begin{eqnarray}
\label{eq:dual:lp:r}
\text{Maximize }   \sum_{i \in V} v_i \\
\text{s.t.~} \sum_{j \in V} w_{j \to i}  \leq  \lambda & & \text{ for all } i \in V \label{eq:dual:lp:1:r} \\
v_j \leq w(j)d(i, j) + w_{j \to i} & & \text{ for all } i, j \in V \label{eq:dual:lp:2:r} \\
v_i, w_{j \to i} \geq 0 & & \text{ for all  } i, j \in V \label{eq:dual:lp:3:r}
\end{eqnarray}
Note that, given any dual solution $\{v_i, w_{j \to i}\}$, we have that $\sum_{i \in V} v_i \geq \foptf_\lambda$.

\medskip
\noindent \textbf{LP-Relaxation for $k$-Median:}
\begin{eqnarray}
\label{eq:kmed:lp:r}
\text{Minimize } \sum_{i, j \in V} w(j)d(i, j) \cdot x_{j \to i} \\
\text{s.t.~} \sum_{i \in V} y_{i}  \leq k  & & \label{eq:kmed:lp:1:r} \\
\sum_{i \in V} x_{j \to i}  \geq  1 & & \text{ for all } j \in V \label{eq:kmed:lp:2:r} \\
x_{j \to i} \leq y_i & & \text{ for all  } i, j \in V\label{eq:kmed:lp:3:r} \\
y_i, x_{j \to i} \geq 0 & & \text{ for all  } i, j \in V \label{eq:kmed:lp:4:r}
\end{eqnarray}
The variable $y_i$ indicates to what extent a point $i \in V$ is included in the solution, and  $x_{j \to i}$ indicates to what extent $j \in V$ has been assigned to the point $i \in V$. We refer to a collection of variables $\{y_i, x_{j \to i}\}$ satisfying the constraints of this LP as a \emph{fractional solution to $k$-median}.
We denote the objective value of the optimal fractional $k$-median solution by $\foptc_k$.

\medskip
\noindent\textbf{Remark on Fractional Solutions.}
In the same way that integral $k$-median (resp.~UFL) solutions are completely defined by the points taken as centres (resp.~opened as facilities), the variables $\{y_i\}$ define a fractional solution to $k$-median (resp.~UFL), where each point is implicitly (fractionally) assigned to the closest centers (resp.~open facilities).

\medskip
\noindent\textbf{Integrality Gap of $k$-Median and UFL.}
Clearly, the cost of the optimal fractional solutions for $k$-median and UFL are at most the cost of the optimal integral solutions for these problems.
It is also known that the LP-relaxations for the $k$-median and UFL problems have small constant integrality gaps, i.e. the costs of the optimal solutions to fractional $k$-median and UFL are at most a small constant multiplicative factor smaller than the costs of the optimal integral solutions. In particular, we use the fact that, for any $k \in \mathbb N$, $\optc_k \leq 3 \cdot \foptc_k$ \cite{esa/ArcherRS03}.

\subsection{LMP Approximations for UFL}



We now introduce the notion of \emph{Lagrange multiplier preserving} (LMP) approximations for UFL \cite{JainV01}. Fix some instance of UFL, i.e. a metric space $(V,w,d)$ and facility opening cost $\lambda > 0$. We now define the notation of an LMP $\alpha$-approximation for UFL, for $\alpha \geq 1$.

\medskip
\noindent\textbf{LMP Approximations.}
Given a solution for this instance of UFL which has a total facility opening cost of $F$ and a total connection cost of $C$, we say that it is an LMP $\alpha$-approximation if
\begin{equation}
\label{eq:int:lmp:r}
\alpha \cdot F + C \leq \alpha \cdot \foptf_\lambda.
\end{equation}
For a solution $\{y_i, x_{j \to i}\}$, (\ref{eq:int:lmp:r}) can be stated equivalently as $$\alpha \cdot \lambda \sum_{i \in V} y_i + \sum_{i,j \in V} w(j)d(i,j) \cdot x_{j \to i} \leq \alpha \cdot \foptf_\lambda.$$




\medskip
\noindent\textbf{LMP Algorithms.}
We say that an algorithm $\Afl$ for UFL is LMP $\alpha$-approximate if it always returns an LMP $\alpha$-approximation.


\medskip
\noindent\textbf{Connection to $k$-Median.}
The following lemma shows a simple connection between LMP $\alpha$-approximations to UFL and $\alpha$-approximate solutions to $k$-median. This key insight allows us to use LMP algorithms to solve the $k$-median problem \cite{JainV01}.




\begin{lemma}\label{lem:simple LMP case}
    Fix a metric space $(V,w,d)$. Let $\{y_i, x_{j \to i}\}$ be an LMP $\alpha$-approximation for UFL with facility opening cost $\lambda > 0$ that opens exactly $k$ facilities, i.e. $k = \sum_{i \in V} y_i$. Then $\{y_i, x_{j \to i}\}$ is an $\alpha$-approximate solution to the $k$-median problem.
\end{lemma}

\begin{proof}
    Let $\{y^\star_i, x^\star_{j \to i}\}$ be an optimal solution to the $k$-median problem. By using the fact that $\{y_i, x_{j \to i}\}$ is LMP $\alpha$-approximate and treating $\{y^\star_i, x^\star_{j \to i}\}$ as a solution to UFL, we get that
    $$ \alpha \lambda \cdot \sum_{i \in V} y_i + \sum_{i,j \in V} w(j)d(i,j) \cdot x_{j \to i} \leq \alpha \cdot \foptf_\lambda \leq \alpha \cdot \left( \lambda \cdot \sum_{i \in V} y^\star_i + \sum_{i,j \in V} w(j)d(i,j) \cdot x^\star_{j \to i} \right) $$
    which implies that
    $$   \sum_{i,j \in V} w(j)d(i,j) \cdot x_{j \to i} \leq \alpha \cdot \sum_{i,j \in V} w(j)d(i,j) \cdot x^\star_{j \to i} + \alpha \lambda \cdot \left( \sum_{i \in V} y^\star_i - \sum_{i \in V} y_i\right) \leq \alpha \cdot \sum_{i,j \in V} w(j)d(i,j) \cdot x^\star_{j \to i}, $$
    where we are using the fact that $\sum_{i \in V} y_i = k$ and $\sum_{i \in V} y^\star_i \leq k$. Since $\{y_i, x_{j \to i}\}$ is a feasible $k$-median solution, the lemma follows.
\end{proof}

\subsection{Dual Fitting}

The reason that we define LMP approximations with respect to the fractional optimum is so that we can analyse our algorithms via dual fitting. In particular, we often take convex combinations of feasible dual solutions while analysing our algorithms. We use the following lemma throughout our paper.

\begin{lemma}\label{lem:convex}
Let $(V,w,d)$ be a metric space. Let $\{v^1_i, w^1_{j \to i}\}$ and $\{v^2_i, w^2_{j \to i}\}$ be dual solutions with facility opening costs $\lambda_1$ and $\lambda_2$ respectively. Let $\alpha_1, \alpha_2 \geq 0$ such that $\alpha_1 + \alpha_2 \leq 1$. For all $i,j \in V$ let
$$\tilde v_{i}  :=  \alpha_1 \cdot v_{i}^{1} + \alpha_{2} \cdot v_{i}^{2} \qquad \textrm{ and } \qquad \tilde w_{j \to i} := \alpha_1 \cdot w_{j \to i}^{1} + \alpha_{2} \cdot w_{j \to i}^{2}.$$
Then $\{\tilde v_i, \tilde w_{j \to i}\}$ is a dual solution with facility opening cost $\alpha_1\lambda_1 + \alpha_2\lambda_2$.
\end{lemma}

\begin{proof}
    To see that (\ref{eq:dual:lp:1:r}) is satisfied, note that for every $i \in V$ 
    $$\sum_{j \in V} \tilde w_{j \to i} = \alpha_1 \cdot \sum_{j \in V} \tilde w^1_{j \to i} + \alpha_2 \cdot \sum_{j \in V} \tilde w^2_{j \to i} \leq \alpha_1 \lambda_1 + \alpha_2 \lambda_2.$$
    Similarly, for (\ref{eq:dual:lp:2:r}), note that for every $i,j \in V$ 
    $$\tilde v_j = \alpha_1 v^1_j + \alpha_2 v^2_j \leq  \alpha_1 \cdot (w(j)d(i,j) + w^1_{j \to i}) + \alpha_2 \cdot (w(j)d(i,j) + w^2_{j \to i})$$
    $$= (\alpha_1 + \alpha_2) \cdot w(j)d(i,j) + \alpha_1 w^1_{j \to i} + \alpha_2 w^2_{j \to i} \leq w(j)d(i,j) + \tilde w_{j \to i}. $$
\end{proof}
\noindent Completely analogous results hold for taking convex combinations of fractional solutions for $k$-median and UFL.

\subsection{The Dynamic Setting}

Our results focus on designing dynamic algorithms that are capable of either explicitly maintaining solutions or answering certain queries for some weighted metric space that evolves over time. More precisely, we assume that there is some (possibly infinite) underlying metric space $(\mathcal V,d)$, and that as input we are given a sequence of \emph{updates} (weighted point insertions and deletions) $\sigma_1, \sigma_2, \dots$ which define some initially empty weighted metric subspace $(V,w,d)$ of $(\mathcal V, d)$ that changes over time. After each update (i.e. every time a point is either inserted or deleted) we must update the data structure accordingly. The time taken to update the data structure is called the \emph{update time}. The time taken to answer a query supported by the data structure is known as the \emph{query time}.

\subsection{Organization of \Cref{part:II}}\label{sec:P3 org}

In \Cref{part:II}, we describe our dynamic algorithm for maintaining an approximation to the cost of the optimal solution for the $k$-median problem and prove the following theorem.

\begin{theorem}
    There exists is a randomized dynamic algorithm that, given any $k \in \mathbb N$, maintains a $O(1)$-approximation to the cost of the optimal integral solution for $k$-median with $\tilde O(k)$ amortized update time.
\end{theorem}

\noindent In \Cref{sec:DS}, we describe the core data structure that we use to implement our dynamic
algorithms. In \Cref{sec:frac LMP alg}, we give our LMP algorithm for uniform facility location. In \Cref{sec:frac k med}, we give our dynamic algorithm for fractional $k$-median and $k$-median value. Finally, in \Cref{sec:frac sparsifier}, we show how to use sparsification in order to speed up our dynamic algorithms.


\section{The Core Data Structure}\label{sec:DS}

In this section, we describe the core data structure $\RadiiMP$ that we use to implement our dynamic algorithms. This data structure can be updated efficiently and implements fast queries that allow us to access values that are used by algorithms that we design in the subsequent sections.

\subsection{The Radii $r^{(\beta)}_i$}

Fix a metric space $(V,w,d)$ and a facility opening cost $\lambda > 0$. For any $\beta \geq 0$, we let $d^{(\beta)}$ be a new metric on $V$ such that $d^{(\beta)}(i, j) := \beta \cdot d(i, j)$ for all $i, j \in V$. In other words, $d^{(\beta)}$ is obtained by scaling the distances of $d$ by a factor of $\beta$.
Next, for all $i \in V$ and $r, \beta \geq 0$, let $B^{(\beta)}(i, r) := \{ j \in V \; | \; d^{(\beta)}(i, j) \leq r\}$ denote the (closed) ball of radius $r$ around the point $i$ w.r.t.~the metric $d^{(\beta)}$. Finally, for all $i \in V$ and $\beta \geq 0$, we define
\begin{equation}
\label{eq:def:r}
r^{(\beta)}_i := \min \left\{ r \geq 0 \; \middle| \; \sum_{j \in B^{(\beta)}(i, r)} w(j)\left(r - d^{(\beta)}(i, j) \right) \geq \lambda\right\}.
\footnote{To see that $r^{(\beta)}_j$ is well defined, note that this set is the preimage of a non-negative continuous function on a closed set, and hence must contain its infimum.}
\end{equation}

\noindent
Intuitively, we can interpret the value $r_i^{(\beta)}$ as follows. Keep growing a ball around $i$ in the metric space $(V, w, d^{(\beta)})$. If we take a snapshot of this process when the ball is of radius $r$, then at that time each point $j$ in this ball contributes $w(j)(r - d^{(\beta)}(i, j))$ towards the uniform facility opening cost of $\lambda$. We stop growing this ball when the sum of these contributions becomes equal to $\lambda$, and $r^{(\beta)}_i$ is precisely the radius at which stop this process.
Using the radii $r^{(\beta)}_i$, we also define
\begin{equation}
\label{eq:def:C}
C^{(\beta)}_i := \frac{1}{\lambda} \cdot \sum_{j \in B^{(\beta)}(i, r^{(\beta)}_i)} w(j) \left( r^{(\beta)}_i - d^{(\beta)}(i,j) \right) \cdot w(i)d^{(\beta)}(i,j).
\end{equation}
The radii $\{r^{(\beta)}_i\}$ and the values $\{C^{(\beta)}_i\}$ are central objects in the algorithms that we design. We now give a dynamic data structure that allows us to compute these values efficiently.

\subsection{The Data Structure \RadiiMP}

The data structure $\RadiiMP$ is initialized with a parameter $\beta \geq 0$.
Given a dynamic metric space $(V,d, w)$ which evolves via a sequence of point insertions and deletions, the data structure $\RadiiMP$ supports the following query operations.
\begin{itemize}
    \item \textsc{Radius}$(i, \lambda)$: The input to this query is a point $i \in V$ and a facility opening cost $\lambda > 0$. In response, the data structure outputs the value of $r^{(\beta)}_i$ defined with respect to facility opening cost $\lambda$.
    \item \textsc{Connection-Cost}$(i, \lambda)$: The input to this query is a point $i \in V$ and a facility opening cost $\lambda > 0$. In response, the data structure outputs the value of $C^{(\beta)}_i$ defined with respect to facility opening cost $\lambda$.
\end{itemize}
The following lemma summarises the behaviour of this data structure.

\begin{lemma}\label{lem:maintaining r}
    The data structure $\RadiiMP$ is deterministic, has a worst-case update time of $O(n \log n)$, and can answer \textsc{Radius} and \textsc{Connection-Cost} queries in $O(\log n)$ time.
\end{lemma}

\noindent
We now show how to implement $\RadiiMP$ to get the guarantees described in \Cref{lem:maintaining r}.

\subsection{Binary Search Data Structure}\label{sec: search DS}

Suppose we have an initially empty dynamic set $\mathcal S$ which undergoes point insertions and deletions and values $d_x \geq 0$ and $q_x \geq 0$ for each $x \in S$ which are specified at the time that $x$ is inserted into $\mathcal S$.
We now show how to implement a deterministic data structure that dynamically maintains this set $\mathcal S$ and the corresponding values $\{d_x, q_x\}$, supports useful queries, and can be updated efficiently. We maintain a \emph{balanced binary search tree} $\mathcal T$, where each node in $\mathcal T$ corresponds to a point $x \in \mathcal S$ and the values are sorted according to the $d_x$ (so that an in-order traversal of the tree $\mathcal T$ yields the sequence $x_1,\dots, x_n$ where $d_{x_1} \leq \dots \leq d_{x_n}$). Given a node $u\in \mathcal T$, we let $x_u \in \mathcal S$ denote the point corresponding to $u$ and we abbreviate $d_{x_u}$ and $q_{x_u}$ by $d_u$ and $q_u$ respectively.
Given nodes $u, v \in \mathcal T$, we write $u \prec v$ if and only if $u$ appears before $v$ when performing an in-order traversal of $\mathcal T$. We denote by $\mathcal T_u$ the subtree of $\mathcal T$ rooted at $u$, and define $q(\mathcal T_u)$ to be $\sum_{v \in \mathcal T_u} q_v$. Furthermore, we denote by $u^-$ and $u^+$ the minimum and maximum nodes in $\mathcal T_u$ with respect to the ordering $\prec$. 
For $u \in \mathcal T$, $\eta \in \mathbb N$, we define
$$\phi(u) := \sum_{v \in \mathcal T_u} q_v(d_{u^+} - d_v) \text{\qquad and \qquad} \phi_\eta(u) := \phi(u) + \eta \cdot (d_{u^+} - d_{u^-}).$$
Furthermore, for $v \in \mathcal T_u$, we also define
$$\phi(u,v) := \sum_{w \in \mathcal T_u, \; w \preceq v}q_w(d_{v} - d_w) \text{\qquad and \qquad} \phi_\eta(u,v) := \phi(u,v) + \eta \cdot (d_{v} - d_{u^-}).$$
The data structure supports the following query operations.

\begin{itemize}
    \item \textsc{Mass-Query}$(u, \eta)$: The input to this query is a node $u \in \mathcal T$ and an integer $\eta$. In response, the data structure outputs the value of $\phi_\eta(u)$.
    \item \textsc{Search-Query}$(u, \mu, \eta)$: The input to this query is a node $u \in \mathcal T$, a real number $\mu > 0$ and an integer $\eta$. In response, the data structure outputs the tuple $(u^\star, \phi_\eta(u, u^\star))$, where $$u^\star := \max \{v \in \mathcal T_u \, | \, \phi_\eta(u, v) < \mu\}.$$
\end{itemize}

\noindent We denote this data structure by $\mathcal R$. In \Cref{sec: search DS implementation}, we show how to implement this data structure so that it has a worst-case update and query time of $O(\log n)$.
In \Cref{sec: search DS maintaining r}, we then use this data structure to implement $\RadiiMP$.

\subsection{Implementation of Binary Search Data Structure}\label{sec: search DS implementation}

We implement the tree $\mathcal T$ so that each node $u$ in the tree stores the point $x_u$, the value $d_u$, the value $q_u$, the value of $\phi(u)$, the value of $q(\mathcal T_u)$, pointers to its left and right child (if it has any), and pointers to $u^-$ and $u^+$.
We now show how to efficiently implement each query, followed by how to efficiently update the data structure.

\medskip
\noindent \textbf{Implementing \textnormal{\textsc{Mass-Query}}.}
Given a node $u$, we can retrieve the value of $\phi(u)$ in $O(1)$ time. Using the pointers to $u^-$ and $u^+$, we can retrieve the values of $d_{u^-}$ and $d_{u^+}$ in $O(1)$ time. Hence, we can compute $\phi_\eta(u) = \phi(u) + (d_{u^+} - d_{u^-})\cdot \eta$ in $O(1)$ time.
It follows that we can implement \textsc{Mass-Query} to run in $O(1)$ time.

\medskip
\noindent \textbf{Implementing \textnormal{\textsc{Search-Query}}.} Let $u$ be a node in $\mathcal T$. Let $v$ and $w$ be the left and right children of $u$ respectively. If $u$ has no left (resp.~right) child, we set $v$ (resp.~w) to $\texttt{null}$. Note that $u^- = v^-$, $u^+ = w^+$, and $u^- \prec \dots \prec v^+ \prec u \prec w^- \prec \dots \prec u^+$.
We implement \textsc{Search-Query}$(u, \mu, \eta)$ recursively using Algorithm~\ref{alg: search query}.

\begin{algorithm}[H]\label{alg: search query}
    \SetAlgoLined
    \DontPrintSemicolon

    \If{ $v \neq \textnormal{\texttt{null}}$ }{
        \If{ $\phi_\eta(u, v^+) \geq \mu$ }{
        \Return $\textsc{Search-Query}(v, \mu, \eta)$\;
        }
        \If{ $\phi_\eta(u,u) \geq \mu$ }{
            \Return $(v^+, \phi_\eta(u, v^+))$\;
        }
    }
    \If{ $w = \textnormal{\texttt{null}}$ \textnormal{\textbf{or}} $\phi_\eta(u, w^-) \geq \mu$ }{
        \Return $(u, \phi_\eta(u, u))$\;
    }
    $\mu' \leftarrow \mu - \phi_\eta(u, w^-)$\;
    $\eta' \leftarrow \eta + q_u + \sum_{v' \in \mathcal T_v} q_{v'}$\;
    \Return $\textsc{Search-Query}(w, \mu', \eta')$\;
    \caption{\textsc{Search-Query}$(u, \mu, \eta)$}
\end{algorithm}

\begin{claim}
    Algorithm~\ref{alg: search query} returns the correct response to the query \textsc{Search-Query}$(u, \mu, \eta)$.
\end{claim}

\begin{proof}
    We prove this claim by induction on the depth of the subtree rooted at the node $u$. For the base case where $u$ has a subtree of depth 1 (i.e.~$u$ is a leaf), we know that $v = w = \texttt{null}$, so Algorithm~\ref{alg: search query} returns $(u, \phi_\eta(u,u))$ as required.

    Now, for the inductive step, suppose that the subtree rooted at $u$ has depth $\ell$ and that Algorithm~\ref{alg: search query} returns the correct response when the subtree rooted at the input node has depth less than $\ell$. First, consider the case that $u^\star \prec u$. Then it must be the case that $v \neq \texttt{null}$. Hence, the algorithm either returns the output of $\textsc{Search-Query}(v, \mu, \eta)$ if $u^\star \prec v^+$ or $(v^+, \phi_\eta(u, v^+))$ if $u^\star = v^+$. If $u^\star = u$, then our algorithm returns $(u, \phi_\eta(u, u))$. Otherwise, $u^\star \succ u$, and we recurse on $w$. In this case, we must appropriately modify $\mu$ and $\eta$. We do this by noting that, for $w' \succ w^-$, 
    $$\phi_\eta(u, w') = \phi_\eta(u, w^-) + \phi_{\eta + q_u + \sum_{v' \in \mathcal T_v} q_{v'}}(w, w').$$ Hence, $\phi_\eta(u, w') \geq \mu$ iff $\phi_{\eta + q_u + \sum_{v' \in \mathcal T_v} q_{v'}}(w, w') \geq \mu - \phi_\eta(u, w^-)$. Note that $\mu - \phi_\eta(u, w^-) > 0$ since $u^\star > u$.
\end{proof}

\begin{claim}
    Algorithm~\ref{alg: search query} can be implemented to run in $O(\log n)$ time.
\end{claim}

\begin{proof}
    In order to implement this algorithm to run in $O(\log n)$ time, we argue that there are at most $O(\log n)$ levels in the recursion and that we spend at most $O(1)$ time at each level. First note that at most one recursive call is made per call to Algorithm~\ref{alg: search query} and that the subtree of the node in the recursive call has strictly smaller depth. Hence, we make at most $O(\log n)$ calls in total since the tree is balanced. If we can find the values of $\phi_\eta(u, v^+)$, $\phi_\eta(u, u)$, and $\phi_\eta(u, w^-)$ in $O(1)$ time, then all other operations can clearly be implemented in $O(1)$ time and the claimed running time follows.
    It is sufficient to compute $\phi(u, v^+)$, $\phi(u, u)$, and $\phi(u, w^-)$.
    Since $\phi(u, v^+) = \phi(v, v^+) = \phi(v)$ and we have access to the node $v$, we can compute $\phi(u, v^+)$ in $O(1)$ time. Similarly, by noting that $\phi(u,u) = \phi(u, v^+) + (d_u - d_{v^+}) \cdot q(\mathcal T_v)$ and $\phi(u, w^-) = \phi(u,u) + (d_{w^-} - d_u) \cdot (q(\mathcal T_v) + q_u)$, we can also compute $\phi_\eta(u, u)$ and $\phi_\eta(u, w^-)$ in $O(1)$ time.
\end{proof}

\medskip
\noindent \textbf{Updating the Data Structure.} After an update, the tree $\mathcal T$ may no longer be balanced, and some of the nodes $u$ in $\mathcal T$ may no longer store the correct values of $\phi(u)$ and $q(\mathcal T_u)$, and may no longer be storing the correct pointers to $u^-$ and $u^+$ (since these nodes might change). Using a standard implementation such as an AVL tree, we can re-balance the tree in $O(\log n)$ time. In order to efficiently update the information stored at the nodes, we first note that there are at most $O(\log n)$ nodes in the tree whose subtrees change as the result of an update and the proceeding balancing operation, and that the information stored at a node is completely determined by its subtree. We scan through all such nodes $u$, starting with the ones at the bottom of the tree, and update these values in the following manner. Suppose $u$ is an internal node in $\mathcal T$ which stores the point $x_u$, and let $v$ and $w$ be the left and right children of $u$ respectively. Then we have that $q(\mathcal T_u) = q_u + q(\mathcal T_v) + q(\mathcal T_w)$, so we can find $q(\mathcal T_u)$ in $O(1)$ time.
Furthermore, we have that $u^- = v^-$, $u^+ = w^+$, and
$$ \phi(u) = \phi(v) + \phi(w) + q(\mathcal T_v) (d_{w^+} - d_{v^+}) + q_u(d_{w^+} - d_u). $$
Note that if $u$ has no left (resp. right) child, then $u^- = u$, $u^+ = w^+$, and $\phi(u) = \phi(w) + q_u(d_{w^+} - d_u)$ (resp. $u^- = v^-$, $u^+ = u$, and $\phi(u) = \phi(v) + q(\mathcal T_v)(d_{u} - d_{v^+}$)). Hence, we can update all of the information stored at the node $u$ in $O(1)$ time as long as the correct information is stored at the nodes $v$ and $w$. In total, it takes $O(\log n)$ time to update the tree $\mathcal T$.

\subsection{Implementation of \RadiiMP}\label{sec: search DS maintaining r}

We now show how to use this binary search data structure $\mathcal R$ in order to implement $\RadiiMP$. Let $(V,w,d)$ be the dynamic metric space which evolves via a sequence of point insertions and deletions. At each point in time, we maintain a data structure $\mathcal R_i$ for each $i \in V$. Furthermore, for each $i \in V$, $\mathcal R_i$ contains each point $j \in V$ with $d_j = d^{(\beta)}(i, j)$ and $q_j = w(j)$. Upon the insertion of a point $i$, we add $i$ to $\mathcal R_j$ with $d_i = d^{(\beta)}(i, j)$ and $q_i = w(i)$ for all $j \in V$, create a new data structure $\mathcal R_i$, and insert each $j \in V$ into $\mathcal R_i$ with $d_j = d^{(\beta)}(i, j)$ and $q_j = w(j)$ for all $j \in V$. This entire process takes $O(n \log n)$. Similarly, upon the deletion of a point $i \in V$, we remove $i$ from $\mathcal R_j$ for all $j \in V$ and delete the data structure $\mathcal R_i$, again taking $O(n \log n)$ time.

\medskip
\noindent \textbf{Computing $r^{(\beta)}_i$.} To see how we can obtain the value of $r^{(\beta)}_i$ from $\mathcal R_i$, note that calling \textsc{Search-Query}$(u_i, \lambda, 0)$ on $\mathcal R_i$, where $u_i$ is the root of the tree in the data structure $\mathcal R_i$, returns the node $u^\star$ corresponding to the point $i^\star$ and the value $\gamma$ such that
$$ \gamma = \sum_{j \in B^{(\beta)}(i, r)} w(j) \! \left(r - d^{(\beta)}(i, j) \right) < \lambda,$$
where $r = d^{(\beta)}(i^\star, i)$, and for any point $i'$ such that $d^{(\beta)}(i', i) > d^{(\beta)}(i^\star, i)$, we have that 
$$\sum_{j \in B^{(\beta)}(i, d^{(\beta)}(i', i))} w(j) \! \left(d^{(\beta)}(i', i) - d^{(\beta)}(i, j) \right) \geq \lambda.$$
It follows that
$$\sum_{j \in B^{(\beta)}(i, r)} w(j)\! \left(r^{(\beta)}_i - d^{(\beta)}(i, j) \right) = \sum_{j \in B^{(\beta)}(i, r^{(\beta)}_i)} w(j)\! \left(r^{(\beta)}_i - d^{(\beta)}(i, j) \right) = \lambda,$$
and so $r^{(\beta)}_i = r + (\lambda - \gamma)/ \sum_{j \in B^{(\beta)}(i, r)} w(j)$. Since $\mathcal R_i$ maintains a balanced tree of the points sorted by their distance from $i$, given the node that stores the point $i^\star$, we can compute $\sum_{j \in B^{(\beta)}(i, r)} w(j)$ in $O(\log n)$ time.
In order to do this, first note that  $\sum_{j \in B^{(\beta)}(i, r)} w(j) = \sum_{v \preceq u^\star} q_{v}$.\footnote{Here, the ordering $\preceq$ is taken with respect to the tree stored by $\mathcal R_i$.} Then, using the tree $\mathcal T$ stored by $\mathcal R_i$, we compute this value recursively by calling Algorithm~\ref{alg: weight mass} on the input $(u_i, u^\star)$. It's easy to see that this algorithm runs in $O(\log n)$ time and returns the correct value.
Hence, given the node $u^\star$ and the value of $\gamma$, we can find $r^{(\beta)}_i$ in $O(\log n)$ time.

\begin{algorithm}[H]\label{alg: weight mass}
    \SetAlgoLined
    \DontPrintSemicolon
    Let $v$ and $w$ be the left and right children of $u$ respectively\;

    \If{ $u \preceq u^\star$ }{
        \Return $q_{u^\star} + q(\mathcal T_v) + \textsc{Compute-Weight}(w, u^\star)$\;
    }
    \Return \textsc{Compute-Weight}$(v, u^\star)$\;
    \caption{\textsc{Compute-Weight}$(u, u^\star)$}
\end{algorithm}

\medskip
\noindent \textbf{Computing $C^{(\beta)}_i$.} We can easily extend the data structure so that we can also find $C^{(\beta)}_i$ in $O(\log n)$ time. We do this by modifying the balanced binary tree $\mathcal T$ so that, at each node $u \in \mathcal T$, it stores 
$$\psi(u) := \sum_{v \in \mathcal T_u} q_v d_v \text{\qquad and \qquad} \psi^2(u) := \sum_{v \in \mathcal T_u} q_v d_v^2.$$
It's easy to see that we can also efficiently update these values in the same way as the values of $q(\mathcal T_u)$ stored by the nodes. Furthermore, by using analogous algorithms to the one presented in Algorithm~\ref{alg: weight mass}, we can easily use these values to find $\sum_{v \preceq u} q_v d_v$ and $\sum_{v \preceq u} q_v d_v^2$ in $O(\log n)$ for any node $u$ in $\mathcal T$.

By noticing that the point $i^\star$ corresponding to the node $u^\star$ returned from the call to \textsc{Search-Query}$(u_i, \lambda, 0)$ on $\mathcal R_i$ satisfies
$$C^{(\beta)}_i = \frac{1}{\lambda} \cdot \sum_{j \in B^{(\beta)}(i, r^{(\beta)}_i)} w(j) \! \left(r^{(\beta)}_i - d^{(\beta)}(i, j) \right) \cdot w(i)d^{(\beta)}(i,j)  $$
$$ = \frac{w(i)}{\lambda} \cdot \left( r^{(\beta)}_i \cdot \sum_{j \in B^{(\beta)}(i, r^{(\beta)}_i)} w(j) d^{(\beta)}(i, j) - \sum_{j \in B^{(\beta)}(i, r^{(\beta)}_i)} w(j) \! \left(d^{(\beta)}(i,j)\right)^2 \right)$$ 
$$= \frac{w(i)}{\lambda} \cdot \left( r^{(\beta)}_i \cdot \sum_{v \preceq u^\star} q_v d_v -  \sum_{v \preceq u^\star} q_v d_v^2 \right), $$
we can find $C^{(\beta)}_i$ in $O(\log n)$ time after finding the value of $r^{(\beta)}_i$. 


\section{An LMP Algorithm for UFL}\label{sec:frac LMP alg}

In this section, we present the dynamic data structure \textsc{FracLMP}. Given some dynamic metric space $(V, w, d)$ undergoing point insertions and deletions, the data structure \textsc{FracLMP} supports queries that, given any facility opening cost $\lambda > 0$, return an LMP $4$-approximation for UFL with facility opening cost $\lambda$. In order to implement these queries, we use a modified variant of an algorithm by Mettu-Plaxtion~\cite{MettuP00}, which we call the fractional MP algorithm. We show that this gives a simple fractional algorithm for UFL, which is LMP $4$-approximate. We note that this LMP variant is very similar to the one in \cite{esa/ArcherRS03}.

\subsection{The Fractional Mettu-Plaxton Algorithm}

Let $(V,w,d)$ be a metric space. The fractional MP algorithm constructs an LMP $4$-approxiamtion for UFL with facility opening cost $\lambda$ as described in the following pseudocode.

\begin{algorithm}[H]\label{alg: fractional MP}
    \SetAlgoLined
    \DontPrintSemicolon
    For $i \in V$, set $y_i \leftarrow (1 / \lambda) \cdot w(i) \, r^{(1/4)}_i$\;
    For $i,j \in V$, set $x_{j \to i} \leftarrow (1 / \lambda) \cdot \max\left(0, w(i) \! \left( r^{(1/4)}_j - d^{(1/4)}(i, j) \right)\right)$\;
    \caption{\textsc{FractionalMP}$(V,d,w,\lambda)$}
\end{algorithm}


\begin{theorem}
\label{th:MP:primary}
Algorithm~\ref{alg: fractional MP} returns a fractional LMP $4$-approximation for UFL. 
\end{theorem}

\subsection{Analysis of Algorithm~\ref{alg: fractional MP}}

To analyze Algorithm~\ref{alg: fractional MP} and show that the solution $\{y_i, x_{j \to i}\}$ is LMP $4$-approximate, we construct a feasible dual solution and apply dual fitting.
Towards this end, for every $i, j \in V$ and $\beta \geq 0$, define $w^{(\beta)}_{j \to i} := \max\left(0, w(i) \! \left( r^{(\beta)}_j - d^{(\beta)}(i, j) \right)\right)$.\footnote{Note that these correspond to the dual variables $w_{j \to i}$, as defined in~\Cref{sec:LPs}.} Note that $x_{j \to i} = (1 / \lambda) \cdot w^{(1/4)}_{j \to i}$. We can now make the following observations.

\begin{observation}\label{obs: MP anal 1}
    For all $j \in V$, $\beta \geq 0$, we have that $\sum_{i \in B^{(\beta)}(j, r^{(\beta)}_j)} w^{(\beta)}_{j \to i} = \lambda$.
\end{observation}

\begin{proof}
    Follows directly from the definitions of $w^{(\beta)}_{j \to i}$ and $r^{(\beta)}_i$.
\end{proof}

\begin{observation}\label{lem: r <= r + d}
    For all $i,i' \in V$, $\beta \geq 0$, we have that $r^{(\beta)}_{i} \leq r^{(\beta)}_{i'} + d^{(\beta)}(i,i')$.
\end{observation}

\begin{proof}
    Let $i,i' \in V$. Since $d^{(\beta)}$ is a metric, we have that $B^{(\beta)} (i, r + d^{(\beta)}(i, i') ) \supseteq B^{(\beta)} (i', r )$ and $d^{(\beta)}(i, i') - d^{(\beta)}(i, j) \geq - d^{(\beta)}(i', j)$ for all $j \in V$, $r \geq 0$. It follows that, for any $r \geq 0$
    $$ \sum_{j \in B^{(\beta)}(i, r + d^{(\beta)}(i, i'))} w(j) \!\left (r  + d^{(\beta)}(i, i') - d^{(\beta)}(i, j) \right) \geq \sum_{j \in B^{(\beta)}(i', r)} w(j) \! \left (r - d^{(\beta)}(i', j) \right). $$
    Hence, $r^{(\beta)}_i \leq r^{(\beta)}_{i'} + d^{(\beta)}(i,i')$.
\end{proof}

\begin{observation}\label{obs: MP anal 2}
    For all $i,j \in V$, $\beta \geq 0$, we have that $w(i)r^{(\beta)}_i \geq w^{(\beta)}_{j \to i}$.
\end{observation}

\begin{proof}
    If $r^{(\beta)}_j - d^{(\beta)}(i, j) \leq 0$, then $w^{(\beta)}_{j \to i} = 0 \leq r^{(\beta)}_j$. If $r^{(\beta)}_j - d^{(\beta)}(i, j) > 0$, then, by \Cref{lem: r <= r + d}, $w^{(\beta)}_{j \to i} = w(i)(r^{(\beta)}_j - d^{(\beta)}(i, j)) \leq w(i)(r^{(\beta)}_i + d^{(\beta)}(i, j) - d^{(\beta)}(i, j)) = w(i)r^{(\beta)}_i$.
\end{proof}

\noindent
It follows from Observations~\ref{obs: MP anal 1} and \ref{obs: MP anal 2} that $\{y_i, x_{j \to i}\}$ is a feasible solution to UFL.
In particular, by \Cref{obs: MP anal 1}, each point $j \in V$ has sufficient mass within its ball $B^{(1/4)}(j, r_j^{(1/4)})$ to get fully assigned (fractionally) to points within the concerned ball, so $\sum_{i \in V} x_{j \to i} = 1$, and by \Cref{obs: MP anal 2}, $y_i \geq x_{j \to i}$ for all $i,j \in V$.

For all $j \in V$ and  $\beta \geq 0$, we define $C_j^{(\beta)} := (1/\lambda) \cdot \sum_{i \in V} w^{(\beta)}_{j \to i} \cdot w(j) d^{(\beta)}(i, j)$. Thus, the connection cost incurred by a given client $j \in V$ under the MP algorithm (in the original metric $d$) equals $4 C^{(1/4)}_j = \sum_{i \in V} x_{j \to i} \cdot w(j) d(i,j)$. It follows that
the total cost incurred by the MP algorithm equals $\sum_{j \in V} \left( w(j)r_j^{(1/4)} + 4 C_j^{(1/4)}\right)$.


\subsubsection{Proof of \Cref{th:MP:primary}}
The proof of \Cref{th:MP:primary} now follows from the two lemmas stated below, which we prove in~\Cref{sec:lm:dual:fitting:1} and~\Cref{sec:lm:dual:fitting:2} respectively.

\begin{lemma}
\label{lm:dual:fitting:1}
For all $i, j \in V$, let $\tilde v_j := w(j)  r^{(1/2)}_j$ and $\tilde  w_{j \to i} := \max\left(0, w(j) \!\left(r^{(1/2)}_j - d(i, j)\right)\right)$. Then $\{ \tilde v_j, \tilde w_{j \to i} \}$ is a feasible dual solution.
\end{lemma}

\begin{lemma}
\label{lm:dual:fitting:2}
For each $j \in V$, we have $w(j)r_j^{(1/4)} + C_j^{(1/4)} \leq w(j)r_j^{(1/2)}$. 
\end{lemma}

\noindent It follows from Lemma~\ref{lm:dual:fitting:2} that
$$ 4 \lambda \cdot \sum_{i \in V} y_i + \sum_{i,j \in V} w(j) d(i,j) \cdot x_{j \to i} = 4 \cdot \sum_{j \in V} \left( w(j)r_j^{(1/4)} + C_j^{(1/4)}\right) \leq 4 \cdot \sum_{j \in V} w(j)r^{(1/2)}_j.$$
By Lemma~\ref{lm:dual:fitting:1}, $\sum_{j \in V} w(j) r^{(1/2)}_j$ is the objective of a feasible dual solution and hence is at most $\foptf_\lambda$. It follows that $\{y_i, x_{j \to i}\}$ is a fractional LMP $4$-approximation.

\subsubsection{Proof of~\Cref{lm:dual:fitting:1}}
\label{sec:lm:dual:fitting:1}

In order to show that $\{ \tilde v_j, \tilde w_{j \to i} \}$ is a feasible dual solution, we show that these variables satisfy constraints $(\ref{eq:dual:lp:1:r})$, $(\ref{eq:dual:lp:2:r})$, and $(\ref{eq:dual:lp:3:r})$. Clearly, the variables are all non-negative, and
$$ \tilde w_{j \to i} + w(j)d(i, j) = \max\left(0, w(j) \!\left(r^{(1/2)}_j - d(i, j)\right)\right) + w(j)d(i, j)$$ 
$$\geq w(j) \! \left(r^{(1/2)}_j - d(i, j)\right) + w(j)d(i, j) = w(j)r^{(1/2)}_j = \tilde v_j, $$
so constraints $(\ref{eq:dual:lp:2:r})$ and $(\ref{eq:dual:lp:3:r})$ are satisfied. To see that $(\ref{eq:dual:lp:1:r})$ is satisfied, we first note that, by Observation~\ref{lem: r <= r + d}, for all $i \in V$ we have that
$$r^{(1/2)}_{j} - d(i, j) = r^{(1/2)}_{j} - 2 \cdot d^{(1/2)}(i, j) \leq r^{(1/2)}_{i} - d^{(1/2)}(i,j).$$
It follows that
$$ \sum_{j \in V} \tilde w_{j \to i} = \sum_{j \in V} \max\left(0, w(j) \!\left(r^{(1/2)}_j - d(i, j)\right)\right) \leq \sum_{j \in V} \max\left(0, w(j) \!\left(r^{(1/2)}_i - d^{(1/2)}(i, j)\right)\right) = \lambda. $$

\subsubsection{Proof of~\Cref{lm:dual:fitting:2}}
\label{sec:lm:dual:fitting:2}

Throughout this proof, we fix a point $j \in V$.
To simplify notation, we will write $r_j, 
C_j$ instead of $r_j^{(1/2)}, 
C_j^{(1/2)}$. 
Let $V^\star := B(j, r_j)$. Consider a new metric space $(V^\star, d^\star)$, where $d^\star(i, j) = d(i, j)/4$ for all $i, j \in V^\star$.\footnote{$d^\star$ is the metric $d^{(1/4)}$ restricted to $V^\star$}
Let $r^{\star}_j, 
C^{\star}_j$ be the values of $r_j,
C_j$ defined with respect to the metric space $(V^{\star}, d^{\star})$. At the end of \Cref{sec:lm:dual:fitting:2}, we show that $w(j)r_j^{(1/4)} + C_j^{(1/4)} \leq w(j)r^{\star}_j + C^{\star}_j$ (see \Cref{lm:monotone}). Thus, from this point onward, our goal is to show that
\begin{equation}
\label{eq:toshow}
w(j)r^{\star}_j + C^{\star}_j \leq w(j)r_j.
\end{equation}
We start by making the following observations.
\begin{eqnarray}
\label{eq:toshow:1}
\sum_{i \in V^{\star}} w(i)\! \left(r_j - \frac{d(i, j)}{2} \right)  =  \lambda & & \\
\sum_{i \in V^{\star}} w(i)\! \left(r_j^{\star} - \frac{d(i, j)}{4} \right)  =  \lambda & &  \label{eq:toshow:2} \\
C^{\star}_j  =  \frac{1}{\lambda} \cdot \sum_{i \in V^{\star}} w(i)\! \left(r_j^{\star} - \frac{d(i, j)}{4} \right) \frac{w(j)d(i, j)}{4} & & \label{eq:toshow:3} 
\end{eqnarray}
Let $m^{\star} = \sum_{i \in V^\star} w(i)$. From~(\ref{eq:toshow:1}),~(\ref{eq:toshow:2}) and~(\ref{eq:toshow:3}), we infer that
\begin{eqnarray}
\label{eq:toshow:4}
r_j & = & \frac{1}{m^{\star}} \left( \lambda + \sum_{i \in V^{\star}} \frac{w(i)d(i, j)}{2} \right) \\
\label{eq:toshow:5}
r^{\star}_j + \frac{C^{\star}_j}{w(j)} & = & \frac{1}{m^{\star} \lambda} \left( \lambda + \sum_{i \in V^{\star}} \frac{w(i)d(i, j)}{4} \right)^2 - \frac{1}{\lambda} \sum_{i \in V^{\star}} w(i)\!\left( \frac{d(i, j)}{4} \right)^2
\end{eqnarray}
Now, we multiply both of the above inequalities by $m^{\star}$, and obtain
\begin{eqnarray}
\label{eq:toshow:6}
m^{\star} r_j & = & \lambda+\sum_{i \in V^{\star}} \frac{w(i)d(i,j)}{2} \\ 
\label{eq:toshow:7}
m^{\star}\!\left( r^{\star}_j + \frac{C^{\star}_j}{w(j)}\right) & = & \frac{1}{\lambda} \left( \lambda+\sum_{i \in V^{\star} } \frac{w(i)d(i, j)}{4} \right)^2 - \frac{m^{\star}}{\lambda} \cdot \sum_{i \in V^{\star}}  w(i)\!\left( \frac{d(i,j)}{4} \right)^2
\end{eqnarray}
From~(\ref{eq:toshow:6}) and~(\ref{eq:toshow:7}), we derive that
\begin{eqnarray}
\lambda m^{\star} \!\left(r^{\star}_j + \frac{C^{\star}_j}{w(j)}\right) - \lambda m^{\star}  r_j 
& = & \left( \sum_{i \in V^{\star}} \frac{w(i)d(i, j)}{4} \right)^2 - m^{\star} \cdot \sum_{i \in V^{\star}} w(i)\!\left( \frac{d(i, j)}{4} \right)^2 \leq 0
\label{eq:toshow:8}
\end{eqnarray}
where (\ref{eq:toshow:8}) follows from applying Cauchy-Schwarz to the vectors $\left\langle\sqrt{w(i)}\right\rangle_i$ and $\left\langle\sqrt{w(i)} d(i,j)/4\right\rangle_i$.
This concludes the proof of~\Cref{eq:toshow}, which, in turn, implies~\Cref{lm:dual:fitting:2}.

\begin{lemma}[Monotonicity Lemma]
\label{lm:monotone}
Consider any metric space $(V,w,d)$, facility opening cost $\lambda > 0$, and any point $i' \in V$. Let $V':= V \setminus \{ i'\}$, and let $r'_j, w'_{j \to i}, C'_j$ be the values of $r_j^{(1)}, w_{j \to i}^{(1)}, C_j^{(1)}$ defined with respect to the metric subspace $(V', d, w)$. Then for all $j \in V'$, we have $w(j)r_j^{(1)} + C_j^{(1)} \leq w(j)r_j' + C_j'$.
\end{lemma}

\begin{proof}

Fix any point $j \in V'$. Throughout this proof, we will refer to $r_j^{(1)}, w_{j \to i}^{(1)}, C_j^{(1)}$ as $r_j, w_{j \to i}, C_j$ respectively. Also, we will let $B'(j, r) := \{ i \in V' : d(i, j) \leq r\}$ denote the ball of radius $r$ around $j$ in the metric space $(V', d, w)$. We will show that $w(j)r_j + C_j \leq w(j)r'_j + C'_j$. In words, this means that as we switch from the input $(V', d, w)$ to $(V,w,d)$ by adding one extra point $i'$, the sum of the (fractional) facility opening cost and connection cost paid by client $j$ can only decrease.

\newcommand{\s}{\texttt{size}}

If $i' \notin B(j, r_j)$, then $r_j = r'_j$ and $C_j = C'_j$, and so the lemma trivially holds. From now onward,  assume that $i' \in B(j, r_j)$. In this case, it is easy to verify that $r_j \leq r'_j$. Let $\chi := r'_j - r_j \geq 0$. Thus, because of the insertion of the point $i'$, the concerned ball around $j$ has shrunk its radius by an additive $\chi$ factor.
Recall that $w_{j \to i} = \max(0, w(i)(r_j - d(i,j)))$. Since
$$\sum_{i \in B'(j, r'_j)} w(i) \! \left(r'_j - d(i, j)\right) = \lambda = \sum_{i \in B(j, r_j)} w(i) \! \left(r_j - d(i, j)\right),$$
we have that
\begin{equation}
\label{eq:derive}
w_{j \to i'}  = \chi \cdot \sum_{i \in B(j, r_j) \setminus \{i' \}} w(i) + \sum_{i \in B'(j, r'_j) \setminus B(j, r_j)} w(i)\!\left( r'_j - d(i, j) \right).
\end{equation}
Finally, let $\xi := (w(j)r'_j + C'_j) - (w(j)r_j + C_j)$.  We now infer that
\begin{equation}
\label{eq:derive2}
\lambda \xi = \lambda w(j) \! \left( r'_j - r_j \right) +   \sum_{i \in B'(j, r'_j)} w'_{j \to i} \cdot w(j)d(i, j) - \sum_{i \in B(j, r_j) \setminus \{i'\}} w_{j \to i} \cdot w(j)d(i, j)     - w_{j \to i'} \cdot w(j)d(i', j).
\end{equation}
By rearranging terms, we get that
$$\sum_{i \in B'(j, r'_j)} w'_{j \to i} \cdot w(j)d(i, j) - \sum_{i \in B(j, r_j) \setminus \{i'\}} w_{j \to i} \cdot w(j)d(i, j) $$
$$=\sum_{i \in B(j, r_j) \setminus \{i'\}} \left( w'_{j \to i} - w_{j \to i} \right) \cdot w(j)d(i, j) + \sum_{i \in B'(j, r'_j) \setminus B(j, r_j)} w'_{j \to i} \cdot w(j)d(i, j) $$
$$ =   w(j)\chi \cdot \sum_{i \in B(j, r_j) \setminus \{i'\}} w(i) d(i, j) + w(j) \cdot \sum_{i \in B'(j, r'_j) \setminus B(j, r_j)} w(i)\!\left( r'_j - d(i, j) \right) \cdot d(i, j)$$
$$ \geq   w(j)\chi \cdot \sum_{i \in B(j, r_j) \setminus \{i'\}} w(i) d(i, j) + w(j)d(i', j) \cdot \sum_{i \in B'(j, r'_j) \setminus B(j, r_j)} w(i)\!\left( r'_j - d(i, j) \right),$$
where the final inequality follows from lower bounding $d(i,j)$ by $d(i', j)$ for all $i \in B'(j, r'_j) \setminus B(j, r_j)$. Now, we can apply (\ref{eq:derive}) to express this lower bound as
$$ w(j)\chi \cdot \sum_{i \in B(j, r_j) \setminus \{i'\}} w(i) d(i, j) + w(j)d(i', j) \cdot \left( w_{j \to i'} - \chi \cdot \sum_{i \in B(j, r_j) \setminus \{i' \}} w(i)\right) $$
$$ = w(j)\chi \cdot \sum_{i \in B(j, r_j) \setminus \{i'\}} w(i) \! \left(d(i, j) - d(i', j) \right) + w_{j \to i'} \cdot w(j)d(i', j). $$
We can now combine (\ref{eq:derive2}) with this lower bound to get that
$$ \lambda \xi \geq \lambda w(j)\chi +w(j)\chi \cdot \sum_{i \in B(j, r_j) \setminus \{i'\}} w(i) \! \left(d(i, j) - d(i', j) \right) $$
$$ = w(j)\chi \cdot \left( \lambda - \sum_{i \in B(j, r_j) \setminus \{i'\}} w(i) \! \left(d(i', j) - d(i, j) \right) \right)$$
$$ \geq w(j)\chi \cdot \left( \lambda - \sum_{i \in B(j, r_j) \setminus \{i'\}} \max \left(0, w(i) \! \left(d(i', j) - d(i, j) \right) \right) \right) \geq 0. $$
To see why the last inequality holds, consider the following argument. Suppose that we are growing a ball around $j$ in the metric space $(V', d, w)$ and we have reached radius $d(i',j)$. Since $d(i',j) \leq r_j \leq r'_j$, it must be the case that the total contribution of all the points in this ball, towards the facility opening cost of $j$, is at most $\lambda$.
\end{proof}

\subsection{The Data Structure \FracLMP}

Given a dynamic metric space $(V,d, w)$ which evolves via a sequence of point insertions and deletions, the data structure $\FracLMP$ supports the following query operations.
\begin{itemize}
    \item \textsc{Solution}$(\lambda)$: The input to this query is a facility opening cost $\lambda > 0$. In response, the data structure outputs the values $\{y_i\}$ of the solution produced by running \textsc{FractionalMP} with facility opening cost $\lambda$.
    \item \textsc{Connection-Cost}$(\lambda)$: The input to this query is a facility opening cost $\lambda > 0$. In response, the data structure outputs the connection cost $\sum_{i,j \in V} w(j)d(i,j) \cdot x_{j \to i}$ of the solution produced by running \textsc{FractionalMP} with facility opening cost $\lambda$.
\end{itemize}
The following lemma summarises the behaviour of this data structure.

\begin{lemma}\label{lem:fracLMP}
    The data structure $\FracLMP$ is deterministic, has a worst-case update time of $O(n \log n)$, and can answer \textsc{Solution} and \textsc{Connection-Cost} queries in $O(n\log n)$ time.
\end{lemma}

\noindent
We now show how to implement $\FracLMP$ to get the guarantees described in \Cref{lem:fracLMP}.

\subsection{Implementation of $\FracLMP$}

The data structure $\FracLMP$ maintains the data structure $\RadiiMP$ with parameter $\beta = 1/4$, which it uses as a black box to answer queries. Since $\RadiiMP$ is deterministic and has a worst-case update time of $O(n \log n)$, this gives the desired update time. We now show how to implement the \textsc{Solution} and \textsc{Connection-Cost} queries in $O(n\log n)$ time.

\medskip
\noindent \textbf{Implementing \textnormal{\textsc{Solution}}.} Whenever a query $\textsc{Solution}(\lambda)$ is made, we use $\RadiiMP$ to obtain the values $\{r^{(1/4)}_i\}$ in $O(n \log n)$ time. Since $y_i = (1/\lambda) \cdot w(i) r^{(1/4)}_i$, we can then compute the values of $\{y_i\}$ in $O(n)$ time. In total, this takes $O(n \log n)$ time.

\medskip
\noindent \textbf{Implementing \textnormal{\textsc{Connection-Cost}}.} Whenever a query $\textsc{Connection-Cost}(\lambda)$ is made, we use $\RadiiMP$ to obtain the values $\{C^{(1/4)}_i\}$ in $O(n \log n)$ time. Since $\sum_{i \in V} w(j)d(i,j) \cdot x_{j \to i} = 4 C^{(1/4)}_j$, it follows that
$$\sum_{i,j \in V} w(j)d(i,j) \cdot x_{j \to i} = 4 \sum_{j \in V} C^{(1/4)}_j.$$
Hence, we can compute the value of $\sum_{i,j \in V} w(j)d(i,j) \cdot x_{j \to i}$ in $O(n)$ time. In total, this takes $O(n \log n)$ time.

\section{Dynamic Fractional $k$-Median}\label{sec:frac k med}

In this section, we present the dynamic data structure $\FracKMed$. The properties of this data structure immediately imply the following theorem.

\begin{theorem}\label{thm:dynamic frac k-med}
    There exists a deterministic dynamic algorithm that has a worst-case update time of $O(n \log n)$ and, given any $\epsilon > 0$, $k \in \mathbb N$, returns a $4(1 + \epsilon)$-approximate fractional solution for $k$-median in $O(n \log n \log\log_{1 + \epsilon}(n \Delta))$ time.
\end{theorem}

\noindent
Subsequently, in \Cref{sec:dynamic value k-med}, we prove the following corollary of \Cref{thm:dynamic frac k-med}.

\begin{corollary}\label{cor: dynamic value k-med}
    There exists a deterministic dynamic algorithm that has a worst-case update time of $O(n \log n)$ and, given any $\epsilon > 0$, $k \in \mathbb N$, returns a $12(1 + \epsilon)$-approximation to the cost of the optimal integral solution for $k$-median in $O(n \log n \log\log_{1 + \epsilon}(n \Delta))$ time.
\end{corollary}

\subsection{The Data Structure \FracKMed}

Given a dynamic metric space $(V,w,d)$ which evolves via a sequence of point insertions and deletions, the data structure $\FracKMed$ supports the following query operation.
\begin{itemize}
    \item \textsc{Solution}$(\epsilon, k)$: The inputs to this query are parameters $\epsilon \in (0,1)$, $k \in \mathbb N$. In response, the data structure computes a $4(1+\epsilon)$-approximate fractional solution $\{y_i, x_{j \to i}\}$ to the $k$-median problem and returns the variables $\{y_i\}$ and the connection cost $\sum_{i,j \in V} w(j) d(i,j) \cdot x_{j \to i}$.    
\end{itemize}
The following lemma summarises the behaviour of this data structure.

\begin{lemma}\label{lem:fracKMed}
    The data structure $\FracKMed$ is deterministic, has a worst-case update time of $O(n \log n)$, and can answer \textsc{Solution} queries in $O(n \log n \log\log_{1 + \epsilon}(n \Delta))$ time.
\end{lemma}



\noindent
The data structure $\FracKMed$ maintains the data structure $\FracLMP$, which it uses as a black box to answer queries. 
We now show how to implement the \textsc{Solution} query. In \Cref{sec:frac k med algo}, we give the algorithm that the data structure uses to construct the solution whenever a query is made. In \Cref{sec:frac k med anal}, we analyse this algorithm and show that it computes the desired output. Finally, in \Cref{sec:frackmed imp}, we show how to implement $\FracKMed$.

\subsection{Constructing the Fractional Solution}\label{sec:frac k med algo}

The data structure $\FracKMed$ computes the fractional solution $\{y_i, x_{j \to i}\}$ as follows.

\medskip
\noindent \textbf{Step I:} If $k = 1$, skip directly to step IV. Otherwise, let $L := d_{\min}w_{\min}/8$, $H := 2n \cdot d_{\max}w_{\max}$ and $\lambda_t := (1 + \epsilon)^{t-1} \cdot L$. Consider the $T = O(\log_{1 + \epsilon}(n \Delta))$ many values $\lambda_1 < \dots < \lambda_T$, where $H \leq \lambda_T < (1 + \epsilon) H$.

\medskip
\noindent \textbf{Step II:} There must necessarily exist an index $t \in [T-1]$ such that $\FracLMP$ fractionally opens at least $k$ facilities when run with facility opening cost $\lambda_t$ and less than $k$ when run with facility opening cost $\lambda_{t+1}$.
We can perform a binary search on the values $\lambda_1, \dots, \lambda_T$ in order to find such a value of $t$, making $O(\log \log_{1 + \epsilon}(n \Delta))$ many queries to $\textsc{FracLMP}$.
Let $\{y_i^{1}\}$, $C^1$ and  $\{y_i^{2}\}$, $C^2$ denote the fractional solutions and connection costs returned by $\FracLMP$ with facility opening costs $\lambda_{t}$ and $\lambda_{t+1}$ respectively.

\medskip
\noindent \textbf{Step III:} Define real numbers $\alpha_1, \alpha_2 \geq 0$ such that $\alpha_1 \cdot \sum_{i \in V}y^{1}_i + \alpha_2 \cdot \sum_{i \in V}y^{2}_i = k$ and $\alpha_1 + \alpha_2 = 1$. For all $i \in V$, let $y_{i}  :=  \alpha_1 \cdot y_{i}^{1} + \alpha_{2} \cdot y_{i}^{2}$. Return the solution $\{y_i\}$ and its cost $\alpha_1 C^1 + \alpha_2 C^2$.

\medskip
\noindent \textbf{Step IV:} Query $\FracLMP$ with facility opening cost $n \cdot d_{\max}w_{\max}/4$ to obtain a solution $\{y_i^\star, x_{j \to i}^\star\}$. For all $i \in V$, let $y_{i}  :=  (\sum_{j \in V} y^\star_j)^{-1} \cdot y^\star_i$. Return the solution $\{y_i\}$ and its cost $\sum_{i,j \in V} w(j) d(i,j) \cdot y_i$.


\subsection{Analysis of \FracKMed}\label{sec:frac k med anal}

We begin with the case that $k > 1$. We first show that there must necessarily exist an index $t \in [T-1]$ such that $\FracLMP$ fractionally opens at least $k$ facilities when run with facility opening cost $\lambda_t$ and less than $k$ when run with facility opening cost $\lambda_{t+1}$.

\begin{observation}\label{obs:discrete lambda 1}
    Running $\FracLMP$ with $\lambda \leq d_{\min}w_{\min}/8$ fully opens every facility.
\end{observation}

\begin{proof}
    We first note that $\foptf_\lambda \leq \lambda n$, since it is at most the cost of the solution that fully opens each facility. Now suppose that running $\FracLMP$ with this value of $\lambda$ returns a solution that fractionally opens $\ell \in [1, n]$ facilities.
    Since the algorithm $\FracLMP$ is LMP $4$-approximate, it follows that
    $$ 4 \cdot \lambda n \geq 4 \cdot \foptf_\lambda \geq 4\cdot \lambda \ell +  (n - \ell) \cdot d_{\min}w_{\min} \geq 4 \cdot \lambda (\ell + 2(n - \ell)).$$
    It follows that $n \geq \ell + 2(n - \ell)$ and so $\ell \geq n$. Hence, the solution returned by $\FracLMP$ must fully open every facility.
\end{proof}

\begin{observation}\label{obs:discrete lambda 2}
    Running $\FracLMP$ with $\lambda \geq n \cdot d_{\max}w_{\max}/\delta$ factionally opens at most $(1 + \delta)$ facilities.
\end{observation}

\begin{proof}
    We first note that $\foptf_\lambda \leq \lambda + n \cdot d_{\max}w_{\max}$, since it is at most the cost of a solution that fully opens $1$ facility. Now suppose that running $\FracLMP$ with this value of $\lambda$ returns a solution that fractionally opens $\ell \in [1, n]$ facilities. Then, by the same arguments as \Cref{obs:discrete lambda 1}, we have that
    $$ 4 \cdot \lambda(1 + \delta) \geq 4 \cdot (\lambda + n \cdot d_{\max}w_{\max}) \geq \foptf_\lambda \geq 4 \cdot \lambda \ell +  (n - \ell) \cdot d_{\min}w_{\min} \geq 4 \cdot \lambda \ell. $$
    It follows that $\ell \leq (1 + \delta)$. Hence, the solution returned by $\FracLMP$ fractionally opens at most $(1 + \delta)$ facilities.
\end{proof}

\noindent
It follows from \Cref{obs:discrete lambda 1} and \Cref{obs:discrete lambda 2} with $\delta = 1/2$ that $\FracLMP$ opens less than $k$ facilities for facility opening cost $\lambda_1$ and at least $k$ facilities for facility opening cost $\lambda_T$. So there must necessarily exist an index $t \in [T-1]$ such that $\FracLMP$ fractionally opens at least $k$ facilities when run with facility opening cost $\lambda_t$ and less than $k$ when run with facility opening cost $\lambda_{t+1}$.

Let $\{v^{1}_j, w_{j \to i}^{1}\}$ and $\{v^{2}_j, w_{j \to i}^{2}\}$ denote optimal dual solution for UFL with facility opening costs $\lambda_t$ and $\lambda_{t+1}$ respectively. Let $\lambda := \lambda_t$ and recall that $\lambda_{t+1} = (1 + \epsilon)\lambda$. We have that
\begin{eqnarray}
\label{eq:duality:r:1}
4  \lambda \cdot \sum_{i \in V}y^1_i + \sum_{i, j \in V} w(j)d(i, j) \cdot x^{1}_{j \to i} & \leq & 4 \cdot \sum_{j \in V} v^{1}_j \\
\label{eq:duality:r:2}
4  (1+\epsilon)\lambda \cdot\sum_{i \in V}y^2_i + \sum_{i, j \in V} w(j)d(i, j) \cdot x^{2}_{j \to i} & \leq & 4 \cdot \sum_{j \in V} v^{2}_j
\end{eqnarray}
Define $\hat{v}_j^{2} := (1+\epsilon)^{-1} \cdot v_j^{2}$ and $\hat{w}_{j \to i}^{2} := (1+\epsilon)^{-1} \cdot w_{j \to i}^{2}$ for all $i,j \in V$. It follows from \Cref{lem:convex} that $\{\hat{v}_j^{2}, \hat w^{2}_{j \to i}\}$ is a dual solution for UFL with facility opening cost $\lambda$.



We now rearrange terms in~(\ref{eq:duality:r:1}) and~(\ref{eq:duality:r:2}) to obtain
\begin{eqnarray}
\label{eq:duality:r:3}
4(1 + \epsilon) \lambda \cdot \sum_{i \in V}y^1_i + \sum_{i, j \in V} w(j) d(i, j) \cdot x^{1}_{j \to i} & \leq & 4(1 + \epsilon) \cdot \sum_{j \in V} v^{1}_j  \\
\label{eq:duality:r:4}
4(1+\epsilon) \lambda \cdot \sum_{i \in V}y^2_i + \sum_{i, j \in V} w(j) d(i, j) \cdot x^{2}_{j \to i} & \leq & 4(1+\epsilon) \cdot \sum_{j \in V} \hat{v}^{2}_j
\end{eqnarray}
We now take convex combinations of the other variables in the same way as $\{y^1_i\}$ and $\{y^2_i\}$. For all $i,j \in V$ let
$$\tilde x_{j \to i} := \alpha_1 \cdot x_{j \to i}^{1} + \alpha_{2} \cdot x_{j \to i}^{2}, \quad \tilde v_{i}  :=  \alpha_1 \cdot v_{i}^{1} + \alpha_{2} \cdot \hat v_{i}^{2}, \quad \textrm{ and } \quad \tilde w_{j \to i} := \alpha_1 \cdot w_{j \to i}^{1} + \alpha_{2} \cdot \hat w_{j \to i}^{2}.$$
We now infer that
\begin{equation}
\label{eq:convex:3}
4(1 + \epsilon) \lambda \cdot \sum_{i \in V}y_i + \sum_{i, j \in V} w(j) d(i, j) \cdot x_{j \to i} \leq 4(1+\epsilon) \cdot \sum_{j \in V} \tilde v_j.
\end{equation}
It follows from \Cref{lem:convex} that $\{\tilde v_j, \tilde w_{j \to i}\}$ is a dual solution for UFL with facility opening cost $\lambda$.
Since $k = \sum_{i \in V} y_i$, it follows that $\{y_i, x_{j \to i}\}$ is an LMP $4(1 + \epsilon)$-approximation for UFL. Hence, by \Cref{lem:simple LMP case}, $\{y_i, x_{j \to i}\}$ is a $4(1 + \epsilon)$-approximate solution to $k$-median.

Finally, note that the cost of this fractional $k$-median solution is $\sum_{i, j \in V} w(j)d(i, j) \cdot x_{j \to i} = \alpha_1 \cdot \sum_{i, j \in V} w(j)d(i, j) \cdot x^1_{j \to i} + \alpha_2 \cdot \sum_{i, j \in V} w(j)d(i, j) \cdot x^2_{j \to i}$ as required.

\medskip
\noindent \textbf{Case $k=1$.} Now suppose that $k = 1$. Then we can turn the fractional solution $\{y_i^\star, x_{j \to i}^\star\}$ obtained by running $\FracLMP$ with facility opening cost $n \cdot d_{\max}w_{\max} / 4$ into an LMP $4$-approximate solution that opens exactly 1 facility.
Recall that $y_i = (\sum_{j \in V} y^\star_j)^{-1} \cdot y^\star_i$ for all $i \in V$. Clearly, $\sum_{i \in V} y_i = 1$. This defines a unique assignment of points to facilities $\{x_{j \to i}\}$, where $x_{j \to i} = y_i$. Let $\delta := \sum_{i \in V} y^\star_{i} - 1$ and note that, for each $j \in V$, $\sum_{i \in V} w(j)d(i,j) \cdot x_{j \to i} \leq \sum_{i \in V} w(j)d(i,j) \cdot x^\star_{j \to i} + \delta \cdot d_{\max} w_{\max}$. Hence, we get that
$$ 4 \lambda \cdot \sum_{i \in V} y_i + \sum_{i,j \in V} w(j)d(i,j) \cdot x_{j \to i} \leq 4 \lambda \cdot \left( \sum_{i \in V} y^\star_i - \delta \right) + \sum_{i,j \in V} w(j)d(i,j) \cdot x^\star_{j \to i} + \delta n \cdot d_{\max} w_{\max} $$
$$  = 4 \lambda \cdot  \sum_{i \in V} y^\star_i + \sum_{i,j \in V} w(j)d(i,j) \cdot x^\star_{j \to i} + \delta \cdot \left( n \cdot d_{\max} w_{\max} - 4 \lambda \right) 
 \leq 4 \foptf_\lambda.$$
It follows that $\{y_i, x_{j \to i}\}$ is a fractional LMP $4$-approximate solution for UFL that opens exactly 1 facility. Applying \Cref{lem:simple LMP case}, we get that $\{y_i, x_{j \to i}\}$ is an $4$-approximate fractional solution to the $k$-median problem. The connection cost of this fractional $k$-median solution is $\sum_{i,j \in V} w(j) d(i,j) \cdot y_i$.

\subsection{Implementation of \FracKMed}\label{sec:frackmed imp}

We now show how to implement $\FracKMed$ so that it is deterministic, and $O(n \log n)$ worst-case update time, and $O(n\log n \log \log_{1 + \epsilon}(n \Delta))$ query time. For now, we assume knowledge of $d_{\min}$, $d_{\max}$, $w_{\min}$, and $w_{\max}$.

\medskip
\noindent \textbf{Update Time.} The data structure $\FracKMed$ maintains the data structure $\FracLMP$ as well as the value of $\sum_{j \in V} w(j)d(i,j)$ for all $i \in V$. The data structure $\FracLMP$ is deterministic and has a worst-case update time of $O(n \log n)$. We can store the values $\{\sum_{j \in V} w(j)d(i,j)\}$ in a balanced binary tree, so they can be accessed in $O(\log n)$ time and updated in $O(n \log n)$ time after an update. Hence, the update time of $\FracKMed$ is $O(n \log n)$.

\medskip
\noindent \textbf{Query Time.}
We can perform the binary search in $O(\log \log_{1 + \epsilon}(n \Delta))$ steps. Each step of the binary search requires us to make one query to $\FracLMP$, which takes $O(n \log n)$ time, and then compute how many facilities are opened by the solution returned by the query, which can be done in $O(n)$ time. Hence, this takes $O(n\log n \log \log_{1 + \epsilon}(n \Delta))$ time in total. After completing the binary search, we can obtain $\{y^1_i\}$, $C^1$, $\{y^2_i\}$, and $C^2$ in $O(n \log n)$ time by making $O(1)$ many queries to $\FracLMP$. We can then compute $\alpha_1$ and $\alpha_2$ and construct the output in $O(n)$ time.

In step IV, we can compute the solution $\{y_i\}$ in $O(n \log n)$ time by making one query to $\FracLMP$ and scaling the output. Using the binary tree that we maintain, we can then find $y_i \cdot \sum_{j \in V} w(j)d(i,j)$ in $O(\log n)$ time for each $i \in V$. Hence, we can compute the connection cost of the unique solution defined by the variables $\{y_i\}$ in $O(n \log n)$ time.

\medskip
\noindent \textbf{Maintaining $d_{\min}$, $d_{\max}$, $w_{\min}$, and $w_{\max}$.}
In order to perform the binary search, we need to know the values of $d_{\min}$, $d_{\max}$, $w_{\min}$, and $w_{\max}$. It is easy to deterministically maintain these values with $O(n \log n)$ update time. In particular, we maintain balanced binary trees of all the (non-zero) distances $\{d(i,j)\}_{i \neq j}$ and all of the weights $\{w(i)\}$. Every time a point $i$ is inserted (resp. deleted) into $V$, we add (resp. remove) all the distances from $i$ to the other points in $V$ to the binary tree storing the non-zero distances. This takes $O(n \log n)$ time in total. The binary tree storing the weights can be maintained in an analogous way in only $O(\log n)$ time.
Using these data structures, we can then obtain all of these values for the current metric space $V$ in $O(\log n)$ time.

\subsection{Proof of \Cref{cor: dynamic value k-med}}\label{sec:dynamic value k-med}

Suppose we query the data structure $\FracKMed$ and it returns a connection cost of $C$. Since this is the cost of a $4(1 + \epsilon)$-approximate fractional solution to the $k$-median problem, we know that $\foptc_k \leq C \leq 4(1 + \epsilon)\foptc_k$. Since the LP-relaxation of $k$-median has an integrality gap of at most $3$, we know that $\foptc_k \leq \optc_k \leq 3\foptc_k$, and hence $\optc_k \leq 3C \leq 12(1 + \epsilon)\optc_k$. It follows that $\FracKMed$ can also be used to maintain a $12(1 + \epsilon)$-approximation to the cost of the optimal integral solution to the $k$-median problem.

\section{Dynamic Sparsification  (Take II)}\label{sec:frac sparsifier}

In this section, we show how to use the dynamic algorithm of \cite{ourneurips2023} in order to obtain the following theorem.

\begin{theorem}\label{thm:dynamic frac k-med sparse}
    There exists a randomized dynamic algorithm that, given any $k \in \mathbb N$, maintains a $O(1)$-approximation to the cost of the optimal integral solution for $k$-median with $\tilde O(k)$ amortized update time.
\end{theorem}

\noindent Unlike \Cref{app:sparsification}, this result does not follow immediately from applying the algorithm of \cite{ourneurips2023} in a black box manner. This is because the sparsifier does not preserve the cost of solutions.\footnote{Note that the sparsifier is a bicriteria approximation, but is not a coreset.} Thus, some extra work is required to efficiently approximate the cost of the fractional solution maintained by the algorithm in the original unsparsified space.

We now show how to use \Cref{cor: dynamic value k-med} and the sparsifier of \cite{ourneurips2023} in a white box manner in order to prove \Cref{thm:dynamic frac k-med sparse}. For simplicity, we focus on the case that the input space is unweighted. Using the ideas from \Cref{sec:weighted sparsifier}, this can easily be extended to work for weighted metric spaces.

\subsection{The Cost of The Sparsifier}

Recall that, given a dynamic metric space $(V,d)$ which evolves via a sequence of point insertions and deletions and a parameter $k \in \mathbb N$, the algorithm $\Sparsifier$ (see \Cref{app:sparsification}) explicitly maintains the following objects with an amortized update time of $\tilde O(k)$.
\begin{itemize}
    \item A set $V' \subseteq V$ of at most $\tilde O(k)$ points.
    \item An assignment $\sigma : V \longrightarrow V'$ such that $\cl(\sigma, V) \leq O(1) \cdot\optc_k(V)$ with probability at least $1 - O(1/n^c)$.\footnote{Here, $\cl(\sigma, V) = \sum_{j \in V} d(j, \sigma(j))$.}
\end{itemize}

\noindent
In order to approximate the cost of our solution in the space $V$, it will be necessary for us to be able to approximate the cost of the assignment $\sigma$, i.e.~the value of $\cl(\sigma, V)$. It turns out that we can maintain an approximation $\gamma$ of $\cl(\sigma, V)$ by appropriately modifying the sparsifier. In particular, we have the following lemma.

\begin{lemma}
    We can modify $\Sparsifier$ so that it explicitly maintains a value $\gamma$ such that $\cl(\sigma, V) \leq \gamma \leq O(1) \cdot\optc_k(V)$ with probability at least $1 - O(1/n^c)$.
\end{lemma}

\begin{proof}
    Using the notation from \cite{ourneurips2023}, their algorithm maintains a partition of the space $V$ into $t = \tilde O(1)$ many sets $\{C_i\}_{i \in [t]}$ and a collection of values $\{\nu_i\}_{i \in [t]}$ such that
    $$ \cl(\sigma, V) \leq \sum_{i = 1}^t 2 \nu_i |C_i| \leq O(1) \cdot \optc_k(V) $$
    with high probability (see their Lemmas B.6 and B.10). The sets $C_i$ correspond to the points removed from each of the ``layers'' in their data structure, and can easily be maintained explicitly with $\tilde O(1)$ overhead. Similarly, the value $\nu_i$ is computed each time the $i^{th}$ layer is reconstructed, and can easily be maintained explicitly. Thus, we can modify the algorithm to explicitly maintain the value of $\gamma := \sum_{i = 1}^t 2 \nu_i |C_i|$ with only $\tilde O(1)$ overhead in the running time.
\end{proof}

\subsection{Proof of \Cref{thm:dynamic frac k-med sparse}}

As described in \Cref{app:sparsification}, we feed the metric space $(V, d)$ to the sparsifier, and then run the algorithm from \Cref{cor: dynamic value k-med} on the sparsified metric space $(V', w',d)$ maintained by the sparsifier. Let $\{x_{j \to i}', y_i'\}$ denote the fractional solution to the $k$-median problem on $(V',w',d)$ maintained by our dynamic algorithm. Note that our dynamic algorithm is capable of explicitly maintaining the connection cost $\sum_{i,j \in V'} w'(j)d(i,j) \cdot x_{j \to i}'$.
We now show that 
\begin{equation}\label{eq:fracking}
\gamma + \sum_{i,j \in V'} w'(j)d(i,j) \cdot x_{j \to i}' = \Theta(1) \cdot \optc_k(V).
\end{equation}
Since we can maintain the LHS of Equation~(\ref{eq:fracking}) explicitly with $\tilde O(k)$ update time, the theorem follows.

Define a solution $\{x_{j \to i}, y_i\}$ in the space $(V,d)$ be setting $y_i = y_i'$ for all $i \in V'$, $x_{j \to i} = x_{\sigma(j) \to i}'$ for all $i \in V', j \in V$, and setting all over variables to $0$.
Now consider the following lemma.

\begin{lemma}\label{lem:frac GUHA}
    The solution $\{y_i, x_{j \to i}\}$ is a feasible and $\sum_{i,j \in V'} w'(j)d(i,j) \cdot x_{j \to i}' \leq O(1) \cdot \optc_k(V)$.
\end{lemma}

\begin{proof} To Do.
    It is clear that the variables $\{y_i, x_{j \to i}\}$ satisfy (\ref{eq:kmed:lp:4:r}). To see that the other constraints are also satisfied,
    note that $\{y_i', x_{j \to i}'\}$ is a feasible solution to the $k$-median problem for $(V',w')$. It follows that $\sum_{i \in V} y_i = \sum_{i \in V'} y'_i \leq k$, $\sum_{i \in V} x_{j \to i} = \sum_{i \in V'} x'_{\sigma(j) \to i} \geq 1$ for all $j \in V$, and $x_{j \to i} = x'_{\sigma(j) \to i} \leq y'_i = y_i$ for all $j \in V, i \in V'$,\footnote{Note that $x_{j \to i} = 0 = y_i$ when $i \notin V'$.} establishing (\ref{eq:kmed:lp:1:r}), (\ref{eq:kmed:lp:2:r}), and (\ref{eq:kmed:lp:3:r}) respectively.
    
    For the second part of the lemma, let $w'(i) := |\sigma^{-1}(i)|$ for all $i \in V'$. Now, let $S^\star$ be an optimal integral solution to the $k$-median problem on $(V,d)$. Then there exists some set $S' \subseteq V'$ of size at most $k$ such that $\cl(S', V', w') \leq 2\cl(S^\star, V', w')$. It follows that
    $$\sum_{i, j, \in V'} w'(j)d(i,j) \cdot x'_{j \to i} \leq O(1) \cdot \foptc_k(V',w') \leq O(1) \cdot \sum_{i \in V'} w'(i)d(i, S') \leq O(1) \cdot \sum_{i \in V'} w'(i)d(i, S^\star) $$
    $$ = O(1) \cdot \sum_{j \in V} d(\sigma(j), S^\star) \leq O(1) \cdot \sum_{j \in V} d(j, \sigma(j)) + O(1) \cdot \sum_{j \in V} d(\sigma(j), S^\star) \leq O(1) \cdot \optc_k(V). $$    
\end{proof}

\noindent
It follows that
$$ \Theta(1) \cdot \optc_k(V) \leq \foptc_k(V) \leq  \sum_{i,j \in V} d(i,j) \cdot x_{j \to i} \leq \cl(\sigma, V) + \sum_{i,j \in V} w(j)d(i,j) \cdot x_{j \to i}'$$
$$ \gamma + \sum_{i,j \in V} w(j)d(i,j) \cdot x_{j \to i}' \leq O(1) \cdot \optc_k(V). $$

\newpage

\bibliography{bibl.bib}
\bibliographystyle{alpha}

\end{document}